\documentclass[final,11pt,amsfonts,superscriptaddress,nofootinbib,longbibliography,notitlepage]{article}
\usepackage[margin=1in]{geometry}
\usepackage[numbers]{natbib}
\usepackage{physics}
\usepackage{authblk}  
\usepackage{float}

\usepackage{graphicx}
\usepackage[dvipsnames]{xcolor}
\usepackage{rotating}
\usepackage{amsmath,amssymb,graphics,amsthm,isomath}
    \usepackage{amsfonts,dsfont,mathtools}
    \usepackage{bbm}
\usepackage{bm}
\usepackage{array}
\usepackage{appendix}
\newcolumntype{P}[1]{>{\centering\arraybackslash}p{#1}}
\usepackage{multirow}
\usepackage{physics}
\usepackage{braket}
\usepackage{adjustbox}
\usepackage{stmaryrd}
\usepackage{blkarray, bigstrut}
\usepackage{tikz-cd}
\usepackage[colorlinks=true, urlcolor=violet, linkcolor=blue, citecolor=red, hyperindex=true, linktocpage=true]{hyperref}
\usepackage[capitalise,compress]{cleveref}

\allowdisplaybreaks

\newtheorem{thm}{Theorem}
\numberwithin{thm}{section}
\newtheorem{cor}[thm]{Corollary}
\newtheorem{lem}[thm]{Lemma}

\newtheorem{prop}[thm]{Proposition}

\newtheorem{defn}[thm]{Definition}

\newtheorem{conj}{Conjecture}[section]
\newtheorem{openprob}{Open Problem}[]

\makeatletter
\renewcommand{\p@subsection}{}
\renewcommand{\p@subsubsection}{}
\makeatother

\usepackage{xcolor}
\usepackage{mathtools}

\newcommand\bea{\begin{eqnarray}}
\newcommand\eea{\end{eqnarray}}
\newcommand\be{\begin{equation}}
\newcommand\ee{\end{equation}}
\newcommand\bes{\begin{subequations}}
\newcommand\ees{\end{subequations}}
\newcommand\bed{\begin{displaymath}}
\newcommand\eed{\end{displaymath}}
\newcommand\beal{\begin{aligned}}
\newcommand\eeal{\end{aligned}}
\newcommand\bew{\begin{widetext}}
\newcommand\eew{\end{widetext}}
\newcommand\beit{\begin{itemize}}
\newcommand\eeit{\end{itemize}}
\def\bea{\begin{array}}
\def\eea{\end{array}}
\newcommand\been{\begin{enumerate}}
\newcommand\eeen{\end{enumerate}}

\newcommand{\ident}[0]{\mathds{1}}

\newcommand{\poly}{\mathsf{poly}}

\usepackage{verbatim}

\newcommand{\im}{\operatorname{im}}
\newcommand{\rs}{\operatorname{rs}}

\newcommand{\Aut}{\operatorname{Aut}}

\newcommand{\transpose}[0]{\mathsf{T}}

\usepackage{tikz}

\definecolor{red1}{rgb}{0.76, 0.23, 0.13}
\definecolor{green1}{rgb}{0.0, 0.5, 0.0}
\usepackage{pifont} 
\newcommand{\cmark}{\textcolor{green1}{\ding{51}}}
\newcommand{\xmark}{\textcolor{red1}{\ding{55}}}

\usepackage{booktabs}

\begin{document}

\title{Automorphism gadgets in homological product codes}

\author[1]{Noah Berthusen}
\author[1]{Michael J. Gullans}
\author[1,2]{Yifan Hong\footnote{\href{mailto:yhong137@umd.edu}{yhong137@umd.edu}}}
\author[1]{\\Maryam Mudassar}
\author[1]{Shi Jie Samuel Tan\footnote{\href{mailto:stan97@umd.edu}{stan97@umd.edu}}}
\affil[1]{\small Joint Center for Quantum Information and Computer Science,

University of Maryland and NIST, College Park, MD 20742, USA}
\affil[2]{\small Joint Quantum Institute, University of Maryland and NIST, College Park, MD 20742, USA}

\date{\today}

\maketitle
\renewcommand{\thefootnote}{\fnsymbol{footnote}}
\setcounter{footnote}{5}
\footnotetext{{The authors are listed in alphabetical order.}}
\setcounter{footnote}{0}
\renewcommand{\thefootnote}{\arabic{footnote}}
\begin{abstract}
    The homological product is a general-purpose recipe that forges new quantum codes from arbitrary classical or quantum input codes, often providing enhanced error-correcting properties. When the input codes are classical linear codes, it is also known as the hypergraph product. We investigate structured homological product codes that admit logical operations arising from permutation symmetries in their input codes. We present a broad theoretical framework that characterizes the logical operations resulting from these underlying automorphisms. In general, these logical operations can be performed by a combination of physical qubit permutations and a subsystem circuit. In special cases related to symmetries of the input Tanner graphs, logical operations can be performed solely through qubit permutations. We further demonstrate that these ``automorphism gadgets'' can possess inherent fault-tolerant properties such as effective distance preservation, assuming physical permutations are free. Finally, we survey the literature of classical linear codes with rich automorphism structures and show how various classical code families fit into our framework. Complementary to other fault-tolerant gadgets for homological product codes, our results further advance the search for practical fault tolerance beyond topological codes in platforms capable of long-range connectivity.
\end{abstract}

\newpage
\tableofcontents

\section{Introduction}

Quantum error correction, and more generally fault-tolerant quantum computation, promises reliable large-scale computations at the cost of additional resources \cite{Knill_1998, Aharonov_2008, Kitaev_2003}. Since the advent of topological codes \cite{Kitaev_2006, Bombin_2006}, there has been a significant push to lower the resource overhead for error correction, culminating in the discovery of high-rate quantum low-density parity-check (qLDPC) codes \cite{Breuckmann_2021, Bravyi_2024}. These codes sacrifice geometric locality in order to more efficiently encode their logical qubits compared to topological codes. It was then shown by Gottesman that constant-rate qLDPC codes can enable fault-tolerant quantum computation with asymptotically constant spatial overhead \cite{Gottesman_2014} but linear time overhead, compared to the unencoded circuit. The time overhead has since been improved down to polylogarithmic \cite{nguyen2024FT, tamiya2024FT}.

Although the above results are encouraging from a theoretical standpoint, in practice there are often a lot of large factors hiding in the constants which may make near-term implementations difficult. On this front, there has been a recent surge of interest in constructing explicit low-overhead gadgets applicable to large families of qLDPC codes, going by the names of bridges \cite{cross2024}, adapters \cite{swaroop2025} and extractors \cite{he2025extr} as well as homomorphic measurements \cite{xu2024fast}. These gadgets can measure arbitrary combinations of logical Pauli operators and so are well-suited for the Pauli-based computational model \cite{Litinski_2019}. In this work, we will explore logical gates implemented by unitary circuits such as qubit permutations and transversal gates. A particularly alluring type of gate is an automorphism gate, in which logical operations are performed by simply permuting the physical qubits.

Automorphisms are a reordering of the bits (or qubits) that maps a code to itself. The group arising from the composition of classical error correcting code automorphisms has been a topic of intense study, both as a subject of intrinsic mathematical interest and for applications such as improved decoding~\cite{macwilliams_1964}.
It was later shown that automorphisms could enable low-overhead logical gates in quantum error correcting codes~\cite{calderbank_1997, Grassl_2013} as well as assist in their decoding~\cite{koutsioumpas2025, pacenti2025}.
Unfortunately, logical gates that can be implemented in this way are limited, see Section~\ref{sec:bounds}; nonetheless, several recent works have shown that automorphism gates can provide computational speed-ups~\cite{xu2024fast, malcolm2025} and increased logical gate fidelity~\cite{hong2024, burton2024, reichardt2024}, motivating a deeper understanding of how to obtain quantum codes with useful automorphism groups.

Given a stabilizer code, it is not often obvious how to compute the available symmetries and the resulting logical gates. For classical linear codes, it is known that computing the automorphism group is intimately related to the \textsc{Permutation Code Equivalence} problem that can be reduced to the \textsc{Graph Isomorphism} problem~\cite{Petrank_1997}. Sayginel et al.~\cite{sayginel2025} made progress on this problem for stabilizer codes by providing a method for finding fault-tolerant logical Clifford gates arising from code automorphisms. 
The method presented there is general and applies to all stabilizer codes. In this work, we focus on Calderbank-Shor-Steane (CSS) stabilizer codes constructed from the product of classical or quantum codes, of which the hypergraph product~\cite{Tillich_2014}, homological product~\cite{Freedman_2014, Bravyi_2014_HGP, Zeng_2019, Campbell_2019}, lifted product~\cite{Panteleev_2021, panteleev2021quantum}, and balanced product \cite{Breuckmann_2021_BP} are examples. Intuitively, since these codes are ``products'' of underlying codes, one might expect them to inherit the automorphism symmetries of their input classical codes.
For the cases of hypergraph and homological products, we show that this intuition is nearly correct, and we provide a systematic method of obtaining a quantum error-correcting code with a desired automorphism group. In particular, we show that the automorphism groups of the underlying classical codes can be utilized by the resulting quantum code to yield a set of logical gates which can be implementing on the physical level through a combination of qubit permutations and a (not necessarily shallow) circuit on a subsystem of the data qubits. Because of the possible existence of these circuits, we refer to these logical gates as \textit{automorphism gadgets}, as opposed to \textit{automorphism gates}, which consist solely of qubit permutations. To this end, we show that a stricter notion of code automorphism, known as Tanner graph automorphism, does give rise to true automorphism gates in the corresponding product codes.
It might be of interest to some readers that Lockhart and Gonz\'{a}lez-Guill\'{e}n have studied the \textsc{Stabilizer State Isomorphism} problem which can be loosely defined as the problem that decides whether there exists some qubit permutation that sends a stabilizer state to another stabilizer state. They provided evidence that the \textsc{Stabilizer State Isomorphism} is an intermediate problem for QCMA and showed that the polynomial hierarchy collapses if the  problem is QCMA-complete~\cite{lockhart2017quantum}. While the problem is not exactly the same as the problem of identifying true automorphism gates for product codes, the complexity of the \textsc{Stabilizer State Isomorphism} problem is likely to suggest hardness for the problem of our interest.

\subsection{Related works}

We acknowledge some previous and related works regarding quantum code automorphisms and their utility in fault-tolerant computation. Grassl and Roetteler \cite{Grassl_2013} introduced the concept of utilizing code automorphisms to perform fault-tolerant logical gates in quantum stabilizer codes. They also show that combining automorphisms with transversal CNOT can significantly enlarge the set of fault-tolerant gates for certain CSS codes. Chao and Reichardt \cite{Chao_2018_Aut} later showed how the full logical Clifford group could be realized in the $\llbracket 15, 7, 3\rrbracket$ quantum Hamming code through a combination of permutations (automorphisms), transversal $H$, and round-robin CZ circuits. Gong and Renes \cite{gong2024} demonstrate a similar result using quantum Reed-Muller codes and propose an implementation in 2D neutral atom arrays.

Hong et al. \cite{LRESC} presented a construction of inherited automorphism gates for hypergraph product codes in the context of concatenation. In this work, we provide a broad theoretical framework that encapsulates and builds upon this earlier work. In particular, we generalize to higher-fold homological products and prove statements about inherent fault tolerance of the resulting gadgets.

More recently, Malcolm et al. \cite{malcolm2025} show how automorphisms of classical codes lift to automorphisms in a subsystem variant of the hypergraph product code. They leverage these inherited automorphisms to achieve a $O(1)$ Clifford compiling ratio using the classical Simplex code family. However, it should be noted that although the weight of the gauge check operators of the subsystem hypergraph product code can be constant, the weight of the stabilizer generators is proportional to the code distance, and so the LDPC property of the stabilizers must be sacrificed if one desires more than a constant code distance. Nonetheless, they show that small-length instances achieve competitive performance under circuit-level noise compared to surface codes of similar distance.

Guyot and Jaques \cite{guyot2025} provide generic bounds on the parameters of any quantum CSS code based on the size and structure of its automorphism group. Since the product constructions in our work produce CSS codes, all of our results are broadly encompassed by their bounds. When restricted to the class of product codes, we are able to derive tighter bounds by examining automorphism limitations of the input codes themselves.

\subsection{Summary of results}

Our contributions can be sorted into three categories: general recipes, analyses of fault tolerance, and a survey of input codes.

Our primary result is explicit recipes of logical gadgets for homological product codes based on the underlying automorphism structures of their input codes; the details of the construction can be found in Section \ref{sec:aut gadgets}. Typically, a code is defined with respect to its codewords, in which case a code automorphism is defined as a permutation symmetry of the codewords. However, the homological product is defined with respect to the parity-check matrices of the classical seed codes rather than their codewords and so depends strongly on the choice of parity-check matrices.
Additionally, both the parity-check matrix and its transpose are put on an equal footing in the homological product; as a consequence, we will need to keep track of how the columns \emph{and} rows transform in the parity-check matrix from a code automorphism.
Geometrically, the Tanner graph of the homological quantum code will be (up to node relabeling) the Euclidean graph product of the input Tanner graphs. As such, one can always arrange the qubits and parity checks of the homological product code in a rectangular layout so that all qubit-check interactions are strictly ``horizontal'' or ``vertical''. Permutations of the (qu)bits in the input codes now become permutations of rows and columns in the homological product code; see Figure \ref{fig:2D HGP Aut gadget} for an illustration.

\begin{thm}[Inherited automorphism gadgets (informal)]
    For a homological product code with two input codes, which can either be classical linear codes or quantum CSS codes, the number of distinct automorphism gadgets is the product of the number of distinct automorphisms (or gadgets) of the input codes. In particular, the group generated by these automorphism gadgets is a direct product of the automorphism groups of the input codes.
\end{thm}

By ``distinct'', we mean that each automorphism gadget corresponds to a distinct transformation (gate) on our logical qubits. This property does not generically hold for all codes, but we show in Theorem \ref{thm:dual aut bound} and Corollary \ref{cor:dual aut equal} that this property does indeed hold for classical linear codes with dual distance at least 3.\footnote{An analogous result for quantum CSS codes has been recently shown in \cite{guyot2025}.} Our recipes are very general and include the case where an input quantum code can itself be a homological product code with automorphism gadgets. These gadgets can then be lifted once again by further homological products. From this recursion, the above theorem also includes the case where we take \textsf{k}-fold homological products of \textsf{k} input codes.

Within the category of fault tolerance, we derive several results concerning both the circuit depth and the spread of correlated errors for our automorphism gadgets.

\begin{cor}[Circuit depth preservation (informal)]
    Suppose we have a code possessing an automorphism gadget with circuit depth $T$. Then its lifted gadget after a homological product will also have circuit depth $T$.
\end{cor}

The above corollary provides us with an incredible amount of flexibility when it comes to designing homological product codes with desired automorphism gadgets; namely, it implies that we do not need to worry about any conflicts between different input codes.

A generic code automorphism corresponding to a permutation transformation on the columns of a parity-check matrix may not necessarily be a permutation transformation on the rows. We identify a subgroup of code automorphisms, which we call Tanner graph automorphisms, where the above property does indeed hold. The name stems from the fact that when a code automorphism acts as a permutation on both sides of the parity-check matrix, it can be interpreted as an automorphism of the associated Tanner graph. An automorphism gadget corresponding to an input Tanner graph automorphism acts strictly as a permutation and so becomes a code automorphism of the homological product code. In fact, we show an even stronger result, which is that these Tanner graph automorphisms lift to Tanner graph automorphisms themselves.

\begin{thm}[Tanner graph automorphism inheritance (informal)]
    Suppose we have a code possessing a Tanner graph automorphism of its parity-check matrix. Then its lifted gadget after a homological product is also a Tanner graph automorphism consisting solely of qubit permutations on the physical level.
\end{thm}

The above theorem also implies that when we recursively perform additional homological products, these Tanner graph automorphisms carry over to every iteration of the product, irrespective of the other input codes. Finally, for the case where one of the input codes is classical, we show that our automorphism gadgets can be made distance-preserving upon ignoring (gauging) some of the logical qubits. Specifically, we show that circuit-level correlated errors do not decrease the effective fault distance, under the assumption that physical permutations do not spread errors. This assumption is typically valid when permutations are performed by physical movement rather than a series of entangling gates, like in trapped ion or neutral atom architectures. We were only able to show a partial result (Theorem \ref{thm:d_eff 4D HGP}) for the case where \emph{both} input codes are quantum.

\begin{thm}[Effective distance preservation (informally Theorems \ref{thm:d_eff 2D HGP} and \ref{thm:d_eff 3D HGP 2})]
    For hypergraph product and (quantum $\times$ classical) homological product codes, if the input automorphism gadgets preserve the effective fault distance, then so will their lifted versions after the homological product.
\end{thm}

The above theorem says that even when our automorphism gadgets involve circuits that can spread errors, these correlated errors do not affect our homological product code in the most adversarial way. Along the way, we prove some technical statements about logical sector weights in repeated homological products with classical input codes (Proposition \ref{prop:3D HGP restricted logical sector weight}) and quantum input codes (Proposition \ref{prop:distance-bounds-subsystem-hom-prod-code}), which may be of independent interest.

We also explore the rich literature of classical linear codes and survey various code families that best suit our framework. These codes can be fed into the homological product to either produce a new CSS code or upgrade the logical gadgets of an existing CSS code.
Table~\ref{tab:codes_summary} presents a summary of the classical codes we study in this work, including information about their parameters, automorphism groups, and method of implementation. To the best of our knowledge, the only previously studied examples of homological product codes with automorphism gates were hypergraph products of quasicyclic codes~\cite{xu2024fast}. For near-term implementation, the most promising code families we have identified are group-algebra codes and cycle codes. Both of these code constructions are capable of producing LDPC codes with constant rate and growing distance. For classical codes based on group algebras, we show that they always possess a Tanner graph automorphism group isomorphic to the underlying group in the algebra, even when the group is non-abelian. For the cycle codes, any automorphism of the underlying graph is a Tanner graph automorphism. Importantly, the Tanner graph automorphism group for the cycle codes can sometimes be much larger than the length of the code, unlike the lower bound that we have for the group-algebra codes. An interesting observation with these two code constructions is that they often produce rank-deficient parity-check matrices. Nonetheless, the redundant parity checks are crucial to the Tanner graph automorphism structure;  although removing these redundant parity checks will not alter the code automorphism group, it will change both the structure of the homological product as well as the implementation of its corresponding automorphism gadgets.

\begin{table}[t]\renewcommand{\arraystretch}{1.4}
\centering
\begin{tabular}{c|c|c|c|c}
     Classical code $\mathcal{C}$ & Code parameters $[n,k,d]$ &   $\subseteq \mathrm{Aut}(\mathcal{C})$ & LDPC? & Aut gadget \\ 
     \hline     
     Cycle ($\mathsf{G}=(V,E)$) & $\big[\,\abs{V},\abs{E}-\abs{V}+1,\mathrm{girth}(\mathsf{G})\,\big]$ & $\mathrm{Aut}(\mathsf{G})$ & \cmark & perm \\

     Group-algebra $(\mathcal{G})$ & $[\, \abs{\mathcal{G}}, ?, ? \,]$ & $\mathcal{G}$ & \cmark & perm  \\

     $\mathcal{G}$-lift $(\ell,m)$ & $\big[\, m\abs{\mathcal{G}}, \geq (m-\ell)\abs{\mathcal{G}}, ? \,\big]$ & $\mathcal{G}$ & \cmark & perm \\

     Hamming & $[\,2^r-1, 2^r-r-1, 3\,]$ & $\mathrm{GL}_r(\mathbb{F}_2)$ & \xmark & perm + circuit \\

     Simplex & $[\,2^r-1, r, 2^{r-1}\,]$ & $\mathrm{GL}_r(\mathbb{F}_2)$ & \cmark & perm + circuit \\

     Reed-Muller $\mathrm{RM}(r,m)$ & $\big[\,2^m, \sum^r_{i=0} {m \choose i}, 2^{m-r}\,\big]$ & $\mathrm{GA}_m(\mathbb{F}_2)$ & \xmark & perm + circuit \\

     Punctured $\mathrm{RM}^*(r,m)$ & $\big[\,2^m-1, \sum^r_{i=0} {m \choose i}, 2^{m-r}-1 \,\big]$ & $\mathrm{GL}_m(\mathbb{F}_2)$ & \xmark & perm + circuit \\
\end{tabular}
\caption{A survey of classical linear codes and their relevant properties for the construction of automorphism gadgets.}
\label{tab:codes_summary}
\end{table}

Lastly, we present three explicit, finite-length quantum codes built from the hypergraph and homological product that possess interesting automorphism structures. In the first example, we show that the existence of automorphism gadgets in a $\llbracket 52,10,3 \rrbracket$ hypergraph product code can reduce the overhead of external resources when generating logical Clifford gates. In the second and third examples, we use automorphism gadgets in a $\llbracket 48,6,3 \rrbracket$ hypergraph product code and a $\llbracket 288,9,(d_X=16,d_Z=3) \rrbracket$ code to increase the addressability of transversal CZ and CCZ gates based on cohomology invariants \cite{breuckmann2024cups}.

\subsection{Discussion and outlook}

Our general-purpose recipes for automorphism gadgets are directly compatible with other fault-tolerant gadgets for homological product codes, such as distance-preserving single-ancilla syndrome extraction \cite{Manes_2025, tan2024hgp}, single-shot error correction \cite{Campbell_2019}, fold-transversal Clifford gates from $ZX$-dualities \cite{Quintavalle_2023, Breuckmann_2024_ZX}, extractor systems for logical Pauli measurements \cite{he2025extr}, and transversal Clifford and non-Clifford gates from cohomological invariants \cite{breuckmann2024cups, hsin2024}\footnote{We note that the more general construction in the last section of \cite{hsin2024} requires a post-modification by Freedman and Hastings \cite{Freedman_2021}, and it is unclear to us whether any automorphism gadgets can survive this transformation.}. In addition, it has been shown that the row-column structure of the parity checks of homological product codes is particularly amenable to implementation in neutral atom arrays \cite{Xu2024, constantinides2024}; since our gadgets follow the same row-column structure, their implementation should be similarly suitable in this platform.

We note that these gadgets also have an interpretation in condensed matter physics. A quantum stabilizer code implicitly defines a commuting and frustration-free Hamiltonian comprised of a uniform sum over all parity checks and whose ground state subspace is the codespace. If the code is LDPC and the distance is macroscopic (i.e. the code is degenerate) then we say the model is topological \cite{Bravyi_2010_top}. Excitations and energies are related to error syndromes and syndrome weights respectively. There has been a large body of recent works which analyzes qLDPC codes and their physical implications through this lens, ranging from code constructions \cite{rakovszky2023physics1, rakovszky2024physics2} to stability against perturbations \cite{lavasani2024stable, yin2024stable, deroeck2024stable} and thermal fluctuations \cite{Hong_2025, gamarnik2024slow, placke2024glass}. Furthermore, many fault-tolerant gadgets can be understood from the framework of higher-form symmetries \cite{zhu2024higher}. From this physical perspective, a code automorphism is a spectrum-preserving permutation symmetry of the code Hamiltonian, and a Tanner graph automorphism is an exact permutation symmetry of the Hamiltonian; e.g. translational invariance. Notably, this permutation symmetry can act nontrivially on the ground state subspace. We also point out that this notion of permutation symmetry acts on spatial positions and is different than the usual symmetries which act on internal degrees of freedom. Since we show that Tanner graph automorphisms lift to Tanner graph automorphisms in homological product codes, we provide a method to construct topological stabilizer models exhibiting nontrivial permutation symmetry structures beyond translational invariance (but may have long range).

There are a few outstanding open questions and future directions that we summarize below.

\begin{openprob}[Effective distance preservation for quantum input codes]
    Do the automorphism gadgets of (quantum $\times$ quantum) homological product codes inherit the same effective distance preservation from their input gadgets?
\end{openprob}

We could only achieve a partial result (Theorem \ref{thm:d_eff 4D HGP}) on the effective fault distance of our automorphism gadgets in the (quantum $\times$ quantum) homological product case. This problem is related to the lower bound on the minimum distances of generic (quantum $\times$ quantum) homological product codes, which has not yet been shown to match the upper bound \cite{Zeng_2020}. It would be interesting to see if our effective distance lower bound could be tightened to match these upper bounds. This would allow us to feed in multiple quantum CSS codes into the homological product without having to worry about their underlying gadgets decreasing the effective distance of the code.

\begin{openprob}[Full logical Clifford group]
    Does there exist any high-rate, LDPC homological product code family that can realize the full logical Clifford group using a combination of transversal, fold-transversal and constant-depth automorphism gadgets?
\end{openprob}

By ``high rate'', we do not mean constant rate, but rather simply $k = \omega(1)$ in order to discount surface and color codes, which are also LDPC but have $k=O(1)$. Note that transversal, fold-transversal and constant-depth automorphism gadgets (assuming permutations are free) all only cost $O(1)$ spacetime overhead in their implementation and so are extremely desirable from a computational standpoint. A somewhat discouraging observation is that when we take a \textsf{k}-fold homological product of \textsf{k} input codes, the total number of automorphism gadgets is a \textsf{k}-fold product that can scale at most exponentially with \textsf{k};\footnote{Noting that the largest possible automorphism group size on $k$ logical qubits is $\abs{\mathrm{GL}_k(\mathbb{F}_2)} \propto 2^{k^2}$.} however, the number of logical qubits $k$ also increases exponentially with \textsf{k}, and so the size of the Clifford group modulo the Pauli group $\mathcal{C}_k / \mathcal{P}_k \simeq \mathrm{Sp}_{2k}(\mathbb{F}_2)$ scales \emph{doubly} exponentially with \textsf{k}. Fortunately, some previous works have demonstrated that (non-LDPC) quantum codes with seemingly small automorphism groups can be enlarged to fill the entire logical Clifford group when supplemented with only a few additional transversal and fold-transversal gates \cite{Grassl_2013, gong2024, malcolm2025}\footnote{One can directly transfer the results of \cite{malcolm2025} to the usual hypergraph product, but the resulting automorphism gadgets will not have constant depth as far as we are aware.}. It has yet to be seen if any qLDPC stabilizer code family satisfies the same property. We hope that the framework presented in this paper leads to conducive progress along this research front.

\begin{openprob}[Full universality without distillation]
    Does there exist any high-rate, LDPC homological product code family that can natively realize universal logic using a combination of transversal, fold-transversal, constant-depth automorphism gadgets, and measurements without magic state distillation?
\end{openprob}

No-go theorems on transversal gates tell us that universality cannot be achieved solely with the first three operations \cite{Eastin_2009, Webster_2022}. We added the word ``natively'' to discount the protocols that selectively teleport logical qubits from a homological product code to a surface code and perform all logic using surface-code gadgets \cite{Xu2024}. It has been recently shown that a \textsf{k}-fold hypergraph product code (of \textsf{k} classical codes) can admit transversal gates in the \textsf{k}th level of the Clifford hierarchy \cite{breuckmann2024cups, hsin2024}; it has been further proven that this lower bound is tight \cite{fu2025nogo}. Given that universal computation without distillation is possible with 3D surface codes \cite{Vasmer_2019_3D}, perhaps similar techniques could be useful for high-rate homological product codes.

Lastly, we comment on possible generalizations of our framework to other product constructions such as lifted products \cite{panteleev2021quantum} and balanced products \cite{Breuckmann_2021_BP}. Both of these constructions can be viewed as taking the usual hypergraph product of two input codes sharing a common symmetry and then factoring out a  ``diagonal'' subgroup of the resulting product symmetry group, with the balanced product being more general. For example, taking a balanced product of two classical group-algebra codes (Section \ref{sec:group-algebra codes}) with the same base symmetry (Tanner graph automorphism) group results in a two-block group-algebra quantum code \cite{2BGA} if we factor out the full diagonal subgroup. If the base symmetry group is abelian, e.g. in the bivariate bicycle codes \cite{Bravyi_2024}, then it is still a subgroup of the automorphism group of the resulting code. If the base group is non-abelian, then we generically only retain its center after the symmetry reduction \cite{2BGA}, which can be significantly smaller than the original base group itself. However, if our non-abelian base group contains a nontrivial normal subgroup, then we can choose to factor out by this smaller subgroup rather than the full group itself. The remaining structure associated with the quotient group then resembles the usual hypergraph product, and our input Tanner graph automorphisms carry over with respect to this quotient group. It would be interesting to explore whether any structure is preserved for the more generic automorphism gadgets that are not Tanner graph automorphisms.

\subsection{Outline}

The rest of the work is structured as follows. In Section~\ref{sec:preliminaries}, we introduce automorphisms of classical and quantum CSS codes as well as the main quantum code construction studied in this work: the homological product. We then review and discuss in Section~\ref{sec:bounds} several limitations to the automorphism group and resulting codes in both the classical and quantum setting, with special emphasis to codes on graphs. 
In Section~\ref{sec:aut gadgets}, we showcase how automorphisms of classical and quantum error-correcting codes lift to logical gadgets in their homological products. In Section~\ref{sec:fault tolerance}, we prove effective distance preservation for two out of three of our gadgets and describe how they can be incorporated into a fault-tolerant architecture. In Section~\ref{sec:classical families}, we survey a handful of classical error-correcting codes that possess interesting automorphism groups and examine how they fit into our framework. In Section~\ref{sec:examples}, we present three explicit quantum codes where we show that automorphism gadgets can introduce new logical gates or enhance existing logical gates.


\section{Preliminaries}
\label{sec:preliminaries}

\subsection{Notation}

Let $[n]$ denote the set $\{1, 2, \ldots, n\}$. Let $\mathbf{x} \in \mathbb{F}_2^{n}$ be an $n$-dimensional column vector whose Hamming weight $\abs{\mathbf{x}}$ is given by the number of non-zero entries in the vector. In the case where the entries of an $n$-dimensional column vector can be arranged in some grid of size $a \times b$, we denote $\mathbb{M}[\mathbf{x}] \in \mathbb{F}_2^{a \times b}$ as the matrix representation, or reshaping, of the vector $\mathbf{x}$ that respects the grid arrangement of $\mathbf{x}$. We sometimes drop the $\mathbb{M}$ notation for brevity when it is clear that we are referring to the matrix representation of the vector. In addition, we define $\abs{\,\cdot\,}_r$ and $\abs{\,\cdot\,}_c$ as the number of rows and columns respectively of some matrix with non-zero entries. In other words, $\abs{\mathbb{M}[\mathbf{x}]}_r$ is the number of row vectors in $\mathbb{M}[\mathbf{x}]$ where there exists some column with a non-zero entry; with slight abuse of notation, we will sometimes write $\abs{\mathbf{x}}_r \equiv \abs{\mathbb{M}[\mathbf{x}]}_r$ when the underlying 2D arrangement is obvious. We denote the submatrix of the matrix representation of $\mathbf{x}$ with the set of rows with row indices $A \subseteq [a]$ and columns with column indices $B \subseteq [b]$ as $\mathbb{M}[\mathbf{x}][A, B] = \mathbf{x}[A,B]$. We can similarly let $\mathbf{x}[A] \in \mathbb{F}_2^{|A|}$ for $A \subseteq [n]$ denote the restriction of the column vector $\mathbf{x}$ to the row indices in $A$. Let $\oplus$ and $\otimes$ denote the direct sum and tensor product, respectively. It should also be clear from context as to whether an addition operation is done modulo 2. When we are simultaneously discussing both classical and quantum codes, we reserve lowercase letters for quantities related to the classical code and uppercase letters for those of the quantum code, unless stated otherwise.

\subsection{Classical linear codes and automorphisms}\label{sec:code-automorphisms}

A classical binary $[n, k, d]$ linear code $\mathcal{C}$ is defined to be a $k$-dimensional subspace of an $n$-dimensional vector space $\mathbb{F}_2^n$ such that the Hamming weight of any nonzero codeword $\mathbf{c} \in \mathcal{C}$ satisfies $|\mathbf{c}| \geq d$. We typically say that $k$ \emph{logical} bits are encoded in $n$ \emph{physical}, or data, bits.
In particular, $\mathcal{C}$ can be identified as the row space of a generator matrix $G$ where the rows of $G$ form a \emph{logical basis} for $\mathcal{C}$.
Alternatively, $\mathcal{C}$ can also be identified by a (non-unique) parity-check matrix $H$ such that $\mathcal{C} = \ker H$ i.e., $H\mathbf{c} = 0$ for all $\mathbf{c} \in \mathcal{C}$ and equivalently $HG^\transpose=0$.
An automorphism on a classical code is defined to be a permutation of coordinates, or bit positions, that stabilizes the codespace. 
When $\mathcal{C}$ is linear, then it implies that the generator matrix $G$ is related to the permuted $G'$ under a linear transformation. To be precise, let $\sigma$ be an $n\times n$ permutation matrix. Then in order for the rows of $G' = G\sigma$ and $G$ to span the same codespace (row space), there must exist an invertible $k \times k$ matrix $V$ such that $G' = VG$. In other words, we have the following \emph{automorphism condition}:
\begin{align}\label{eq:aut condition}
    G\sigma = VG \, ,
\end{align}
with $\sigma \in \mathrm{S}_n$ and $V \in \mathrm{GL}_k(\mathbb{F}_2)$. Conversely, if there exists such a $\sigma$ and $V$ satisfying \eqref{eq:aut condition}, we say $\sigma$ is an automorphism of $\mathcal{C}$. Since $V$ can potentially map codewords to other codewords, we say it enacts a \emph{logical operation} on the code. Since the composition of permutations is another permutation, and the composition of invertible matrices is an invertible matrix, if $\sigma, \sigma'$ are two code automorphisms, then so is $\sigma\sigma'$. In addition, since $\sigma$ and $V$ are both invertible, we can multiply \eqref{eq:aut condition} on the right by $\sigma^{-1}$ and on the left by $V^{-1}$ to get $G\sigma^{-1} = V^{-1}G$, which tells us that $\sigma^{-1}$ is also a code automorphism. Combined with the associativity of matrix multiplication, we see that the automorphisms $\{ \sigma \}$ form a finite subgroup, which we we call the automorphism group or $\mathrm{Aut}(\mathcal{C})$ for short, of the permutation group $\mathrm{S}_n$. Similarly, the collection of corresponding logical operations $\{ V \}$ forms a subgroup, which we call the \emph{logical} automorphism group $\mathcal{A}(\mathcal{C})$, within the group of all logical operations $\mathcal{L}$. The automorphism condition \eqref{eq:aut condition} induces a correspondence $\phi : \mathrm{Aut}(\mathcal{C}) \mapsto \mathcal{A}(\mathcal{C})$ through the generator matrix $G$.

\begin{prop}
    $\phi : \mathrm{Aut}(\mathcal{C}) \to \mathcal{A}(\mathcal{C})\,;\, \sigma \mapsto V$ is a group homomorphism.
\end{prop}

\begin{proof}
    We first show that $\phi$ is a well-defined mapping. Since $G$ has full (row) rank by definition, there exists a (non-unique) right-inverse $G^{-1}$ that we can choose for our construction of $\phi$. All that remains is to show that for any $\sigma \in \mathrm{Aut}(\mathcal{C})$, we have $\phi(\sigma) \in \{V\} \subseteq \mathcal{A}(\mathcal{C})$. Suppose there exist $V, V' \in \mathcal{A}(\mathcal{C})$. Using the fact described above the proposition statement, we can compose $\sigma$ with its inverse and use \eqref{eq:aut condition} to get $G\sigma\sigma^{-1} = G = V' V^{-1} G$. By multiplying $G^{-1}$ on the right in the previous expression, we obtain $\ident = V' V^{-1}$ and hence $V = V'$ as desired. To show that $\phi$ is a group homomorphism, it suffices to show that $\phi(\sigma \sigma') = \phi(\sigma) \phi(\sigma')$:
    \begin{align}
        \phi(\sigma \sigma') &= VG\sigma'G^{-1} = VV'GG^{-1} = VV' = \phi(\sigma) \phi(\sigma') \, .
    \end{align}
\end{proof}

For a linear code $\mathcal{C}$, we can also construct its orthogonal complement, or dual, $\mathcal{C}^\perp$. A generator matrix of $\mathcal{C}^\perp$ is a parity-check matrix of $\mathcal{C}$. Since $\sigma$ is an orthogonal transformation satisfying $\sigma\sigma^\transpose = \ident$, if $\sigma \in \mathrm{Aut}(\mathcal{C})$, then we also have that $\sigma \in \mathrm{Aut}(\mathcal{C}^\perp)$ and vice versa, which implies that $\mathrm{Aut}(\mathcal{C})$ and $\mathrm{Aut}(\mathcal{C}^\perp)$ are isomorphic groups, i.e. $\mathrm{Aut}(\mathcal{C}) \simeq \mathrm{Aut}(\mathcal{C}^\perp)$. Thus, all the aforementioned statements about $\mathrm{Aut}(\mathcal{C})$ and $\mathcal{A}(\mathcal{C})$ also hold for $\mathrm{Aut}(\mathcal{C}^\perp)$ and $\mathcal{A}(\mathcal{C}^\perp)$. In particular, we have the dual group homomorphism $\phi^\perp : \mathrm{Aut}(\mathcal{C}^\perp) \mapsto \mathcal{A}(\mathcal{C}^\perp)$ and the dual automorphism condition
\begin{align}\label{eq:dual aut condition}
    H\sigma = WH
\end{align}
with $W \in \mathrm{GL}_m(\mathbb{F}_2)$.

The quantum codes we are interested in will be defined with respect to the parity-check matrices of the input codes (rather than the generator matrices). As such, a stricter notion of code automorphism called \emph{Tanner graph automorphism} will be useful later, and so we present it here.

\begin{defn}[Graph automorphism]\label{defn:graph automorphism}
    Given a (hyper)graph $\mathsf{G} = (V,E)$, an automorphism of $\mathsf{G}$ is defined as a permutation of the vertices $\pi(V)$ such that for a collection of vertices $\{ v_j \} \subset V$, we have $\pi\big( \{ v_j \} \big) \in E$ if and only if $\big( \{ v_j \} \big) \in E$.
\end{defn}

In other words, two graphs $\mathsf{G}$ and $\mathsf{G}'$ are isomorphic when one can be obtained from the other upon relabeling vertices and edges. For a linear code, its parity-check matrix $H \in \mathbb{F}^{m\times n}_2$ gives the vertex-edge incidence matrix of its Tanner graph. An automorphism of the Tanner graph given by $H$ is then given by the relation
\begin{align}\label{eq:graph automorphism}
    H = \sigma_{\rm e} H \sigma_{\rm v} \, ,
\end{align}
where $\sigma_{\rm e}$ and $\sigma_{\rm v}$ are permutation matrices whose respective rows and columns encode the edge and vertex relabeling. Comparing \eqref{eq:graph automorphism} with the dual automorphism condition \eqref{eq:dual aut condition}, we see that Tanner graph automorphisms are automorphisms where $W$ is a permutation matrix\footnote{In \cite{sayginel2025}, these are also called ``matrix automorphisms''.}. We denote $\mathcal{T} \subseteq \mathcal{A}$ the subgroup of code automorphisms corresponding to Tanner graph automorphisms. For a given generator matrix $G$, we can always construct a parity-check matrix for which all code automorphisms are Tanner graph automorphisms:

\begin{prop}[All code automorphisms as Tanner graph automorphisms]
\label{prop:all graph auts}
    Every binary linear code $\mathcal{C}$ admits a parity-check matrix $H$ such that \eqref{eq:graph automorphism} holds for all $\sigma_{\rm v} \in \mathrm{Aut}(\mathcal{C})$.
\end{prop}

\begin{proof}
    The rows of any parity-check matrix of $\mathcal{C}$ must span $\mathcal{C}^\perp$, which has dimension $n-k$. Let $H \in \mathbb{F}^{(2^{n-k}-1)\times n}_2$ be a parity-check matrix whose rows consist of all $2^{n-k}-1$ nonzero vectors in $\mathcal{C}^\perp$. Since $\sigma_{\rm v} \in \mathrm{Aut}(\mathcal{C}) \simeq \mathrm{Aut}(C^\perp)$, it must preserve $\mathcal{C}^\perp$ by definition and so can only permute the $2^{n-k}-1$ rows of $H$ (note that the all-zero vector remains invariant and is not included in $H$).
\end{proof}

Although Proposition \ref{prop:all graph auts} provides us a method to make all code automorphisms Tanner graph automorphisms, the cost is a dense parity-check matrix with an exponentially large number of rows. In the product constructions we will be interested in, the number of physical qubits as well as the weight of the stabilizer checks will be directly tied to the row and column weights of the input parity-check matrices, and so it will be pertinent to ensure that these weights do not blow up for the sake of fault tolerance.

\subsection{Quantum CSS codes and automorphisms}

A $\llbracket n,k,d \rrbracket$ quantum stabilizer code encodes $k$ logical qubits in $n$ data qubits and can detect errors with weight less than $d$. The codespace is defined by its stabilizer group $\mathcal{S}$, an abelian subgroup of the Pauli group $\mathcal{P}_n$ on $n$ qubits \cite{gottesman1997thesis}. A generating set of $\mathcal{S}$, containing Pauli check operators, is typically used to detect and correct errors. The codespace is defined as the simultaneous +1 eigenspace of these check operators (and hence the entire stabilizer group). It is customary to define a stabilizer code in this ``dual'' picture, rather than listing out the logical Pauli operators themselves in analogy to a generator matrix, since a logical Pauli may have many different equivalent representations related by an element of the stabilizer group, a phenomenon called stabilizer degeneracy. A quantum Calderbank-Shor-Steane (CSS) code is a stabilizer code whose check operators are purely $X$-type and $Z$-type Paulis \cite{Calderbank_1996, Steane_1996}. As such, its stabilizer group also factorizes into $X$-type and $Z$-type subgroups. For a quantum CSS code, we can write its $X$-type and $Z$-type parity checks in terms of a binary $2n\times 2n$ block-diagonal matrix
\begin{align}\label{eq:CSS H}
    H_{\rm css} = \begin{pmatrix}
        H_X & \mathbf{0} \\
        \mathbf{0} & H_Z
    \end{pmatrix} \, ,
\end{align}
where $H_X$ and $H_Z$ label the supports of the $X$-type and $Z$-type parity checks respectively. The interpretation of \eqref{eq:CSS H} is that the first $n$ columns label $X$-type support and the last $n$ columns $Z$-type support. 

A permutation $\sigma \in \mathrm{S}_n$ of the qubits now corresponds to the operator $\sigma \oplus \sigma = \mathrm{diag}(\sigma,\sigma)$. The (dual) automorphism condition for CSS codes can then be expressed as
\begin{align}
    \begin{pmatrix}
        W_X & \mathbf{0} \\
        \mathbf{0} & W_Z
    \end{pmatrix}\begin{pmatrix}
        H_X & \mathbf{0} \\
        \mathbf{0} & H_Z
    \end{pmatrix} = \begin{pmatrix}
        H_X & \mathbf{0} \\
        \mathbf{0} & H_Z
    \end{pmatrix}\begin{pmatrix}
        \sigma & \mathbf{0} \\
        \mathbf{0} & \sigma
    \end{pmatrix} \, ,
\end{align}
or
\begin{subequations}\label{eq:CSS aut conditions}
\begin{align}
    W_X H_X &= H_X \sigma  \\
    W_Z H_Z &= H_Z \sigma \, .
\end{align}
\end{subequations}

\subsection{Chain complexes}\label{sec:chain_complexes}
In this section, we introduce some basic definitions of chain complexes that are useful for our subsequent analysis.

\begin{defn}[Chain Complexes]\label{def:chain-complexes}
  A chain complex $\mathcal{C}_\ast$ over a field $\mathbb{F}_q$ consists of a sequence of $\mathbb{F}_q$-vector spaces $\left(C_i\right)_{i \in \mathbb{Z}}$. These vector spaces are related to each other by linear boundary maps $\left(\partial_i^\mathcal{C}: C_i \to C_{i-1}\right)_{i \in \mathbb{Z}}$ satisfying $\partial_{i-1}^\mathcal{C} \circ \partial_i^\mathcal{C} = 0$ for all $i \in \mathbb{Z}$.
  When it is clear from context, we omit the superscript and subscript for the linear boundary maps. 
  By choosing an appropriate basis, we refer to the basis elements of $C_i$ as $i$-cells and an arbitrary vector in $C_i$ as an $i$-chain.
  If there exists bounds $\ell < m \in \mathbb{Z}$ such that for all $i < \ell$ and $i > m$ have $C_i = 0$, then we may truncate the sequence and say that $\mathcal{C}$ is the $(m-\ell + 1)$-term chain complex
  \[\mathcal{C}_\ast = \left(C_m \xrightarrow[]{\partial_m} C_{m-1} \xrightarrow[]{\partial_{m-1}} \ldots \xrightarrow[]{\partial_{\ell + 1}} C_\ell\right).\]
  We furthermore define the following (standard) vector spaces for $i \in \mathbb{Z}$:
  \begin{align}
    \text{the space of } i\text{-cycles: } Z_i\left(\mathcal{C}\right) &\coloneqq \ker\left(\partial_i\right) \subseteq C_i,\\
    \text{the space of } i\text{-boundaries: } B_i\left(\mathcal{C}\right) &\coloneqq \mathrm{im}\left(\partial_{i+1}\right) \subseteq C_i,\\
    \text{the space of } i\text{-homology: } \mathcal{H}_i\left(\mathcal{C}\right) &\coloneqq Z_i\left(\mathcal{C}\right)/B_i\left(\mathcal{C}\right).
  \end{align}

  The cochain complex $\mathcal{C}^\ast$ associated to $\mathcal{C}_\ast$ is the chain complex with vector spaces $\left(C^i \coloneqq C_i\right)_{i \in \mathbb{Z}}$ and linear boundary maps given by the coboundary maps $\left(\delta_i = \partial_{i+1}^\transpose : C^i \to C^{i+1}\right)_{i \in \mathbb{Z}}$ obtained by transposing all the boundary maps of $\mathcal{C}_\ast$. 
  Similarly, by choosing an appropriate basis, we refer to the basis elements of $C^i$ as the $i$-cocells and an arbitrary vector in $C^i$ as an $i$-cochain.
  Thus, the cochain complex is defined as such:
  \[\mathcal{C}^\ast = \left(C^m \xleftarrow[]{\delta_{m-1}} C^{m - 1} \xleftarrow[]{\delta_{m-2}} \ldots \xleftarrow[]{\delta_{\ell}} C^\ell\right).\]
  We can analogously define the spaces of cohomology $\mathcal{H}^i\left(\mathcal{C}\right) = Z^i(\mathcal{C})/B^i(\mathcal{C})$, cocycles $Z^i\left(\mathcal{C}\right) = \ker\left(\delta_i\right)$, and coboundaries $B^i\left(\mathcal{C}\right) = \mathrm{im}\left(\delta_{i-1}\right)$.
\end{defn}

Classical linear codes can be described by 2-term chain complexes where the two vector spaces are the spaces of bits and checks respectively such that these spaces are related to each other by a linear boundary map that can be written as the check matrix $H$.
It is also a well-known fact that quantum CSS codes can be described with a 3-term chain complex by associating the $X$ stabilizers, qubits, and $Z$ stabilizers with the three consecutive vector spaces in a 3-term chain complex. 
An easy way to see the connection is to relate the condition $\partial_{i-1} \circ \partial_i = 0$ to $H_X H_Z^\transpose = 0$. 
We now state an important definition that relates the $X$ and $Z$ distances of a quantum CSS code to its associated chain complex.

\begin{defn}
  For a chain complex $\mathcal{C}$, the $i$-systolic distance $d_i(\mathcal{C})$ and the $i$-cosystolic distance $d^i(\mathcal{C})$ are defined as
  \[d_i(\mathcal{C}) = \min_{c \in Z_i(\mathcal{C})\setminus B_i(\mathcal{C})}|c|,\quad\quad d^i(\mathcal{C}) = \min_{c \in Z^i(\mathcal{C})\setminus B^i(\mathcal{C})}|c|.\]
\end{defn}

Suppose we associate the $X$ stabilizers, qubits, and $Z$ stabilizers to the $\mathbb{F}_2$-vector spaces $C_0, C_1$, and $C_2$, then the $X$ and $Z$ logical operator representatives are given by the basis elements of the $1$-cohomology space and $1$-homology space respectively. The $X$ and $Z$ distances of the quantum code are then $d^1(\mathcal{C})$ and $d_1(\mathcal{C})$.

The framework of chain complexes provides us with a diagrammatic way to write the (dual) automorphism condition for classical linear codes \eqref{eq:dual aut condition}. Interpreting the $m \times n$ parity-check matrix $H: B\mapsto S$ as a linear map from the binary vector spaces of bit flips $B$ to syndromes $S$, a code automorphism is a pair of matrices $\sigma \in \mathrm{S}_n$ and $W \in \mathrm{GL}_m(\mathbb{F}_2)$ such that the following diagram is commutative:
\begin{equation}
\begin{tikzcd}
B \arrow[r, "H"] \arrow[d, "\sigma"'] & S \arrow[d, "W"] \\
B \arrow[r, "H"]                      & S               
\end{tikzcd}
\end{equation}
For quantum CSS codes, we have the binary vector spaces of $X$-syndromes, qubit errors and $Z$-syndromes, organized into a 3-term chain complex $S_X \xrightarrow{H^\transpose_X} Q \xrightarrow{H^{}_Z} S_Z$. The CSS automorphism condition \eqref{eq:CSS aut conditions} can then be rephrased as having matrices $\sigma \in \mathrm{S}_n$, $W_X \in \mathrm{GL}_{m_X}(\mathbb{F}_2)$ and $W_Z \in \mathrm{GL}_{m_Z}(\mathbb{F}_2)$ such that the following diagram is commutative:
\begin{equation}
\begin{tikzcd}
S_X \arrow[r, "H^\transpose_X"] \arrow[d, "W_X"'] & Q \arrow[r, "H^{}_Z"] \arrow[d, "\sigma"] & S_Z \arrow[d, "W_Z"] \\
S_X \arrow[r, "H^\transpose_X"]                   & Q \arrow[r, "H^{}_Z"]                     & S_Z                 
\end{tikzcd}
\end{equation}

\subsection{Homological product of chain complexes}
\label{subsec:homological-product}

This section states the basic notions of the homological product of chain complexes. Note that the homological product is also known as the tensor product of chain complexes.

\begin{defn}[Homological Product]\label{def:homological-product}
  For chain complexes $\mathcal{A}$ and $\mathcal{B}$, the homological product $\mathcal{C} = \mathcal{A} \otimes \mathcal{B}$ is the chain complex given by the vector spaces
  \[C_i \coloneqq \bigoplus_{j \in \mathbb{Z}} A_j \otimes B_{i-j}\]
  and the boundary maps
  \[\partial_i^\mathcal{C} \coloneqq \bigoplus_{j \in \mathbb{Z}}\left(\partial_j^\mathcal{A} \otimes \ident + \ident \otimes \partial_{i-j}^\mathcal{B}\right).\]
\end{defn}

We now state a well-known result about the homological product.

\begin{prop}[K\"unneth Formula]\label{prop:kunneth-formula}
  Let $\mathcal{A}$ and $\mathcal{B}$ be chain complexes over a field $\mathbb{F}_q$, each with a finite numbers of nonzero terms.
  Then for every $i \in \mathbb{Z}$,
  \[\mathcal{H}_i(\mathcal{A} \otimes \mathcal{B}) \cong \bigoplus_{j \in \mathbb{Z}} \mathcal{H}_j(\mathcal{A}) \otimes \mathcal{H}_{i -j}(\mathcal{B}).\]
  Furthermore, for $a \in Z_j(\mathcal{A})$ and $b \in Z_{i-j}(\mathcal{B})$, the isomorphism above maps 
  \[a \otimes b + B_i(\mathcal{A} \otimes \mathcal{B}) \mapsto \left(a+B_j(\mathcal{A})\right)\otimes (b+B_{i-j}(\mathcal{B})).\]
\end{prop}

The K\"unneth formula allows us to compute the number of logical qubits in the homological product code based on properties of the input codes. By analyzing the rank of the homology group, we can determine the number of logical qubits and the structure of their logical Pauli operators.

We note that the definitions stated in Sections~\ref{sec:chain_complexes} and \ref{subsec:homological-product} are with respect to the multi-sector theoretic version of chain complexes. Some readers might be familiar with the single-sector theoretic version introduced in \cite{Bravyi_2014_HGP} which contains an unbounded chain complex where every vector space and boundary map is \emph{identical}. While the single-sector theoretic version simplifies some of the homology analysis and allows the associated quantum code to break the $O(\sqrt{n})$ distance barrier, the code ceases to have good distance bounds once we impose the LDPC condition on it~\cite{golowich2024quantum}. In order to have meaningful fault tolerance~\cite{Gottesman_2014} and interesting automorphisms that arise from the inheritance of \emph{distinct} boundary maps, we restrict our study to homological product codes of the multi-sector theoretic version defined in Sections~\ref{sec:chain_complexes} and \ref{subsec:homological-product}.

\subsection{Hypergraph product codes}\label{sec:HGP codes}

One of the first constructions of quantum CSS codes from classical linear codes with large rate and distance is the hypergraph product (HGP)~\cite{Tillich_2014}. 
The HGP construction is no different from taking the homological product of two 2-term chain complexes that correspond to two classical linear codes.
Consider two classical linear codes with parameters $[n_1,k_1,d_1]$ and $[n_2, k_2, d_2]$ and parity-check matrices $h_1$ and $h_2$, each of dimension $m_i \times n_i$. Taking a Kronecker product of their chain complexes, we obtain the product complex
\begin{equation}\label{eq:HGP complex}
    \begin{tikzcd}
        & {B_1 \otimes S_2} &&& {S_Z} \\
        {B_1 \otimes B_2} && {S_1 \otimes S_2} && Q \\
        & {S_1 \otimes B_2} &&& {S_X}
        \arrow["{\ident \otimes h^{}_2}", from=2-1, to=1-2]
        \arrow["{h^\transpose_1 \otimes \ident}"', from=2-3, to=1-2]
        \arrow["{H_Z}", from=2-5, to=1-5]
        \arrow["{h^\transpose_1 \otimes \ident}", from=3-2, to=2-1]
        \arrow["{\ident \otimes h^{}_2}"', from=3-2, to=2-3]
        \arrow["{H^\transpose_X}", from=3-5, to=2-5]
    \end{tikzcd}
\end{equation}
where $S_X, Q, S_Z$ are the binary vector spaces of $X$-syndromes, qubit errors and $Z$-syndromes respectively.
The CSS parity-check matrices of the corresponding HGP code $\operatorname{HGP}(h_1,h_2)$ can then be read off from \eqref{eq:HGP complex}:
\begin{subequations}\label{eq:HGP H_X, H_Z}
\begin{align}
    H_X &= \big( h^{}_1 \otimes \ident_{n_2} \;\big|\; \ident_{m_1} \otimes h^\transpose_2 \big)  \label{eq:HGP H_X} \\
    H_Z &= \big( \ident_{n_1} \otimes h^{}_2 \;\big|\; h^\transpose_1 \otimes \ident_{m_2} \big) \, . \label{eq:HGP H_Z}
\end{align}
\end{subequations}
The $X$-type and $Z$-type stabilizer checks mutually commute because $H^{}_X H^\transpose_Z = 2h^{}_1 \otimes h^\transpose_2 = 0$ over $\mathbb{F}_2$. The resulting quantum code $\operatorname{HGP}(h_1,h_2)$ has parameters \cite{Tillich_2014}
\begin{subequations}\label{eq:HGP code parameters}
\begin{align}
    n &= n_1n_2 + m_1m_2  \\
    k &= k^{}_1k^{}_2 + k^\transpose_1k^\transpose_2  \\
    d_Z &= \min\big(d^{}_1,d^\transpose_2\big)  \\
    d_X &= \min\big(d^{}_2,d^\transpose_1\big) \, ,
\end{align}
\end{subequations}
with the quantum code distance $d = \min(d_X,d_Z)$.
Geometrically, the HGP code can be arranged in a rectangular layout; see Fig. \ref{fig:HGP layout}. Adopting the naming convention from \cite{ReShape_decoder}, we call the $n^2$ qubits on the left side of \eqref{eq:HGP H_X, H_Z} \emph{left} qubits, $\mathrm{\Lambda}_{\rm L}$, and the other $m^2$ qubits \emph{right} qubits, $\mathrm{\Lambda}_{\rm R}$. In this form, it can be seen that the $Z$-type stabilizer checks consist of qubits from a single row of $\mathrm{\Lambda}_{\rm L}$ and a single column of $\mathrm{\Lambda}_{\rm R }$. $X$-type stabilizer checks correspondingly have qubits from a single row in $\mathrm{\Lambda}_{\rm R}$ and a single column in $\mathrm{\Lambda}_{\rm L}$. The 2D surface (toric) code can be viewed as a HGP code of two classical 1D repetition codes, each embedded on a line (ring). When $h_1$ and $h_2$ have full rank, their transpose codes are trivial ($k^\transpose_1 = k^\transpose_2 = 0$) and so the above code parameters simplify to $k = k_1k_2$ and $d = \min(d_1,d_2)$.

\begin{figure}
    \centering
    \includegraphics[width=0.4\textwidth]{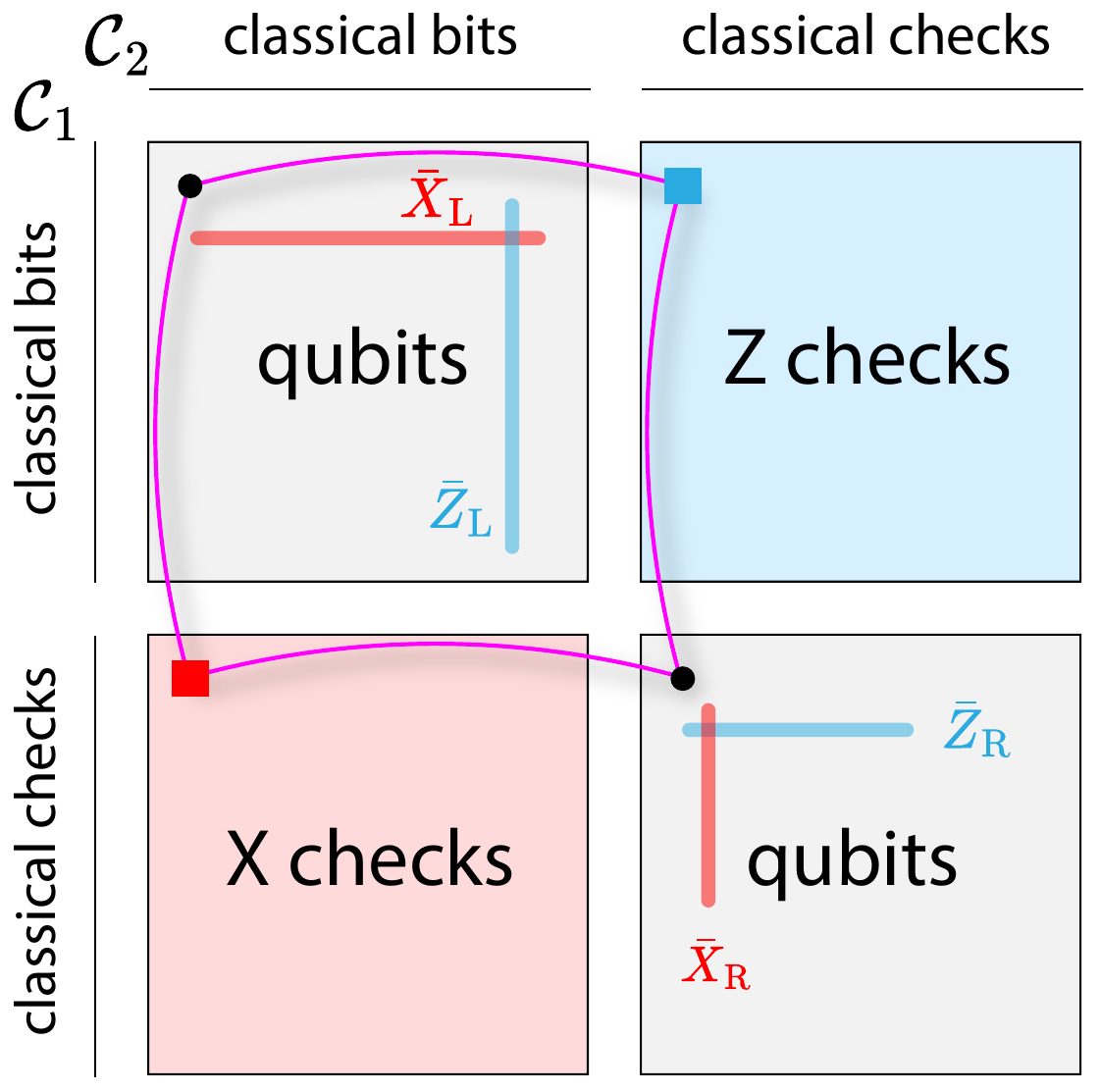} \hfill
    \includegraphics[width=0.55\textwidth]{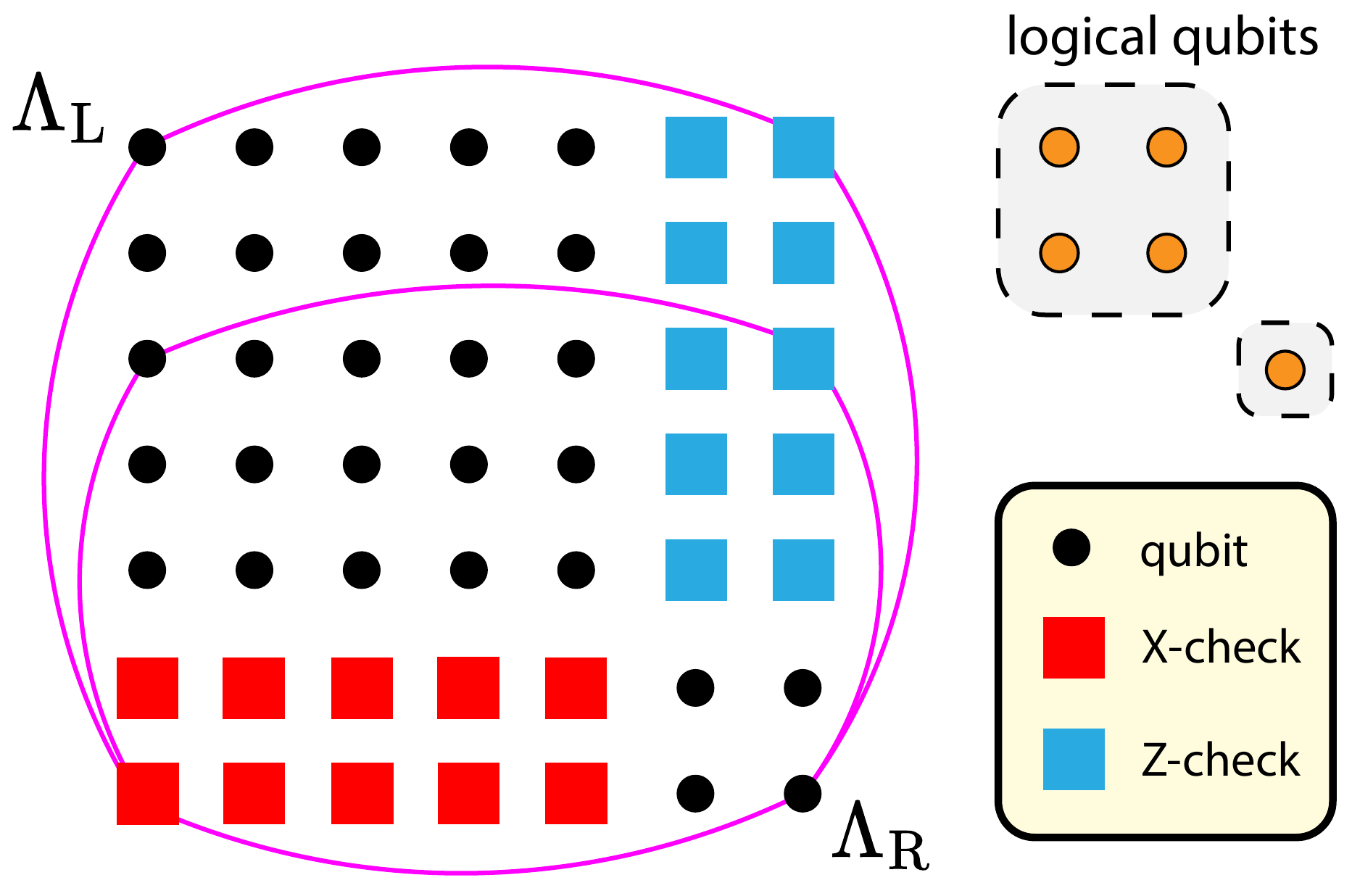}
    \caption{The rectangular layout of a hypergraph product code is depicted. Magenta edges denote a subset of qubit-check connections. The logical qubits are partitioned into those which reside in the left sector $\mathrm{\Lambda}_{\rm L}$ or right sector $\mathrm{\Lambda}_{\rm R}$ of data qubits.}
    \label{fig:HGP layout}
\end{figure}

From the classical parity-check matrices $h_1$ and $h_2$, we can also construct their associated generator matrices $g_1$ and $g_2$, which label a basis of their corresponding codewords and satisfy $hg^\transpose = 0$. In the HGP code, similar to the data qubits, the logical qubits can be partitioned into $k_1k_2$ left logical qubits and $k^\transpose_1k^\transpose_2$ right logical qubits. A canonical basis for the left logical $\overline{X}$ and $\overline{Z}$ operators of the HGP code can then be chosen as \cite{ReShape_decoder, Quintavalle_2023}
\begin{subequations}\label{eq:HGP G_X, G_Z}
\begin{align}
    G_{Z,\text{L}} &= \big( g_1 \otimes \{ \mathbf{e}_i \} \;|\; \mathbf{0} \big)  \label{eq:HGP G_Z}  \\
    G_{X,\text{L}} &= \big( \{ \mathbf{e}_j \} \otimes g_2 \;|\; \mathbf{0} \big)  \label{eq:HGP G_X}  \, ,
\end{align}
\end{subequations}
where $\{ \mathbf{e}_i \} \notin \im(h^\transpose_2)$ and $\{ \mathbf{e}_j \} \notin \im(h^\transpose_1)$ are $n$-dimensional unit vectors. This basis choice ensures that our different logical operators are not stabilizer equivalent.  Note that $\mathrm{rank}(h_1) = n_1-k_1$, and so $k_1$ instances of $\mathbf{e}_j$ are guaranteed to exist (and $k_2$ likewise for $\mathbf{e}_i$). Hence \eqref{eq:HGP G_X, G_Z} is a valid Pauli basis for all $k_1k_2$ left logical qubits. The canonical basis for the right logical qubits follows similarly but with the transpose codes. By putting the classical generator matrices into standard form, i.e. $g = (\ident_k \,|\, A)$, we can further ensure that logical $\overline{X}$ and $\overline{Z}$ operators either have no intersection or intersect at a single qubit indexed by $(i,j)$ in \eqref{eq:HGP G_X, G_Z}; we can use this single qubit intersection to label our $k_1k_2$ left logical qubits, see the top right corner of Figure \ref{fig:HGP layout}. Note that this further refinement is not strictly necessary since for any row of \eqref{eq:HGP G_X}, one can always find a corresponding combination of rows of \eqref{eq:HGP G_Z} which anticommutes with it while commuting with all other rows. Observe that the canonical logical $\overline{Z}$ operators are supported on single columns of $\mathrm{\Lambda}_{\rm L}$, and the canonical logical $\overline{X}$ operators are supported on single rows of $\mathrm{\Lambda}_{\rm L}$. The same holds true for the right logical qubits but on $\mathrm{\Lambda}_{\rm R}$. Note that in the canonical basis, left logical operators have minimum weight $\min(d_1,d_2)$ while right logical operators have minimum weight $\min(d^\transpose_1,d^\transpose_2)$. In fact, if we ignore the right logical qubits and treat them as gauge qubits, the code distance simplifies to $d = \min(d_1,d_2)$. This previous statement can be formalized by examining the stabilizer structure on the left and right sectors, which we relegate to Section \ref{sec:FT 2D HGP}.

For the classical codes we will be interested in, often $d^\transpose$ will be smaller than $d$, and so it will be beneficial to ignore the right logical qubits to not diminish code distance, at the cost of a reduced rate. For the remainder of this paper, we will focus on the left logical qubits unless stated otherwise.

\subsection{Higher-fold homological product codes}\label{sec:homological product codes}

In this section, we discuss the construction of higher-fold homological product codes as well as review some of their well-known facts. We note that there are two distinct constructions of homological product codes:
\begin{enumerate}
    \item Homological product between a quantum CSS code and a classical linear code
    \item Homological product between two quantum CSS codes
\end{enumerate}
We discuss these two constructions in the following two sections separately.

\subsubsection{Homological product: quantum $\times$ classical}

\begin{figure}[t]
    \centering
    \includegraphics[width=0.35\textwidth]{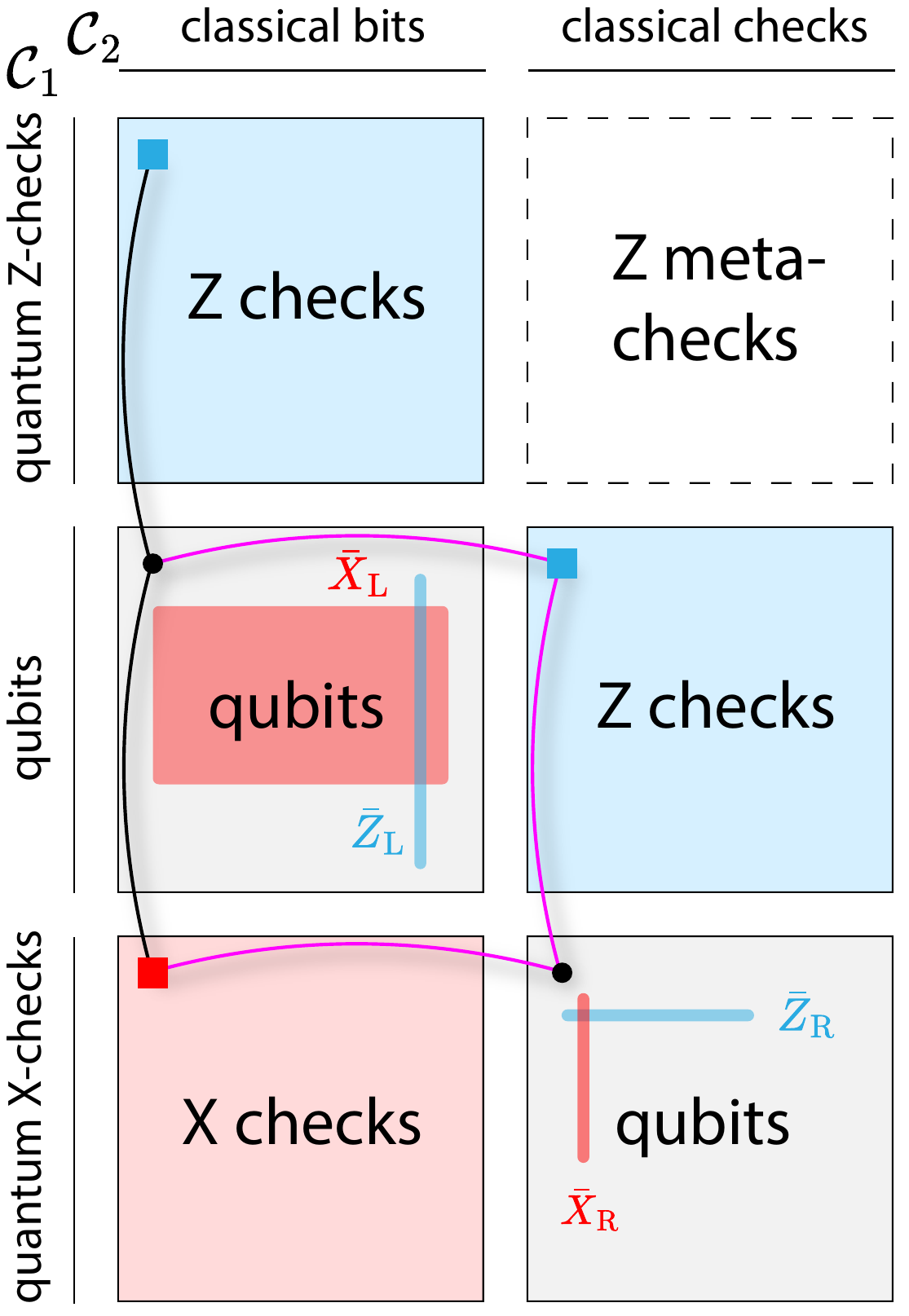}\hfill
    \includegraphics[width=0.6\textwidth]{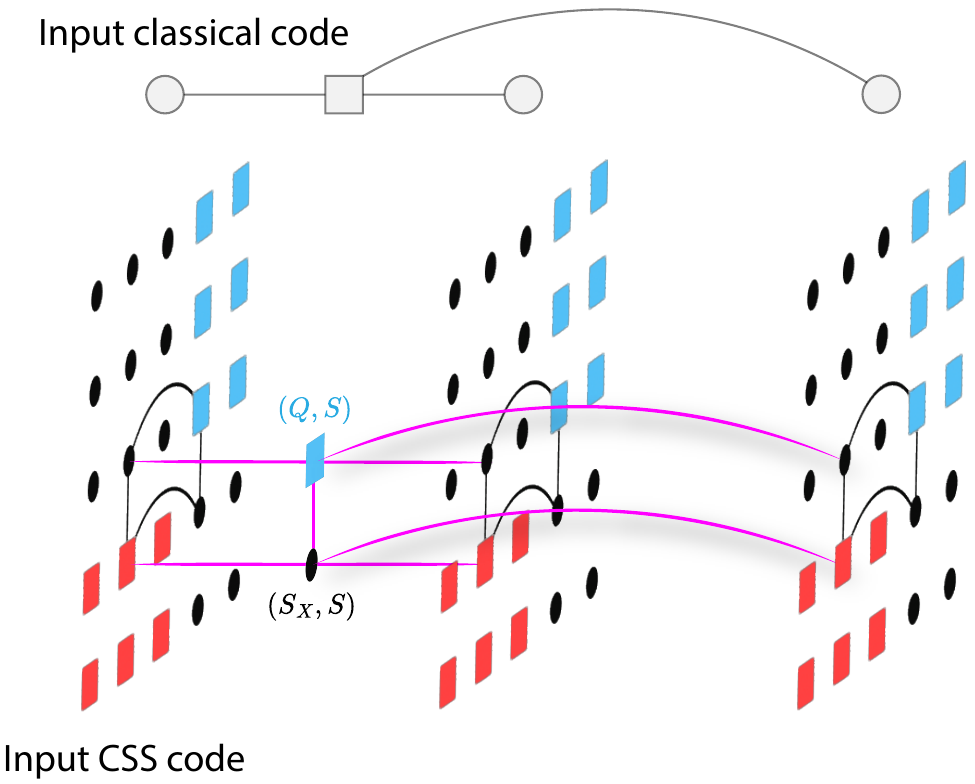}
    \caption{\textbf{Left:} The layout of a (quantum $\times$ classical) homological product code is depicted. \textbf{Right:} New qubits and checks are labeled according to their associated spaces in the product complex \eqref{eq:3D HGP complex}, and their new connections are colored in magenta.}
    \label{fig:3D HGP layout}
\end{figure}

For an input $\llbracket n_{\rm Q}, k_{\rm Q}, (d_X,d_Z) \rrbracket$ quantum CSS code with parity-check matrices $H_X, H_Z \in \mathbb{F}^{m_{\rm Q}\times n_{\rm Q}}_2$ and a $[n_{\rm c}, k_{\rm c}, d_{\rm c}]$ classical code with parity-check matrix $h \in \mathbb{F}^{m_{\rm c}\times n_{\rm c}}_2$, we can take the Kronecker product of their underlying chain complexes and associate the product complex with a \emph{homological product code} \cite{Freedman_2014, Bravyi_2014_HGP, Zeng_2019, Campbell_2019}, which can be thought of as higher-dimensional/fold generalizations of the usual hypergraph product. For example, the 3D surface code can be interpreted as a homological product of the 2D surface code with a 1D classical repetition code. Alternatively, it could also be viewed as a 3-fold hypergraph product of classical repetition codes i.e., a 3D hypergraph product code. Similar to the freedom in choosing which logical Pauli type extends to a membrane in the 3D surface code, we will have two choices for constructing the homological product code. Without loss of generality, we will restrict to the case where we extend the logical $\overline{Z}$ operators of the quantum code. The $\overline{X}$ case follows suit upon reversing the direction of interpretation of the product complex. The tensor product complex takes the form
\begin{equation}\label{eq:3D HGP complex}
\begin{tikzcd}
    & {(S_Z,S)} &&& {\tilde{R}_Z} \\
    {(S_Z,B)} && {(Q,S)} && {\tilde{S}_Z} \\
    {(Q,B)} && {(S_X,S)} && \tilde{Q} \\
    & {(S_X, B)} &&& {\tilde{S}_X}
    \arrow["{\ident \otimes h}", from=2-1, to=1-2]
    \arrow["{H_Z \otimes \ident}"', from=2-3, to=1-2]
    \arrow["{\tilde{M}_Z}", from=2-5, to=1-5]
    \arrow["{H_Z \otimes \ident}", from=3-1, to=2-1]
    \arrow["{\ident \otimes h}", from=3-1, to=2-3]
    \arrow["{H^\transpose_X \otimes \ident}"', from=3-3, to=2-3]
    \arrow["{\tilde{H}_Z}", from=3-5, to=2-5]
    \arrow["{H^\transpose_X \otimes \ident}", from=4-2, to=3-1]
    \arrow["{\ident \otimes h}"', from=4-2, to=3-3]
    \arrow["{\tilde{H}^\transpose_X}", from=4-5, to=3-5]
\end{tikzcd}
\end{equation}
and the new CSS parity-check matrices are given by
\begin{subequations}\renewcommand*{\arraystretch}{1.3}
\label{eq:3D HGP H_X,H_Z}
\begin{align}
    \tilde{H}_X &= \big(\, H_X \otimes \ident \;\big|\; \ident \otimes h^\transpose \,\big)  \label{eq:3D HGP H_X}  \\
    \tilde{H}_Z &= \left(\begin{array}{c|c}
        H_Z \otimes \ident \;\; & \mathbf{0}  \\
        \ident \otimes h & \; H^\transpose_X \otimes \ident
    \end{array}\right) \, .   \label{eq:3D HGP H_Z}
\end{align}
\end{subequations}
Geometrically, the Tanner graph of the homological product code resembles that of the Euclidean graph product between the two input Tanner graphs but with nodes relabeled according to \eqref{eq:3D HGP complex}; see Figure \ref{fig:3D HGP layout} for an illustration. When the input parity-check matrices have full rank, the parameters for the homological product code can be computed through the homology groups \cite{Zeng_2019}:
\begin{subequations}\label{eqn:homological_classical_code_params}
\begin{align} 
    \tilde{n} &= n_{\rm Q}n_{\rm c} + m_Zm_{\rm c}  \\
    \tilde{k} &= k_{\rm Q}k_{\rm c} \label{eqn:kunneth_example_1}  \\
    \tilde{d}_X &= d_X d_{\rm c} \\
    \tilde{d}_Z &= d_Z \, .
\end{align}
\end{subequations}

Note that in Eq.~\ref{eqn:kunneth_example_1}, we have used the K\"unneth formula stated in Proposition~\ref{prop:kunneth-formula} to compute $\tilde{k}$. Let us denote the (co)homology groups of the quantum code and classical code with $\mathcal{H}[{\rm Q}]$ and $\mathcal{H}[{\rm c}]$ respectively. Recall that we have adopted the convention for the logical $\overline{X}$ operator representatives of the CSS code to be given by the cohomology group $\mathcal{H}^1[{\rm Q}]$. By allowing the bits and checks of the classical code to be assigned to the 0-cells and 1-cells of the associated 2-term chain complex, we can observe that the codespace of the classical code is given by $\mathcal{H}^0[{\rm c}]$.
Thus, a straightforward application of the K\"unneth formula tells us that the resulting product complex has the following cohomology group:
\[\tilde{\mathcal{H}}^1 \cong \mathcal{H}^1[\text{Q}] \otimes \mathcal{H}^0[\text{c}].\]
Since the rank of $\tilde{\mathcal{H}}^1$ gives us the number of logical qubits in the homological product code, we can simply sum the product of the ranks of the cohomology groups of the constituent codes to give us Eq.~\ref{eqn:kunneth_example_1}.

From \eqref{eq:3D HGP H_X,H_Z}, we see that, similar to the HGP case, the qubits partition into a left and a right sector. Since the input parity-check matrices are full rank, we only have left-sector logical qubits. For these logical qubits, we can write a canonical basis for logical $\overline{X}$ and $\overline{Z}$ operators:
\begin{subequations}\label{eq:3D HGP G_X,G_Z}
\begin{align}
    \tilde{G}_{Z,{\rm L}} &= \left(\, G_Z \otimes \{ \mathbf{e}_i \} \;\big|\; \mathbf{0} \,\right) \label{eq:3D HGP G_Z} \\
    \tilde{G}_{X,{\rm L}} &= \left(\, G_X \otimes g \;\big|\; \mathbf{0} \,\right)  \, , \label{eq:3D HGP G_X}
\end{align}
\end{subequations}
where $G_X$ and $G_Z$ are a symplectic logical basis for the input quantum code and $\{ \mathbf{e}_i \} \notin \mathrm{rs}(h)$. If the input quantum CSS code is a HGP code, then we can designate $G_Z$ and $G_X$ as the canonical logical basis \eqref{eq:HGP G_X, G_Z} as well and index all logical qubits with 3-dimensional coordinate triples $(i,j,k)$.

We point out that our construction of the homological product code described above has implicitly assumed that the quantum CSS code corresponds to a 3-term chain complex and that the parity-check matrix $h$ of the classical code has full rank. However, in general, it is known that we are able to associate a quantum CSS code to \emph{any} 3 consecutive terms in a chain complex, which in its entirety could consist of more than 3 terms; the same goes for the classical code but with 2 terms. Notable examples include some of the quantum CSS codes that are amenable to single-shot error correction constructed by Campbell~\cite{Campbell_2019}. These CSS codes are constructed from the middle 3 terms in a 5-term chain complex as shown below:
\[\mathcal{C} = \left(R_Z \xrightarrow[]{\partial_3 = M_Z^\transpose} S_Z \xrightarrow[]{\partial_2 = H_Z^\transpose} Q \xrightarrow[]{\partial_1 = H_X} S_X \xrightarrow[]{\partial_{0} = M_X} R_X\right)\]
where we have associated the constitutent vector spaces $C_{3}, C_2, C_1, C_{0}, C_{-1}$ with $R_Z, S_Z, Q, S_X, R_X$ i.e., the space of relations between $Z$ checks, the $Z$ checks, qubits, $X$ checks, and the relations between $X$ checks.
A significant advantage of longer chain complexes is the appearance of ``metacheck'' matrices $M_Z$ and $M_X$. These metacheck matrices can be interpreted as additional ``codes'' on the error syndromes, which can be used to detect and correct syndrome measurement errors.
In this case, when we take the homological product between such CSS codes with a classical code, we end up with a 6-term chain complex. When we apply Proposition~\ref{prop:kunneth-formula} to evaluate $\tilde{\mathcal{H}}^1$ for the new product complex, we now observe contributions from summands other than $\mathcal{H}^1[{\rm Q}] \otimes \mathcal{H}^0[{\rm c}]$ that include
\begin{align}
    \mathcal{H}^0[{\rm Q}] \otimes \mathcal{H}^1[{\rm c}] = \frac{\ker \delta_0[{\rm Q}]}{\im \delta_{-1}[{\rm Q}]} \otimes \frac{\ker \delta_1[{\rm c}]}{\im \delta_0[{\rm c}]} = \frac{\ker H_X^\transpose}{\im M_X^\transpose} \otimes \frac{\ker h^\transpose}{\im m^\transpose} \, .
\end{align}
Formerly, $M_X^\transpose$ did not exist when we only had a 3-term chain complex and $\im m^\transpose$ was trivial when the classical code had no redundant checks.
The additional contribution from the other summands implies that we have additional logical qubits that arise from these additional cohomology group elements that we call ``spurious'' cohomologies.
Because the structure of these spurious cohomologies depends on the redundancy of checks in both the quantum and classical codes, the cosystolic distances for these spurious cohomologies can vary.
By choosing a canonical basis for these spurious cohomologies as well as the cohomologies that we hope to keep, i.e., $\mathcal{H}^1[Q] \otimes \mathcal{H}^0[c]$, we can show that we can treat the spurious cohomologies as gauge logical qubits that do not deprecate the distance of our resulting quantum code even when they dress our other logical operators. The intuition is similar to the case with the HGP: the gauge operators corresponding to the spurious cohomologies are supported in a different subsystem of the data qubits compared to the logical operators.
We provide a more detailed analysis for this phenomenon in Lemma~\ref{lem:dressed-logical-ops} in Section~\ref{sec:fault tolerance}.

\subsubsection{Homological product: quantum $\times$ quantum}
\label{sec:quantumxquantum}

Suppose we are given two quantum CSS codes with parameters that are denoted as $\llbracket n_{\rm Q}, k_{\rm Q}, (d_X,d_Z) \rrbracket$ and $\llbracket n'_{\rm Q}, k'_{\rm Q}, (d'_X,d'_Z) \rrbracket$ and corresponding parity-check matrices $H_X, H_Z, H'_X, H_Z'$.
Similar to the case for the homological product between a quantum CSS code and a classical linear code, we can take the Kronecker product of their underlying chain complexes and associate the product complex with a (quantum $\times$ quantum) \emph{homological product code} \cite{Freedman_2014, Bravyi_2014_HGP}. Because each quantum CSS code is associated to a 3-term chain complex (from our simplifying assumption discussed in the previous section), the product complex now becomes a 5-term chain complex. The 5-term chain gives us a some flexibility to choose which set of three consecutive vector spaces to place the $X$-syndromes, qubits, and $Z$-syndromes. For the purpose of reducing the number of cases that we have to analyze, we assume that the two 3-term chain complexes associated to the two quantum CSS codes have the following form:
\[\mathcal{C} = \left(S_Z \xrightarrow[]{\partial_{1} = H_Z^\transpose} Q \xrightarrow[]{\partial_{0} = H_X} S_X\right),\]
where $C_{1}, C_{0}, C_{-1}$ are the vector spaces for the $Z$ checks, qubits, and $X$ checks respectively. After taking the homological product of the 3-term chain complexes, we let the new $Z$ checks, qubits, and $X$ checks be associated to the vector spaces $\tilde{C}_{1}, \tilde{C}_{0}, \tilde{C}_{-1}$ of the product complex $\tilde{\mathcal{C}}$.
The product complex takes the form
\begin{equation}\label{eq:4D HP complex}
\begin{tikzcd}
    & {(S^{}_Z,S'_Z)} &&& {\tilde{R}_Z} \\
    {(S^{}_Z,Q')} && {(Q,S'_Z)} && {\tilde{S}_Z} \\
    {(S^{}_Z,S'_X)} & {(Q,Q')} & {(S^{}_X,S'_Z)} && \tilde{Q} \\
    {(Q, S'_X)} && {(S_X, Q')} && {\tilde{S}_X} \\
    & {(S^{}_X, S'_X)} &&& {\tilde{R}_X} 
    \arrow["{\ident \otimes H'_Z}", from=2-1, to=1-2]
    \arrow["{H_Z \otimes \ident}"', from=2-3, to=1-2]
    \arrow["{\ident \otimes H_X^{\prime\transpose}}", from=3-1, to=2-1]
    \arrow["{\ident \otimes H'_Z}", from=3-2, to=2-3]
    \arrow["{H_Z \otimes \ident}"', from=3-2, to=2-1]
    \arrow["{H^\transpose_X \otimes \ident}"', from=3-3, to=2-3]   
    \arrow["{H_Z \otimes \ident}", from=4-1, to=3-1]
    \arrow["{\ident \otimes H'_Z}"', from=4-3, to=3-3]
    \arrow["{\ident \otimes H^{\prime\transpose}_X}"', from=4-1, to=3-2]
    \arrow["{H_X^\transpose \otimes \ident}", from=4-3, to=3-2]
    \arrow["{H_X^\transpose \otimes \ident}", from=5-2, to=4-1]
    \arrow["{\ident \otimes H_X^{\prime\transpose}}"', from=5-2, to=4-3]
    \arrow["{\tilde{M}_Z}", from=2-5, to=1-5]
    \arrow["{\tilde{H}_Z}", from=3-5, to=2-5]
    \arrow["{\tilde{H}^\transpose_X}", from=4-5, to=3-5]
    \arrow["{\tilde{M}^\transpose_X}", from=5-5, to=4-5]
\end{tikzcd}
\end{equation}
and the new CSS parity-check matrices are
\begin{subequations}\renewcommand*{\arraystretch}{1.3}
\label{eq:4D HP H_X,H_Z}
\begin{align}
    \tilde{H}_X &= \left(\begin{array}{c|c|c}
        H_Z^\transpose \otimes \ident \;\; & \;\ident \otimes H'_X \;\;& \mathbf{0}  \\
        \mathbf{0} \;\; & H_X \otimes \ident \;\; & \; \ident \otimes H_Z^{\prime\transpose}
    \end{array}\right), \,   \label{eq:4D HP H_X}  \\
    \tilde{H}_Z &= \left(\begin{array}{c|c|c}
        \ident \otimes H_X^{\prime\transpose} \;\; & H_Z \otimes \ident \;\;& \mathbf{0}  \\
        \mathbf{0} \;\; & \;\ident \otimes H'_Z \;\; & \; H_X^\transpose \otimes \ident
    \end{array}\right) \, ,   \label{eq:4D HP H_Z}
\end{align}
\end{subequations}
where the three partitions in the parity-check matrices correspond to the physical qubits in the subspaces $(S^{}_Z, S'_X), (Q, Q'), (S^{}_X, S'_Z)$ respectively; see Figure \ref{fig:4D HGP layout} for an illustration. We refer to these subsets of physical qubits as left (L), middle (M), and right (R) respectively.
Focusing on the middle logical qubits, we can write a canonical basis for logical $\overline{X}$ and $\overline{Z}$ operators:
\begin{subequations}\label{eq:4D HP G_X,G_Z}
\begin{align}
    \tilde{G}_{Z,{\rm M}} &= \left(\, \mathbf{0} \;\big|\; G^{}_Z \otimes G'_Z \;\big|\; \mathbf{0} \,\right), \label{eq:4D HGP G_Z} \\
    \tilde{G}_{X,{\rm M}} &= \left(\, \mathbf{0} \;\big|\; G^{}_X \otimes G'_X \;\big|\; \mathbf{0} \,\right).  \label{eq:4D HGP G_X}
\end{align}
\end{subequations}
One can straightforwardly verify that this basis is symplectic when the input bases are symplectic. We defer the discussion of the exact code parameters for the homological product code to Section~\ref{sec:code-params-quantum-times-quantum} because the code distances for such codes are not sufficiently well understood in the most general case.

\begin{figure}[t]
    \centering
    \includegraphics[width=0.55\textwidth]{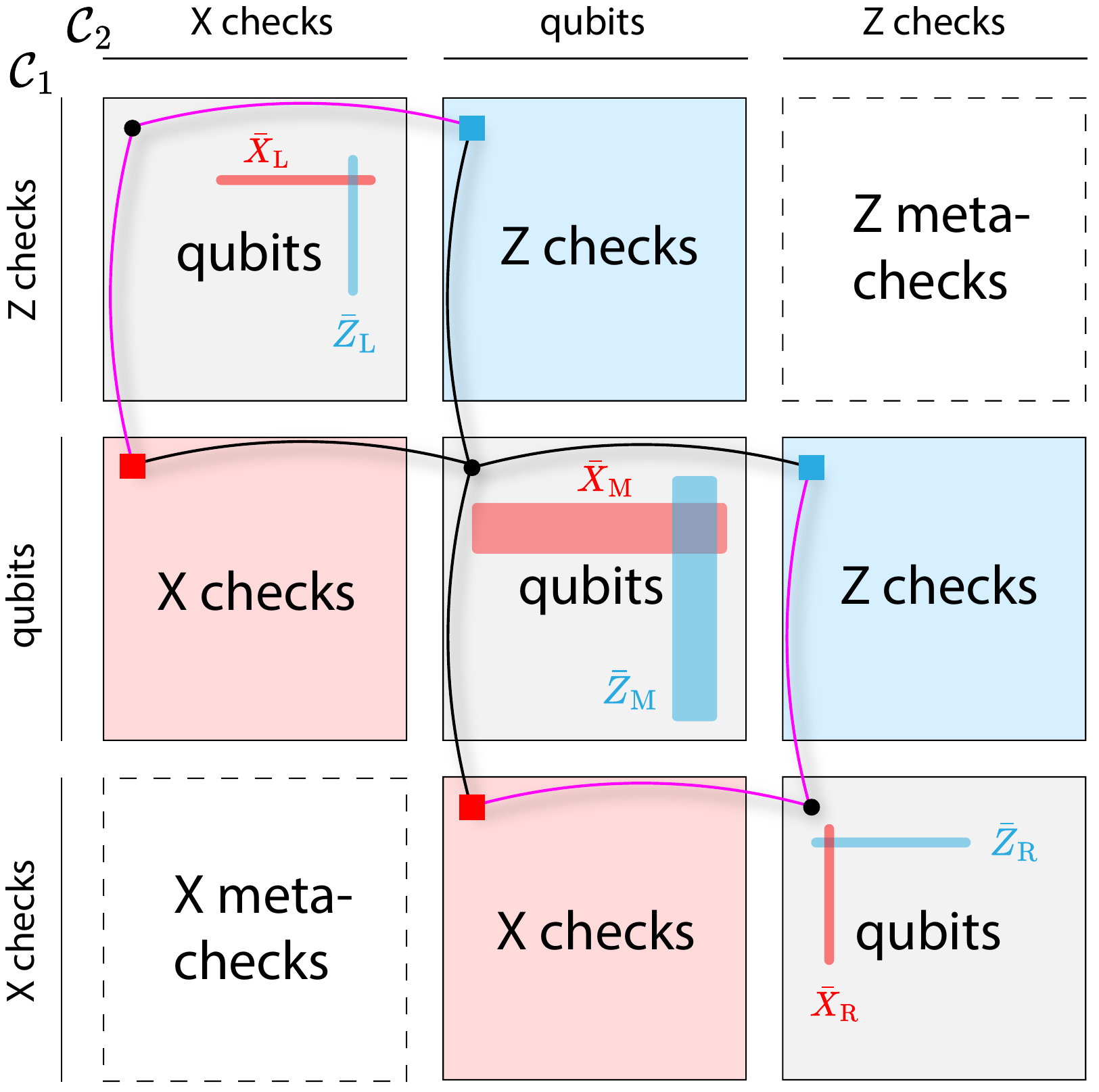}
    \caption{The layout of a (quantum $\times$ quantum) homological product code is depicted. New qubit-check connections are colored in magenta.}
    \label{fig:4D HGP layout}
\end{figure}


\section{Generic bounds and limitations}
\label{sec:bounds}

In this section, we review some limitations of automorphism gates and then derive some broad results that will guide us in tailoring classical codes to our formalism later on in Section \ref{sec:classical families}. We note some of these results are likely ``folklore'' within the error-correction community, but we state them explicitly for completeness.

\subsection{Limitations on automorphisms of classical linear codes}

Recall the classical automorphism conditions:
\begin{subequations}\label{eq:aut conditions}
\begin{align}
    H\sigma &= WH  \label{eq:aut condition h}  \\
    G\sigma &= VG  \label{eq:aut condition g}
\end{align}
\end{subequations}
where $H$ is the classical parity check matrix, $G$ is the generator matrix, and $\sigma$ is a permutation of the physical bits. Also recall that the affine class refers to functions that represent affine transformations over $\mathbb{F}_2$, any of which can be expressed using only the XOR logical gate and constants. From this definition, we can see that the CNOT gate generates all affine transformations~\cite{aaronson2015gates}.

\begin{thm}[Logical action restricted to affine class]
\label{thm:affine restriction}
    Suppose we have a linear code $\mathcal{C}$. The logical action of any element of $\Aut(\mathcal{C})$ is limited to the affine class.
\end{thm}

\begin{proof}
    Let $G \in \mathbb{F}^{k\times n}_2$ be a generator matrix for $\mathcal{C}$, and let $U \in \mathbb{F}^{n\times n}_2$ be a permutation matrix corresponding to some element in $\Aut(\mathcal{C})$. Then the definition of $\Aut(\mathcal{C})$ requires $GU = WG$ for some invertible $W \in \mathbb{F}^{k\times k}_2$. Since $W$ is invertible, its logical action is reversible. Furthermore, $W$ can be decomposed into elementary row operations consisting of row addition and swaps. Row addition is equivalent to logical CNOT gates, with control and target defined by which row is replaced. Row swapping is equivalent to logical SWAPs, which can be decomposed into three CNOT gates. Thus, since $W$ can be fully decomposed into logical CNOT gates, its logical action falls into the affine class of classical reversible gates.
\end{proof}

An immediate corollary is that we cannot achieve a universal set of logical reversible gates on linear codes from physical permutations alone, since the affine class is nonuniversal \cite{aaronson2015gates}.

\begin{cor}[Classical automorphism gates are nonuniversal]
    For any linear code $\mathcal{C}$, the group of logical automorphism gates $\mathcal{A}(\mathcal{C})$ cannot form a universal set.
\end{cor}

Another consequence of Theorem \ref{thm:affine restriction} is that the number of distinct logical operations achieved by physical permutations is upper bounded by the number of invertible matrices $W$, which is the size of the general linear group of degree $k$ over $\mathbb{F}_2$, $\mathrm{GL}_k(\mathbb{F}_2)$. 
Note that the determinant map is trivial over $\mathbb{F}_2$, and so we also have $\mathrm{GL}_k(\mathbb{F}_2) \simeq \mathrm{SL}_k(\mathbb{F}_2) \simeq \mathrm{PSL}_k(\mathbb{F}_2)$ since the only scalars are 0 and 1 in $\mathbb{F}_2$.

\begin{cor}[Upper bound on number of logical gates]
\label{cor:upper bound gates}
    For any linear code $\mathcal{C}$ encoding $k$ logical bits, we have $\abs{\mathcal{A}} \leq \abs{\mathrm{GL}_k(\mathbb{F}_2)}$.
\end{cor}

Suppose we desire a linear code with $\abs{\mathcal{A}} \propto \abs{\mathrm{GL}_k(\mathbb{F}_2)}$, i.e. approaching the upper bound of Corollary \ref{cor:upper bound gates}. Since the total number of distinct permutations of $n$ bits is $\abs{\mathrm{S}_n} = n!$, for a given code dimension $k$, we can derive a lower bound on the minimum code length $n$, or alternatively an upper bound on $k$ for given $n$, by comparing to the size of $\mathrm{GL}_k(\mathbb{F}_2)$.

\begin{cor}\label{cor:k upper bound GL(k,2)}
    For any linear code with $\abs{\mathcal{A}} \propto \abs{\mathrm{GL}_k(\mathbb{F}_2)}$, we must have $k = O\big(\sqrt{n\log n}\big)$.
\end{cor}

\begin{proof}
    Every $W \in \mathrm{im}(\phi) \simeq \mathcal{A}$ needs to have a distinct preimage in $\Aut(\mathcal{C})$ for $\phi$ to be a valid homomorphism. Thus, we have $\abs{\mathcal{A}} = \abs{\Aut(\mathcal{C}) / \ker(\phi)} \leq \abs{\Aut(\mathcal{C})} \leq \abs{\mathrm{S}_n}$ with the first inequality saturated if and only if $\phi$ is injective, i.e. $\ker(\phi) = \ident$. The asymptotic size of $\mathrm{S}_n$ is $n! = 2^{O(n\log n)}$ by Stirling's formula, and the asymptotic size of $\mathrm{GL}_k(\mathbb{F}_2)$ is $\mathrm{\Theta}\big(2^{k^2}\big)$. Equating the two group sizes, we arrive at the statement of the Corollary.
\end{proof}

We note that Corollary \ref{cor:k upper bound GL(k,2)} is likely loose, since it has been shown that all equidistant linear codes are sequences of dual-Hamming codes \cite{Bonisoli_1984}, which have $k = O(\log n)$. The equidistant property is a necessary condition for $\mathcal{A} = \mathrm{GL}_k(\mathbb{F}_2)$ since permutations preserve the weight of codewords.

When the code distance is strictly greater than 2, we can show that $\mathcal{A}(\mathcal{C})$ is contained in $\mathcal{A}(\mathcal{C}^\perp)$, which we will denote $\mathcal{A}^\perp$ for short.

\begin{thm}[Dual automorphism bound]
\label{thm:dual aut bound}
    For any linear code $\mathcal{C}$ with code distance $d\geq 3$, we have $\mathcal{A} \subseteq \mathcal{A}^\perp \simeq \mathrm{Aut}(\mathcal{C}^\perp)$.
\end{thm}

\begin{proof}
    Recall that for any linear code, we have $\mathrm{Aut}(\mathcal{C}) \simeq \mathrm{Aut}(\mathcal{C}^\perp) = \mathrm{Aut}(h)$ by default. So in order to show that $\mathcal{A}$ is contained in $\mathcal{A}^\perp$, we essentially need to show that there cannot exist two distinct physical permutations $\sigma,\sigma'$ that correspond to the same row-transformation $w$ on $h$. We will prove this statement by contradiction. Suppose we have two distinct permutations $\sigma,\sigma' \in \mathrm{S}_n$ with $h\sigma = h\sigma' = wh$. Define the permutation matrix $\sigma'' = \sigma'\sigma^\transpose \neq \ident$. Then we have $h\sigma'' = w^{-1}w h = h$, and upon rearranging, $h(\ident + \sigma'') = 0$, which implies that the columns of $\ident+\sigma''$ are in $\ker h = \mathcal{C}$. However, the weight of any column in $\ident+\sigma''$ is at most 2, which contradicts our assumption that the code distance $d\geq2$. Thus, distinct $\sigma,\sigma' \in \mathrm{Aut}(\mathcal{C}^\perp)$ must correspond to distinct $w,w' \in \mathcal{A}^\perp$, which shows that $\mathrm{Aut}(\mathcal{C}^\perp) \subseteq \mathcal{A}^\perp$. Since $\mathcal{A}^\perp \subseteq \mathrm{Aut}(\mathcal{C}^\perp)$ by definition, we conclude that $\mathcal{A}^\perp \simeq \mathrm{Aut}(\mathcal{C}^\perp)$, i.e. $\ker\phi^\perp$ is trivial. Combining this last statement with the fact that $\mathcal{A} \subseteq \mathrm{Aut}(\mathcal{C}) \simeq \mathrm{Aut}(\mathcal{C}^\perp)$, we arrive at the statement of the Theorem.
\end{proof}

By swapping the roles of $\mathcal{C}$ and $\mathcal{C}^\perp$ in Theorem \ref{thm:dual aut bound}, we can achieve a strict equality if the dual distance also satisfies $d^\perp \geq 3$. In other words, both $\ker\phi$ and $\ker\phi^\perp$ are trivial.

\begin{cor}[Automorphism equivalence]
\label{cor:dual aut equal}
    If both the minimum distance and dual distance are greater than 2, then $\mathrm{Aut}(\mathcal{C}) \simeq \mathcal{A}(\mathcal{C}) \simeq \mathcal{A}(\mathcal{C}^\perp) \simeq \mathrm{Aut}(\mathcal{C}^\perp)$.
\end{cor}

As we will see later in the examples (Section \ref{sec:classical families}), it will sometimes be easier to characterize $\mathcal{A}^\perp$ or $\mathrm{Aut}(\mathcal{C}^\perp)$ than $\mathcal{A}(\mathcal{C})$ itself, which provides us with a convenient way to compute $\mathcal{A}(\mathcal{C})$ when both $d,d^\perp \geq 3$.

\subsection{Limitations on automorphisms of simply connected graphs}

In Section~\ref{sec:code-automorphisms}, we briefly discussed a subgroup of code automorphisms known as Tanner graph automorphisms, which will be useful for our automorphism gadgets later on in Section~\ref{sec:aut gadgets}.
In this section, we review a well-known limitation for the automorphisms of simple, connected graphs before developing a bound for the size of the automorphism group of a popular family of simple, connected graphs. These bounds will constrain the automorphism groups of codes whose Tanner graphs are isomorphic to simple, connected graphs.
We only work with finite, simple, connected graphs that have vertex degree at least 2.

We begin by stating some important definitions regarding path automorphisms, $s$-transitive graphs and $s$-regular graphs.

\begin{defn}[Path automorphism]
    Let paths $P$ and $Q$ of some graph $\mathsf{G}$ be defined as two ordered sets:
    \[P = \left\{v_{p_0}, e_{p_1}, v_{p_1}, \ldots, e_{p_n}, v_{p_{n+1}}\right\},\]
    \[Q = \left\{v_{q_0}, e_{q_1}, v_{w_1}, \ldots, e_{q_n}, v_{q_{n+1}}\right\},\]
    where $v_{p_i}, v_{q_j} \in V(\mathsf{G})$ and $e_{p_i}, e_{q_j} \in E(\mathsf{G})$.
    Then, an automorphism of $\mathsf{G}$ such that $f(P) = Q$ is defined such that $f$
maps the $i$\textsuperscript{th} term of $P$ onto the $i$\textsuperscript{th} term of $Q$, for each index $i$. 
\end{defn}

\begin{defn}[$s$-transitive graph~{\cite[Section 7.5]{tutte1966connectivity}}]
    A graph $\mathsf{G}$ is $s$-transitive, for a given positive integer $s$, if it has at least one path of length $s$ and if for each ordered pair $(P, Q)$ of paths of length $s$ of $\mathsf{G}$, there exists an automorphism $f$ of $\mathsf{G}$ such that $f(P) = Q$.
\end{defn}

Note that when $s = 0$, the graph $\mathsf{G}$ is vertex-transitive i.e., there exists some graph automorphism $f$ such that $f(u) = v$ for any pair of vertices $u, v \in V(\mathsf{G})$. In addition, it is known that an $s$-transitive graph is also a $(s-1)$-transitive graph for $s \geq 1$ as long as the graph has vertex degree at least 2.

\begin{defn}[$s$-regular graph~{\cite[Section 7.7]{tutte1966connectivity}}]\label{defn:s-regular}
    A graph $\mathsf{G}$ is $s$-regular if it is $s$-transitive and for each ordered
pair $(P, Q)$ of paths of length $s$, there is only one automorphism $f$ of $\mathsf{G}$ such that $f(P) = Q$. 
\end{defn}

We emphasize that this is different from the conventional definition of regularity in graphs, where $k$-regular means that each vertex has the same degree $k$. Here, regularity is a statement on the number of automorphisms per a pair of paths $(P,Q)$. 

We are particularly interested in cubic graphs where each vertex has degree 3 because we are interested in the regime where the base code has weight-3 checks so that the quantum code will similarly have small constant-size weight checks when taking products of such codes. In addition, cubic graphs are very well-studied and a lot of work has been done to understand their automorphism groups. The following theorems are useful in limiting the girth and the size of the automorphism group for simply connected graphs with the above properties.

\begin{thm}[Girth bound by regularity~{\cite[Theorem III]{tutte1947family}}]
    Let $\mathsf{G}$ be a $s$-regular cubic graph. Then, the girth $g$ of $\mathsf{G}$ is bounded by the following:
    \[g \geq 2s-2.\]
\end{thm}

Recall that the girth of the graph is related to the distance of the classical code associated to the graph. This implies that the distance of the code is lower bounded by the regularity of its underlying graph.

\begin{thm}[Upper bound on transitivity~{\cite[Theorem 18.6]{biggs1993algebraic}}]\label{thm:transitivity-upper-bound}
    If a graph $\mathsf{G}$ is simply connected, cubic, and $t$-transitive, then $t \leq 5$.
\end{thm}

\begin{thm}[Size of automorphism group of $s$-regular cubic graphs~{\cite[Theorem VI]{tutte1947family}}]\label{thm:s-regular-cubic-graph-automorphism}
    Let $\mathsf{G}$ be a $s$-regular cubic graph with $n$ vertices. Then, 
    \[\left|\mathrm{Aut}(\mathsf{G})\right| = 2^{s-1} \cdot 3n.\]
\end{thm}

Given the above two theorems, we now provide a useful theorem for upper bounding the size of the automorphism group of a simple, connected, cubic, $s$-regular, and $t$-transitive graph with $n$ vertices. We emphasize that the bound is independent of $s$ and $t$.

\begin{thm}[Upper bound on size of automorphism group of simply connected cubic graphs]\label{thm:upper-bound-regular-graph}
Let $\mathsf{G}$ be a simply connected, cubic, $s$-regular, $t$-transitive graph with $n$ vertices. Then,
    \[\left|\mathrm{Aut}(\mathsf{G})\right| \leq 48n.\]
\end{thm}
\begin{proof}
    Thm.~\ref{thm:transitivity-upper-bound} tells us that the simply connected, cubic, regular graph $\mathsf{G}$ can be at most $5$-transitive, implying that $\mathsf{G}$ can be at most $5$-regular by Def.~\ref{defn:s-regular}. From a straightforward application of Thm.~\ref{thm:s-regular-cubic-graph-automorphism}, we obtain the theorem statement.
\end{proof}

Based on Thm.~\ref{thm:upper-bound-regular-graph}, we see that the size of the automorphism group of a code with an underlying graph that is cubic, simply connected, $s$-transitive, and $t$-transitive can at most scale linearly with the number of checks.
This implies that codes constructed from such graphs are not great candidates in the asymptotic limit if we are interested in having a large number of automorphisms.

\subsection{Limitations on automorphisms of quantum CSS codes}

For a classical linear code, an automorphism of its codespace automatically implies an automorphism of its dual. For quantum stabilizer codes, logical Pauli operators related by an element of the stabilizer act equivalently on the codespace, and so distinct logical operations are grouped into equivalence classes under the code's stabilizer group. As a consequence, for a given choice of logical basis operators, there is not a unique ``generator matrix'' and so the classical automorphism condition \eqref{eq:aut condition}, which is defined with respect to a generator matrix, is not a useful test for automorphism. What is useful, on the other hand, is the dual automorphism condition \eqref{eq:dual aut condition}, which naturally takes into account stabilizer degeneracy. There are some additional subtleties related to notions of strong and weak automorphisms for which we refer the reader to \cite{hao2021} for more details. In this work, we focus on CSS codes, whose analyses is much simpler and does not contain such subtleties. Since physical permutations only map Pauli strings to other Pauli strings, it is clear that the logical automorphism group for any quantum stabilizer code is strictly contained in the Clifford group. Furthermore, since permutations cannot change the Pauli type, for a CSS code the logical automorphism group is contained within the affine group, and so Corollary \ref{cor:upper bound gates} naturally applies as well.

A recent work by Guyot and Jaques provides generic upper bounds and limitations on the number of distinct automorphism gates $\abs{\mathcal{A}}$ for any CSS code \cite{guyot2025}. When we use a CSS code as input to the homological product, the number of inherited automorphism gadgets will be constrained by their result. We quote their relevant results here and encourage the reader to explore their paper for the technical details and proofs.

\begin{thm}[Upper bound on distinct automorphism gates for CSS codes; Theorem 2 of \cite{guyot2025}]
    Let $C = \mathrm{CSS}(A,B)$ be an $\llbracket n,k,d\rrbracket$ code with $d \ge 3$. The number of distinct logical operations that can be implemented by permuting qubits in the code is upper-bounded by $\abs{\mathcal{A}} \leq \frac{n!}{k_{\max}!}$, where $k_{\max} = \max(\dim(A), \dim(B)).$
\end{thm}

\begin{cor}[Corollary 5 of \cite{guyot2025}]
    The following logical gates are not permutation addressable on CSS codes with the following rates:
    \begin{itemize}
        \item SWAP gates for codes with an asymptotic rate greater than $1/3$
        \item Any 2-qubit gate for codes with with an asymptotic rate greater than $3/4$
        \item CNOT gates for codes with an asymptotic rate of $\mathrm{\Omega}\left( 
        \sqrt{\frac{\log n}{n}} \right)$
    \end{itemize}
\end{cor}

Since all the product constructions that we use produce CSS codes, they will all be constrained by these upper bounds. However, we will see that this generic upper bound can be loose for the product constructions, which demand additional constraints on the input codes for there to be nontrivial code automorphisms.


\section{Automorphism gadgets} \label{sec:aut gadgets}

In this section, we derive our automorphism gadgets for all the homological product constructions mentioned in Section \ref{sec:HGP codes} and \ref{sec:homological product codes}. We will determine both the physical actions on the data qubits as well as the encoded actions on the logical qubits.
Note that for CSS codes, we only need to analyze the $X$-sector because commutation conditions will completely determine the corresponding transformation on the $Z$-sector. Specifically, for an invertible matrix $U \in \mathrm{GL}_n(\mathbb{F}_2)$ whose columns tabulate the transformation of all $n$ $X$-type Paulis, $(U^{-1})^\transpose = (U^\transpose)^{-1} \equiv U^{-\transpose}$ determines the corresponding action on the $n$ $Z$-type Paulis. This is because any Clifford transformation should preserve the symplectic norm $\Omega = \begin{pmatrix}
    \mathbf{0} & \ident_{n} \\ \ident_{n} & \mathbf{0}
\end{pmatrix}$ i.e. $A\Omega A^\transpose=\Omega$, where $A$ is a $2n\times2n$ matrix with $U$ and $U^{-\transpose}$ on the block diagonal. Preserving the symplectic norm $\Omega$ is equivalent to preserving the Pauli algebra on the qubits.

\subsection{Hypergraph product codes}

\begin{figure}[t]
    \centering
    \includegraphics[width=0.95\textwidth]{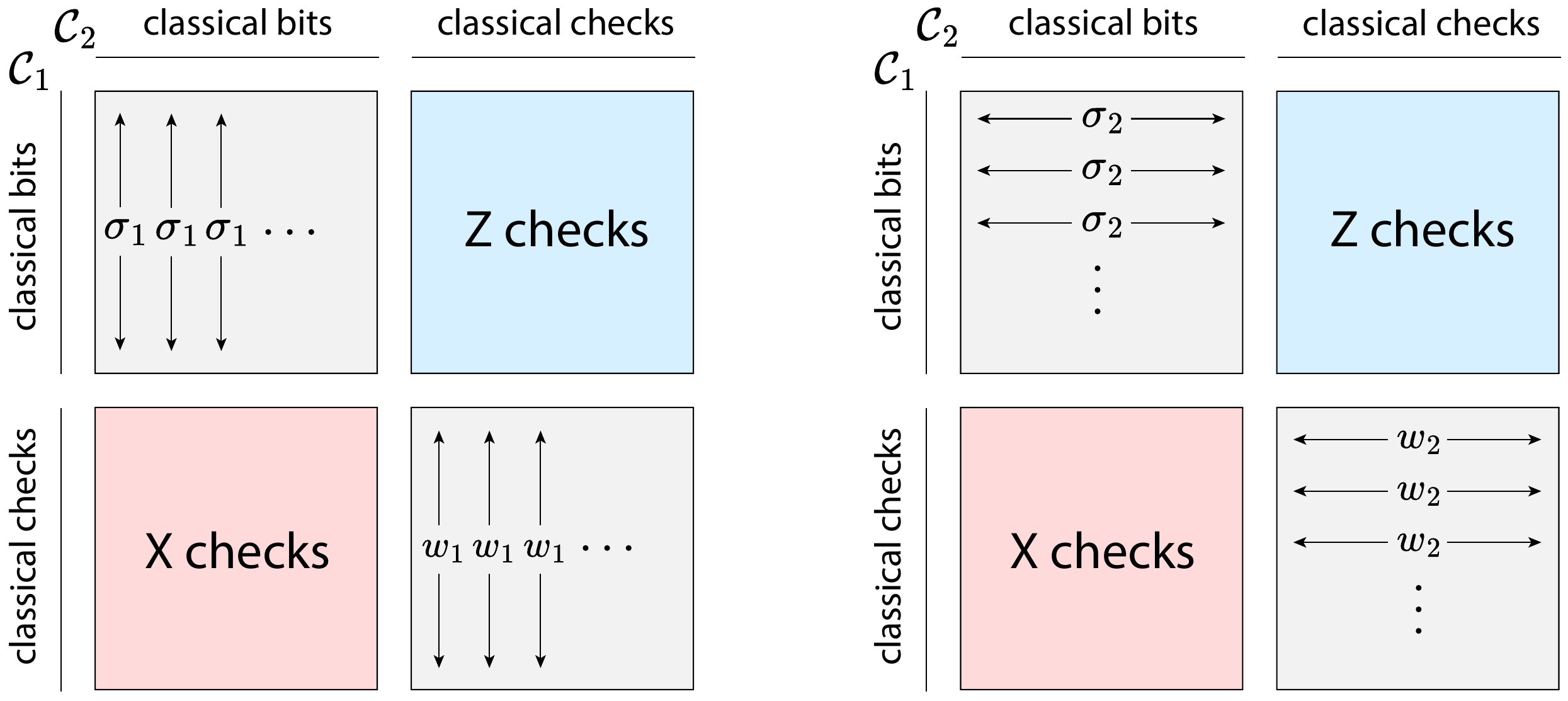}
    \caption{The physical implementations of the HGP automorphism gadgets \eqref{eq:HGP gadgets} are depicted. An automorphism of the first input code $\mathcal{C}_1$ results in a gadget acting within columns of the HGP code, and an automorphism of the second input code $\mathcal{C}_2$ results in a gadget acting within the rows.}
    \label{fig:2D HGP Aut gadget}
\end{figure}

We begin by reviewing the construction in (\cite{LRESC}, Appendix D), which introduced the concept of logical gate inheritance for hypergraph product codes. At a high level, the idea is to leverage the tensor product structure of the HGP to derive logical gadgets stemming from those of the classical input codes. For our purposes, we will focus on the gadgets corresponding to automorphisms of the classical input codes. Denoting $\oplus$ as the direct sum of matrices (as opposed to Kronecker sum), we define the unitary operators
\begin{subequations}\label{eq:HGP gadgets}
\begin{align}
U_1 &\equiv (\sigma_1 \otimes \ident) \oplus (w_1 \otimes \ident) = \begin{pmatrix}
    \sigma_1 \otimes \ident & \mathbf{0} \\ 
    \mathbf{0} & w_1 \otimes \ident
\end{pmatrix} \, , \label{eq:HGP U_1} \\
U_2 &\equiv (\ident \otimes \sigma_2) \oplus \big(\ident \otimes w_2^{-\transpose}\big) = \begin{pmatrix}
    \ident \otimes \sigma_2 & \mathbf{0} \\ 
    \mathbf{0} & \ident \otimes w_2^{-\transpose}
\end{pmatrix} \, .  \label{eq:HGP U_2}
\end{align}
\end{subequations}
We can interpret $U_1$ as acting $\sigma_1 \otimes \ident$ on the left qubits and $w_1 \otimes \ident$ on the right qubits where $\sigma_1$ is some permutation that satisfies $h_1 \sigma_1 = w_1 h_1$. 
Likewise, $\sigma_2$ is chosen to be some permutation for $U_2$ such that $h_2 \sigma_2 = w_2 h_2$.
Interpreting the tensor product as (rows $\otimes$ columns), we see that $U_1$ acts the permutation $\sigma_1$ on all left rows and $w_1$ on all right rows of the HGP code in parallel. $U_2$ follows suit but on the columns; see Figure \ref{fig:2D HGP Aut gadget} for an illustration. To check that $U_1$ preserves the stabilizer group of the HGP code, it suffices to examine its actions on $H_X$ and $H_Z$. For $H_X$ we have
\begin{align}\label{eq:H_XU_1}
    H_X U_1 &= \big( h_1\sigma \otimes \ident \;|\; w_1\otimes h_2^\transpose \big)  \notag \\
    &= \big( w_1h_1 \otimes \ident \;|\; w_1\otimes h_2^\transpose \big)  \notag \\
    &= (w_1 \otimes \ident) \big( h_1 \otimes \ident \;|\; \ident \otimes h_2^\transpose \big)  \notag \\
    &= W_1 H_X \, ,
\end{align}
where $W_1 \equiv w_1 \otimes \ident$. The action on $Z$-type Pauli operators is given by $U^{-\transpose}_1$, and so for $H_Z$ we have
\begin{align}\label{eq:H_ZU_1}
    H^{}_Z U^{-\transpose}_1 &= \big( \sigma_1 \otimes h_2 \;|\; h_1^\transpose w_1^{-\transpose} \otimes \ident \big)  \notag \\
    &= \big( \sigma_1 \otimes h_2 \;|\; (w_1^{-1}h_1)^\transpose \otimes \ident \big)  \notag \\
    &= \big( \sigma_1 \otimes h_2 \;|\; (h_1\sigma_1^{-1})^\transpose \otimes \ident \big)  \notag \\
    &= \big( \sigma_1 \otimes h_2 \;|\; \sigma_1 h_1^\transpose \otimes \ident \big)  \notag \\
    &= (\sigma_1 \otimes \ident)\big( \ident \otimes h_2 \;|\; h_1^\transpose \otimes \ident \big)  \notag \\
    &= W'_1 H^{}_Z  \, ,
\end{align}
where $W'_1 \equiv \sigma_1 \otimes \ident$, and in the first and fourth line we have used the fact that $\sigma_1^{-1} = \sigma_1^\transpose$ for a permutation matrix. \eqref{eq:H_XU_1} and \eqref{eq:H_ZU_1} show that the row spaces of $H_X$ and $H_Z$ remain invariant under $U_1$ \eqref{eq:HGP U_1}, and thus the stabilizer group is preserved. The proof for $U_2$ follows suit upon swapping the sides of the tensor product.

The logical action of $U_1$ \eqref{eq:HGP U_1} can be obtained by examining its action on our canonical logical basis \eqref{eq:HGP G_X, G_Z}. For now, we focus on the action on the left logical qubits. The logical $\overline{Z}$ operators \eqref{eq:HGP G_Z} transform according to
\begin{align}\label{eq:G'_Z}
    G^{}_Z U^{-\transpose}_1 &= \big( g_1\sigma_1 \otimes \{ \mathbf{e}_j \} \;\big|\; \mathbf{0} \big)  \notag \\
    &= \big( v_1g_1 \otimes \{ \mathbf{e}_j \} \;\big|\; \mathbf{0} \big)  \notag \\
    &= (v_1\otimes\ident) \, G_Z   \notag \\
    &= V'_1 G^{}_Z \, .
\end{align}
Hence, we see that within each canonical column of logical qubits, the logical $\overline{Z}$ operators transform the same way as the classical codewords, according to $v$. The action on the logical $\overline{X}$ operators is 
\begin{align}
    G_X U_1 &= \big( \{ \mathbf{e}_i \}\sigma_1 \otimes g_2 \;\big|\; \mathbf{0} \big) = \big( \{ \mathbf{e}_{\sigma_1(i)} \} \otimes g_2 \;\big|\; \mathbf{0} \big) \equiv G'_X \, .
\end{align}
It is not so clear from the above equation what the relation is between $G'_X$ and $G^{}_X$, but we emphasize that on the logical level it is completely determined by \eqref{eq:G'_Z} due to commutation relations between $\overline{X}$s and $\overline{Z}$s. Denote $\overline{V} \in \mathbb{F}^{k^2\times k^2}_2$ as the logical action of $V$ after accounting for stabilizer degeneracy. The Pauli algebra (i.e. the symplectic condition) demands $\overline{V}'_1 \overline{V}^\transpose_1 = \ident$, which fixes $\overline{V}'_1 = \overline{V}^{-\transpose}_1$. The analysis for $U_2$ follows similarly but amongst the rows of logical qubits. 

Overall, the logical $\overline{X}$ actions are given by
\begin{subequations}\label{eq:HGP gadgets logical action}
\begin{align}
    \overline{V}_1 &= v_1^{-\transpose} \otimes \ident  \\
    \overline{V}_2 &= \ident \otimes v_2 \, .
\end{align}
\end{subequations}
Since $[\overline{V}_1, \overline{V}_2] = 0$, the group of inherited logical automorphism gates is $\mathcal{A}_1 \times \mathcal{A}_2$. We comment that although the group generated by \eqref{eq:HGP gadgets logical action} is basis-invariant, the particular structure of these logical gates on the logical qubits will depend on the presentation of the generator matrices of the classical input codes.

Note that $w \in \mathrm{GL}_k(\mathbb{F}_2)$ \eqref{eq:aut conditions} is in general not a permutation matrix, and so our $U_1$ \eqref{eq:HGP U_1} and $U_2$ \eqref{eq:HGP U_2} do not necessarily belong to $\mathrm{Aut}\big( \mathrm{HGP}(\mathcal{C}) \big)$. In the special instances where $w \in \mathrm{S}_k$ is a permutation matrix, then we do indeed have an exact automorphism of the HGP code; i.e. $W^{}_1,W'_1,W^{}_2,W'_2 \in \mathrm{S}_n$. These special instances are precisely the Tanner graph automorphisms of Def. \ref{defn:graph automorphism}. Nonetheless, in the general case, we can decompose an arbitrary $w \in \mathrm{GL}_k(\mathbb{F}_2)$ into elementary operations consisting of permutations and CNOT gates.

\begin{thm}[Number of inherited logical gates for HGP codes]
    Let $\mathrm{HGP}(h_1,h_2)$ be a HGP code, and let $\mathcal{A}_1$ and $\mathcal{A}_2$ be the automorphism groups of the classical input codes. The number of distinct automorphism gadgets, according to \eqref{eq:HGP gadgets} and \eqref{eq:HGP gadgets logical action}, is $\abs{\mathcal{A}_1} \abs{\mathcal{A}_2}$.
\end{thm}

\begin{thm}[Lower bound on automorphism group of HGP codes]
    Let $\mathrm{HGP}(h_1,h_2)$ be a HGP code, and let $\mathcal{T}_1$ and $\mathcal{T}_2$ be the Tanner graph automorphism groups of $h_1$ and $h_2$ respectively. Then we have $\mathcal{T}_1 \times \mathcal{T}_2 \subset \mathrm{Aut}\big(\mathrm{HGP}(h_1,h_2)\big)$.
\end{thm}

So far, we have focused on automorphism gadgets that act as physical permutations on the left sector of data qubits. A complementary pair of automorphism gadgets can be derived using the transpose of the input codes, which corresponds to physical permutations that satisfy the automorphism condition \eqref{eq:dual aut condition} on the transposed parity-check matrices. These ``right-sector'' automorphism gadgets generically act as a permutation on the right sector and a circuit on the left sector. For brevity, we will only focus on the left-sector automorphism gadgets, as we typically have $k_1k_2 \gg k^\transpose_1k^\transpose_2$ in \eqref{eq:HGP code parameters}; however see Section \ref{sec:K4 HGP copy-cup} for an example where we use the right-sector automorphism gadgets. If we restrict to the subgroup of Tanner graph automorphisms in the input codes, then the left-sector and right-sector automorphism gadgets are isomorphic.

\subsection{Homological product codes: quantum $\times$ classical}

Having shown the expressions for automorphism gate inheritance in the hypergraph product, we now inductively derive analogous expressions for higher-fold homological products, comprised of a quantum CSS code and another classical code. The resulting homological product code can then be subsequently tensored again with another classical code to produce another homological product code and so forth.

Suppose the input CSS code possesses a logical gadget of the form $H_X U = W H_X$ and $H^{}_Z U^{-\transpose} = W' H_Z$, and the input classical code possesses an automorphism of the form $h\sigma = wh$. Then their homological product code will inherit the logical gadgets given by the actions of
\begin{subequations}\label{eq:3D HGP gadgets}
\begin{align}
    \tilde{U}_{\rm c} &\equiv \big(\ident \otimes \sigma\big) \oplus \big(\ident \otimes w^{-\transpose}\big)  \label{eq:3D HGP U_C}  \\
    \tilde{U}_{\rm Q} &\equiv \big( U \otimes \ident \big) \oplus \big(W \otimes \ident \big)  \label{eq:3D HGP U_Q}
\end{align}
\end{subequations}
on the Pauli $X$ operators, where the subscripts Q and c label the logical gadgets inherited from those of the input quantum and classical codes respectively. We will now show that the above physical transformations \eqref{eq:3D HGP gadgets} preserve the stabilizer group given by \eqref{eq:3D HGP H_X,H_Z}. Starting with $\tilde{U}_{\rm c}$ \eqref{eq:3D HGP U_C}, its action on $\tilde{H}_X$ \eqref{eq:3D HGP H_X} is
\begin{align}\label{eq:3D HGP H_XU_C}
    \tilde{H}_X \tilde{U}_{\rm c} &= \big(\, H_X \otimes \sigma \;\big|\; \ident \otimes h^\transpose w^{-\transpose} \,\big)  \notag \\
    &= \big(\, H_X \otimes \sigma \;\big|\; \ident \otimes \sigma h^\transpose \,\big)  \notag \\
    &= (\ident \otimes \sigma) \, \tilde{H}_X  \notag \\
    &= \tilde{W}_{\rm c} \tilde{H}_X \, ,
\end{align}
where in the last line we have defined $\tilde{W}_{\rm c} \equiv \ident \otimes \sigma$. The action on $\tilde{H}_Z$ \eqref{eq:3D HGP H_Z} is
\begin{align}\label{eq:3D HGP H_ZU_C}
    \tilde{H}^{}_Z \tilde{U}^{-\transpose}_{\rm c} &= \left(\begin{array}{c|c}
        H_Z \otimes \sigma \;\; & \mathbf{0}  \\
        \ident \otimes h\sigma & \; H^\transpose_X \otimes w
    \end{array}\right)  \notag \\
    &= \left(\begin{array}{c|c}
        H_Z \otimes \sigma \;\; & \mathbf{0}  \\
        \ident \otimes wh & \; H^\transpose_X \otimes w
    \end{array}\right)  \notag \\
    &= (\ident \otimes \sigma \;|\; \ident \otimes w ) \, \tilde{H}_Z  \notag \\
    &= \tilde{W}'_{\rm c} \tilde{H}_Z \, ,
\end{align}
where in the last line we have defined $\tilde{W}'_{\rm c} \equiv (\ident \otimes \sigma \;|\; \ident \otimes w )$. The combination of \eqref{eq:3D HGP H_XU_C} and \eqref{eq:3D HGP H_ZU_C} demonstrate that $\tilde{U}_{\rm c}$ \eqref{eq:3D HGP U_C} preserves the codespace.

Now we move onto $\tilde{U}_{\rm Q}$ \eqref{eq:3D HGP U_Q}. Its action on $\tilde{H}_X$ \eqref{eq:3D HGP H_X} is
\begin{align}\label{eq:3D HGP H_XU_Q}
    \tilde{H}_X \tilde{U}_{\rm Q} &= \big(\, H_X U \otimes \ident \;\big|\; W \otimes h^\transpose \,\big)  \notag \\
    &= \big(\, W H_X \otimes \ident \;\big|\; W \otimes h^\transpose \,\big)  \notag \\
    &= \big(W \otimes \ident \big) \, \tilde{H}_X  \notag \\
    &= \tilde{W}_{\rm Q} \tilde{H}_X
\end{align}
where in the last line we have defined $\tilde{W}_{\rm Q} \equiv W \otimes \ident$. The action on $\tilde{H}_Z$ \eqref{eq:3D HGP H_Z} is
\begin{align}\label{eq:3D HGP H_ZU_Q}
    \tilde{H}^{}_Z \tilde{U}^{-\transpose}_{\rm Q} &= \left(\begin{array}{c|c}
        H_Z U^{-\transpose} \otimes \ident \;\; & \mathbf{0}  \\
        U^{-\transpose} \otimes h & \; H^\transpose_X W^{-\transpose} \otimes \ident
    \end{array}\right)  \notag \\
    &= \left(\begin{array}{c|c}
        W' H_Z \otimes \ident \;\; & \mathbf{0}  \\
        U^{-\transpose} \otimes h & \; U^{-\transpose} H^\transpose_X \otimes \ident
    \end{array}\right)  \notag \\
    &= \big(\, W' \otimes \ident \;\big|\; U^{-\transpose} \otimes \ident \,\big) \, \tilde{H}_Z  \notag \\
    &= \tilde{W}'_{\rm Q} \tilde{H}^{}_Z \, ,
\end{align}
where in the last line we have defined $\tilde{W}'_{\rm Q} \equiv \big(\, W' \otimes \ident \;\big|\; U^{-\transpose} \otimes \ident \,\big)$. The combination of \eqref{eq:3D HGP H_XU_Q} and \eqref{eq:3D HGP H_ZU_Q} show that $\tilde{U}_{\rm Q}$ \eqref{eq:3D HGP U_Q} preserves the codespace.

Let us summarize the structure of all the left-action $\tilde{W}$s we have derived thus far:
\begin{subequations}
\begin{align}
    \tilde{W}_{\rm c} &= \ident \otimes \sigma  \\
    \tilde{W}'_{\rm c} &= (\ident \otimes \sigma \;|\; \ident \otimes w )  \\
    \tilde{W}_{\rm Q} &= W \otimes \ident  \\
    \tilde{W}'_{\rm Q} &= \big(\, W' \otimes \ident \;\big|\; U^{-\transpose} \otimes \ident \,\big)
\end{align}
\end{subequations}
Examining the above equations, we observe that $\tilde{W}_{\rm c}$ is a Kronecker product of permutation matrices and so is a permutation matrix itself. Additionally, if $w, W, U$ are constant-depth circuits, then so are $\tilde{W}'_{\rm c}, \tilde{W}_{\rm Q}, \tilde{W}'_{\rm Q}$ since they all involve Kronecker products of the former with the identity operator.
In particular, if our input quantum CSS code is a HGP code from the previous section, then we can use either $W = W_1 = w_0 \otimes \ident$ or $W = W_2 = \ident \otimes \sigma_0$ (subscript 0 to not be confused with the new $w$ and $\sigma$ defined above). If the input classical code has a low-depth implementation for some $w$, then the inherited automorphism gadget in the resulting homological product code can be implemented in the same depth as well. Furthermore, Tanner graph automorphisms of the input classical code can combine with automorphisms of the input quantum code to produce exact automorphisms of the resulting homological product code.

Analogous to the HGP scenario, we can derive the corresponding logical $\overline{X}$ actions by examining the action of \eqref{eq:3D HGP gadgets} on our canonical logical basis \eqref{eq:3D HGP G_X,G_Z}:
\begin{subequations}\label{eq:3D HGP gadgets logical actions}
\begin{align}
    \tilde{\overline{V}}_C &= \ident \otimes v \\
    \tilde{\overline{V}}_Q &= \overline{V} \otimes \ident \, .
\end{align}
\end{subequations}
Since $\tilde{\overline{V}}_C$ and $\tilde{\overline{V}}_Q$ commute, the group of logical gates is a direct product of those of the input quantum and classical codes.

\subsection{Homological product codes: quantum $\times$ quantum}\label{sec:auto-gadget-quantum-times-quantum-hom-prod-code}

We can just as easily extend the analysis to the case where we take the homological product of two CSS codes. As noted in~\eqref{eq:4D HP complex}, the resulting quantum code can be represented as a 3-term chain complex out of the larger 5-term chain complex. One can also obtain a 5-term chain complex through the repeated homological product, i.e. classical $\times$ classical $\times$ classical $\times$ classical; however, that construction does not yield all instances achievable by quantum $\times$ quantum. In particular, it may be the case that the two quantum codes are not themselves product codes.

Suppose we are given two quantum CSS codes $\mathcal{C}, \mathcal{C}'$ defined by parity check matrices $H_X, H_Z$ and $H'_X, H_Z'$, respectively, which possess logical gadgets of the form
\begin{subequations}\label{eq:4D HGP conditions}
    \begin{align}
        H_XU_1 &= W_1H_X \\
        H_ZU_1^{-\transpose} &= \hat{W}_1H_Z \\
        H_X' U_2 &= W_2H_X' \\
        H_Z' U_2^{-\transpose} &= \hat{W}_2H_Z'
    \end{align}
\end{subequations}
Then the quantum code resulting from the homological product of $\mathcal{C}$ and $\mathcal{C}'$ will inherit logical automorphism gadgets defined as
\begin{subequations}\label{eq:4D HGP gadgets}
\begin{align}
    \tilde{U}_{1} &\equiv \big(\hat{W}_1^{-\transpose} \otimes \ident\big) \oplus \big(U_1 \otimes \ident\big) \oplus \big(W_1 \otimes \ident \big)  \label{eq:4D HGP U_Q1}  \\
    \tilde{U}_{2} &\equiv \big(\ident \otimes W_2\big) \oplus \big(\ident \otimes U_2 \big) \oplus \big(\ident \otimes \hat{W}_2^{-\transpose}  \big)  \label{eq:4D HGP U_Q2}.
\end{align}
\end{subequations}
Once again, we will show that these unitaries preserve the stabilizer group given by \eqref{eq:4D HP H_X,H_Z}:
\begin{subequations}
\begin{align}
    \tilde{H}_X\tilde{U}_1 &= \left(\begin{array}{c|c|c}
        H_Z^\transpose\hat{W}_1^{-\transpose} \otimes \ident \;\; & \;U_1 \otimes H'_X \;\;& \mathbf{0}  \\
        \mathbf{0} \;\; & H_XU_1 \otimes \ident \;\; & \; W_1 \otimes H_Z^{\prime\transpose}
    \end{array}\right) \,  \notag  \\
     &= \left(\begin{array}{c|c|c}
        U_1H_Z^\transpose \otimes \ident \;\; & \;U_1 \otimes H'_X \;\;& \mathbf{0}  \\
        \mathbf{0} \;\; & W_1H_X \otimes \ident \;\; & \; W_1 \otimes H_Z^{\prime\transpose}
    \end{array}\right) \,  \notag \\
    &= \big(\, U \otimes \ident \;\big|\; W \otimes \ident \,\big) \, \tilde{H}_X  \notag \\
    &= \tilde{W}_1 \tilde{H}^{}_X \, \label{eq:4D HGP H_XU_Q1} ,  \\
    \tilde{H}_Z \tilde{U}_{1}^{-\transpose} &= \left(\begin{array}{c|c|c}
        \hat{W}_1 \otimes H_X^{\prime\transpose} \;\; & H_ZU_1^{-\transpose} \otimes \ident \;\;& \mathbf{0}  \\
        \mathbf{0} \;\; & \;U_1^{-\transpose} \otimes H'_Z \;\; & \; H_X^\transpose W_1^{-\transpose} \otimes \ident
    \end{array}\right) \,   \notag \\
    &= \left(\begin{array}{c|c|c}
        \hat{W}_1 \otimes H_X^{\prime\transpose} \;\; & \hat{W}_1 H_Z \otimes \ident \;\;& \mathbf{0}  \\
        \mathbf{0} \;\; & \; U_1^{-\transpose} \otimes H'_Z \;\; & \; U_1^{-\transpose} H_X^\transpose \otimes \ident
    \end{array}\right) \notag \\
    &= \big(\, \hat{W}_1 \otimes \ident \;\big|\; U_1^{-\transpose} \otimes \ident \,\big) \, \tilde{H}_Z  \notag \\
    &= \tilde{W}'_1 \tilde{H}^{}_Z \, , \label{eq:4D HGP H_ZU_Q1} 
\end{align}
\end{subequations}
The combination of \eqref{eq:4D HGP H_XU_Q1} and \eqref{eq:4D HGP H_ZU_Q1} shows that $\tilde{U}_1$ \eqref{eq:4D HGP U_Q1} preserves the codespace. We now look at $\tilde{U}_{2}$ \eqref{eq:4D HGP U_Q2}, the action of which yields
\begin{subequations}
\begin{align}
    \tilde{H}_X \tilde{U}_{2} &= \left(\begin{array}{c|c|c}
        H_Z^\transpose \otimes W_2 \;\; & \; \ident \otimes H'_X U_2 \;\;& \mathbf{0}  \\
        \mathbf{0} \;\; & H_X \otimes U_2 \;\; & \; \ident \otimes H_Z^{\prime\transpose} \hat{W}_2^{-\transpose}
    \end{array}\right) \,  \notag  \\
     &= \left(\begin{array}{c|c|c}
        H_Z^\transpose \otimes W_2 \;\; & \; \ident \otimes W_2 H'_X \;\;& \mathbf{0}  \\
        \mathbf{0} \;\; & H_X \otimes U_2 \;\; & \; \ident \otimes U_2 H_Z^{\prime\transpose}
    \end{array}\right) \,  \notag \\
    &= \big(\, \ident \otimes W_2 \;\big|\; \ident \otimes U_2 \,\big) \, \tilde{H}_X  \notag \\
    &= \tilde{W}_{2} \tilde{H}^{}_X \, ,\label{eq:4D HGP H_XU_Q2}  \\
    \tilde{H}_Z \tilde{U}_{2}^{-\transpose} &= \left(\begin{array}{c|c|c}
        \ident \otimes H_X^{\prime\transpose} W_2^{-\transpose} \;\; & H_Z \otimes U_2^{-\transpose} \;\;& \mathbf{0}  \\
        \mathbf{0} \;\; & \;\ident \otimes H'_ZU_2^{-\transpose} \;\; & \; H_X^\transpose  \otimes \hat{W}_2
    \end{array}\right) \,   \notag \\
    &= \left(\begin{array}{c|c|c}
        \ident \otimes U_2^{-\transpose} H_X^{\prime\transpose} \;\; & H_Z \otimes U_2^{-\transpose} \;\;& \mathbf{0}  \\
        \mathbf{0} \;\; & \;\ident \otimes \hat{W}_2 H'_Z \;\; & \; H_X^\transpose  \otimes \hat{W}_2
    \end{array}\right) \notag \\
    &= \big(\, \ident \otimes U_2^{-\transpose} \;\big|\; \ident \otimes \hat{W}_2 \,\big) \, \tilde{H}_Z  \notag \\
    &= \tilde{W}'_{2} \tilde{H}^{}_Z \, ,\label{eq:4D HGP H_ZU_Q2}
\end{align}
\end{subequations}
And so with \eqref{eq:4D HGP H_XU_Q2} and \eqref{eq:4D HGP H_ZU_Q2} we have shown that $\tilde{U}_{Q'}$ \eqref{eq:4D HGP U_Q2} preserves the codespace. We can make similar statements about the resulting $\tilde{W}$s as we did in the quantum $\times$ classical case, namely that they preserve circuit depth and are also permutations if the input gadgets are Tanner graph automorphisms.

The action of the automorphism gadgets \eqref{eq:4D HGP gadgets} on the canonical logical basis of the middle sector \eqref{eq:4D HP G_X,G_Z} follows an analogous tensor product structure compared to both the HGP \eqref{eq:HGP gadgets logical action} and the (quantum $\times$ classical) homological product \eqref{eq:3D HGP gadgets logical actions} cases:
\begin{subequations}\label{eq:4D HGP gadgets logical actions}
\begin{align}
    \tilde{\overline{V}}_1 &= \overline{V}_1 \otimes \ident \\
    \tilde{\overline{V}}_2 &= \ident \otimes \overline{V}_2 \, .
\end{align}
\end{subequations}
As a reminder, both $\tilde{\overline{V}}_1, \tilde{\overline{V}}_2 \in \tilde{\mathcal{A}}$ are elements of the logical automorphism gadget group.


\section{Fault tolerance of automorphism gadgets}
\label{sec:fault tolerance}

In this section, we analyze the fault tolerance of our logical automorphism gadgets constructed in Section \ref{sec:aut gadgets}. Because we have described three different types of constructions for homological product codes, i.e. classical $\times$ classical, quantum $\times$ classical, and quantum $\times$ quantum, we divide this section into three different subsections that analyze the fault tolerance of the automorphism gadgets for these product codes. For a vector $\mathbf{x} \in \mathbb{F}^{n}_2$, let $\mathbf{x}_L$ and $\mathbf{x}_R$ denote the restrictions of $\mathbf{x}$ to the physical qubits in the left and right sectors respectively. For the case of the quantum $\times$ quantum homological product code, we let $\mathbf{x}_M$ to denote the restriction to the physical qubits in the middle sector. For example, the middle physical qubit sector of a quantum $\times$ quantum homological product code corresponds to the qubit subspace $(\mathrm{Q}, \mathrm{Q}')$ and can be laid out in a grid fashion that is represented by a matrix in $\mathbb{F}_2^{n^{}_{\rm Q} \times n'_{\rm Q}}$. Refer to Figure~\ref{fig:schematic-quantum-quantum} for a diagrammatic representation of these grids and matrices.

For our noise model of interest, we will assume that all single-qubit gates (including idling), multi-qubit gates and measurements are subject to failure. On the other hand, we will assume that physical permutations do \emph{not} spread errors. This last assumption is based on the observation that for architectures with movable qubits, the SWAP gates can be performed by physical motion rather than 2-qubit gates that can cause correlated errors. Often, we can simply relabel the qubits in software rather than perform the physical permutations themselves, as has been typically done in recent experiments~\cite{burton2024, hong2024, reichardt2024}. The physical permutations only need to be performed when performing other fault-tolerant gadgets such as an interblock transversal gate or an intrablock fold-transversal gate; see Sec. \ref{sec:alleviating correlated errors} for more details. Combining the above assumptions on the noise with the structure of logical operators in homological product codes, we will show that our automorphism gadgets possess a form of inherent fault tolerance with respect to circuit-level noise.

\subsection{Fault tolerance for hypergraph product codes}
\label{sec:FT 2D HGP}

For simplicity, we will only analyze logical qubits living in the left sector of physical qubits, as these are the only ones we care about for computation.  We will leverage the following results in the literature regarding the structure of logical operators in hypergraph product codes. Let $|\cdot|_{\rm L}$ denote the weight on the left sector $\mathrm{\Lambda}_{\rm L}$ of data qubits.

\begin{lem}[Hypergraph product logical sector weight; Prop. 2 of \cite{ReShape_decoder}]
\label{lem:2D HGP logical support}
    Suppose we have a HGP code with parity-check matrices $H_X, H_Z$ of the form \eqref{eq:HGP H_X, H_Z} from the hypergraph product of $[n_1,k_1,d_1]$ and $[n_2,k_2,d_2]$ classical linear codes. Let $\mathbf{x}_{\rm L} \in \ker H_Z \backslash \operatorname{rs} H_X$ denote the support of a nontrivial left-sector logical $\overline{X}$ operator, and let $\mathbf{z}_{\rm L} \in \ker H_Z \backslash \operatorname{rs} H_X$ denote the support of a nontrivial left-sector logical $\overline{Z}$ operator. Let $\hat{\mathbf{x}}_{\rm L}$ and $\hat{\mathbf{z}}_{\rm L}$ denote their canonical representations according to \eqref{eq:HGP G_X, G_Z}. Then we have
    \begin{align}\label{eq:HGP logical sector weight}
        \abs{\mathbf{x}_{\rm L}}_{\rm L} \geq d_2 \quad,\quad \abs{\mathbf{z}_{\rm L}}_{\rm L} \geq d_1 \, .
    \end{align}
\end{lem}

A consequence of Lemma \ref{lem:2D HGP logical support} is that the weight of any logical operator on the left sector cannot be decreased upon appending stabilizers. From the HGP code parameters \eqref{eq:HGP code parameters} as well as the description of the HGP automorphism gadgets \eqref{eq:HGP gadgets}, we note that we will sometimes want to ignore the right-sector logical qubits and treat them as gauge qubits which hold no logical information. When we ignore the right-sector logical qubits, our original HGP code, which was a CSS stabilizer code, will now become a CSS \emph{subsystem} code.

\begin{defn}[Left-sector hypergraph product code]
\label{defn:left-sector HGP code}
    Given a HGP code with canonical logical operators \eqref{eq:HGP G_X, G_Z}, the \emph{left-sector} HGP code is the CSS subsystem code obtained upon designating the right-sector logical qubits as gauge qubits.
\end{defn}

Since we are simply ignoring the right-sector logical qubits, the number of data qubits remains the same. The number of logical qubits, however, now decreases to $k=k_1k_2$ (we have lost the $k^\transpose_1k^\transpose_2$ right-sector logical qubits). It is not hard to see that Lemma \ref{lem:2D HGP logical support} remains unchanged upon ignoring right-sector logical qubits because their canonical operators do not have any support in the left-sector.

\begin{lem}[Extension of Lemma \ref{lem:2D HGP logical support}]
\label{lem:2D HGP logical support 2}
    For the case where the classical input codes have rank-deficient parity-check matrices, Lemma \ref{lem:2D HGP logical support} applies to all logical qubits in the left-sector HGP code given by Definition \ref{defn:left-sector HGP code}.
\end{lem}

\begin{proof}
    Similar to the left-sector logical qubits \eqref{eq:HGP G_X, G_Z}, a canonical Pauli basis for the right-sector logical qubits can be written as
    \begin{subequations}
    \begin{align}
        G_{Z,{\rm R}} &= \big(\, \mathbf{0} \,\big|\, \{\mathbf{e}_i\} \otimes g_{2,\transpose} \,\big)  \\
        G_{X,{\rm R}} &= \big(\, \mathbf{0} \,\big|\, g_{1,\transpose} \otimes \{\mathbf{e}_j\} \,\big) \, ,
    \end{align}
    \end{subequations}
    where $g_{1,\transpose}$ and $g_{2,\transpose}$ are generator matrices for the transpose codes (with parity-check matrices $h^\transpose_1$ and $h^\transpose_2$), and $\{\mathbf{e}_i\} \notin \operatorname{rs}h^\transpose_1$ and $\{\mathbf{e}_j\} \notin \operatorname{rs}h^\transpose_2$. When we designate the right-sector logical qubits as gauge operators, the above logical operators can ``dress'' our left-sector logical operators in addition to the code stabilizers. A generic logical $\overline{X}$ operator now takes the form
    \begin{align}
        \mathbf{x}  = \hat{\mathbf{x}}_{\rm L} + \mathbf{s}_X + \hat{\mathbf{x}}_{\rm R} \, ,
    \end{align}
    where $\hat{\mathbf{x}}_{\rm L}$ is a canonical left-sector logical $\overline{X}$ operator, $\mathbf{s}_X$ is an element of the $X$-stabilizer subgroup, and $\hat{\mathbf{x}}_{\rm R}$ is a canonical right-sector logical $\overline{X}$ operator. Now, since $\hat{\mathbf{x}}_{\rm R}$ has no support on the left sector, i.e. $\abs{\hat{\mathbf{x}}_{\rm R}}_{\rm L} = 0$, we have that
    \begin{align}
        \abs{\mathbf{x}}_{\rm L} = \abs{\hat{\mathbf{x}}_{\rm L} + \mathbf{s}_X}_{\rm L} \geq d_2 \, ,
    \end{align}
    where we used Lemma \ref{lem:2D HGP logical support} in the last inequality. An analogous argument holds for the logical $\overline{Z}$ operators.
\end{proof}

\begin{prop}[Distances of left-sector hypergraph product codes]
\label{prop:left-sector HGP distances}
    For a left-sector HGP code with CSS parity-check matrices \eqref{eq:HGP H_X, H_Z} and classical input code parameters $[n_1,k_1,d_1]$ and $[n_2,k_2,d_2]$, the $X$ and $Z$ distances are given by
    \begin{align}\label{eq:left-sector HGP distances}
        d_X = d_2 \quad,\quad d_Z = d_1 \, .
    \end{align}
\end{prop}

\begin{proof}
    Lemma \ref{lem:2D HGP logical support 2} implies that $d_X \geq d_2$ and $d_Z \geq d_1$. At the same time, our canonical logical basis \eqref{eq:HGP G_X, G_Z} tells us that there exists a logical $\overline{X}$ with weight $d_2$ and a logical $\overline{Z}$ with weight $d_1$, and so we also have $d_X \leq d_2$ and $d_Z \leq d_1$. Thus, we arrive at \eqref{eq:left-sector HGP distances}.
\end{proof}

We are now ready to show that our automorphism gadgets preserve the full distance of left-sector HGP codes under our noise assumptions.

\begin{thm}[Effective distance preservation of automorphism gadgets in left-sector hypergraph product codes]
\label{thm:d_eff 2D HGP}
    Suppose physical qubit permutations do not spread errors, and we have a $\llbracket n,k,d \rrbracket$ left-sector HGP code according to Definition \ref{defn:left-sector HGP code} and Proposition \ref{prop:left-sector HGP distances} with automorphism gadgets given by \eqref{eq:HGP gadgets}. Then the effective fault distance of the automorphism gadgets is $d_{\rm eff} = d$. In other words, $d$ elementary faults are required to effect a nontrivial logical operation in the code.
\end{thm}

\begin{proof}
    The proof essentially follows from Lemma \ref{lem:2D HGP logical support 2} and Proposition \ref{prop:left-sector HGP distances}, which say that the logical support on the left sector of physical qubits cannot be less than the code distance. As a consequence, at least $d$ elementary faults must occur on the left sector in order to effect a nontrivial logical operation on the code. Since our HGP automorphism gadgets \eqref{eq:HGP gadgets} only permute the left sector, they do not spread errors within the left sector and thus preserve the distance of the code.
\end{proof}

\subsection{Fault tolerance for homological product codes: quantum $\times$ classical}\label{sec:ft-auto-hom-prod-code}

For the case of classical $\times$ quantum homological product codes, the analysis is largely similar. We first begin by stating a recent result of Tan and Stambler~\cite{tan2024hgp} before connecting it with the fault tolerance of our automorphism gadget for the quantum $\times$ classical homological product code.

\begin{lem}[Homological product logical sector weight; Lemma 26 of \cite{tan2024hgp}]
\label{lem:3D HGP logical support}
    Suppose we have a (quantum $\times$ classical) homological product code with a classical input code with distance $d_{\rm c}$ and a quantum CSS input code with distances ($d_X,d_Z$). In addition, suppose that the parity-check matrix of the classical input code has full rank; i.e. there are no redundant checks. Let $\mathbf{x} \in \ker \tilde{H}_Z \backslash \operatorname{rs} \tilde{H}_X$ denote the support of a nontrivial logical $\overline{X}$ operator, and let $\mathbf{z} \in \ker \tilde{H}_Z \backslash \operatorname{rs} \tilde{H}_X$ denote the support of a nontrivial logical $\overline{Z}$ operator. Then we have
    \begin{align}
        \abs{\mathbf{x}}_{\rm L} \geq d_X d_{\rm c} \quad,\quad \abs{\mathbf{z}}_{\rm L} \geq d_Z \, .
    \end{align}
\end{lem}

Again, we observe that the weight of any logical operator on the left sector cannot be decreased upon appending stabilizers. A technical assumption of the above lemma is that the classical input code has a full-rank parity-check matrix. For our purposes, often this will not be the case, i.e. there are linearly dependent parity checks; see the later Section \ref{sec:classical families} for such examples. As a consequence, one might worry about additional logical qubits resulting from the spurious (co)homologies in the product construction that could potentially spoil the logical sector weight property of Lemma \ref{lem:3D HGP logical support}. Fortunately, and similar to Definition \ref{defn:left-sector HGP code} in the HGP case, we will now show that these additional logical qubits can be chosen to live on the right sector and do not affect Lemma \ref{lem:3D HGP logical support} when restricted to the left-sector.

\begin{defn}[Left-sector homological product code]
\label{defn:left-sector 3D homological product code}
    Given a (quantum $\times$ classical) homological product code with canonical logical operators \eqref{eq:3D HGP G_X,G_Z}, the \emph{left-sector} homological product code is the CSS subsystem code obtained upon designating the right-sector logical qubits as gauge qubits.
\end{defn}

\begin{lem}[Extension of Lemma \ref{lem:3D HGP logical support}]
\label{lem:3D HGP logical support 2}
    For the case where the classical input code has a rank-deficient parity-check matrix, Lemma \ref{lem:3D HGP logical support} still holds for the left-sector homological product code given by Definition \ref{defn:left-sector 3D homological product code}.
\end{lem}

\begin{proof}
    We begin by constructing a canonical Pauli basis for the right-sector logical operators. From the K\"unneth formula, we know that there will be $k^\transpose_{\rm c}\big( n_{\rm Q} - \operatorname{rank}H^\transpose_X \big)$ spurious logical qubits in this right sector. Define
    \begin{subequations}
    \begin{align}
        \tilde{G}_{Z,{\rm R}} &= \big( \;\mathbf{0}\;|\; \{\mathbf{e}_i\} \otimes g^{}_\transpose \;\big)  \label{eq:3D HGP G_ZR} \\
        \tilde{G}_{X,{\rm R}} &= \big( \;\mathbf{0}\;|\; G_{X,\transpose} \otimes \{\mathbf{e}_j\} \;\big)  \, , \label{eq:3D HGP G_XR}
    \end{align}
    \end{subequations}
    where $\{\mathbf{e}_i\} \notin \operatorname{rs}H^\transpose_X$ and $\{\mathbf{e}_j\} \notin \operatorname{rs}h^\transpose$, and $g^{}_\transpose$ and $G_{X,\transpose}$ are the generator matrices of the transpose codes that satisfy $h^\transpose g^\transpose_\transpose = 0$ and $H^\transpose_X G^\transpose_{X,\transpose} = 0$ respectively. By construction, \eqref{eq:3D HGP G_ZR} and \eqref{eq:3D HGP G_XR} form a canonical Pauli basis for the spurious logical qubits. The statement of the Lemma then follows by the simple observation that the left-sector and right-sector logical qubits' canonical bases act on disjoint data qubits, and so Lemma \ref{lem:3D HGP logical support} is unaffected even when we dress our left-sector logical qubits with the right-sector gauge qubits.
\end{proof}

\begin{prop}[Distances of left-sector homological product code]
\label{prop:left-sector homological product code distance}
    For a left-sector homological product code with CSS parity-check matrices \eqref{eq:3D HGP H_X,H_Z} and input code parameters $\llbracket n_{\rm Q},k_{\rm Q},(d_X,d_Z) \rrbracket$ and $[n_{\rm c},k_{\rm c},d_{\rm c}]$, the $X$ and $Z$ distances are given by
    \begin{align}\label{eq:left-sector homological product code distance}
        \tilde{d}_X = d_X d_{\rm c} \quad,\quad \tilde{d}_Z = d_Z \, .
    \end{align}
\end{prop}

\begin{proof}
    Lemma \ref{lem:3D HGP logical support 2} implies that $\tilde{d}_X \geq d_X d_{\rm c}$ and $\tilde{d}_Z \geq d_Z$. At the same time, our canonical logical basis \eqref{eq:3D HGP G_X,G_Z} tells us that there exists a logical $\overline{X}$ with weight $d_X d_{\rm c}$ and a logical $\overline{Z}$ with weight $d_Z$, and so we also have $\tilde{d}_X \leq d_X d_{\rm c}$ and $\tilde{d}_Z \leq d_Z$. Thus, we arrive at \eqref{eq:left-sector homological product code distance}.
\end{proof}

Lemma \ref{lem:3D HGP logical support} essentially gives us the fault tolerance of our automorphism gadgets acting on the left logical qubits, when the input gadgets are code automorphisms.

\begin{thm}[Effective distance preservation of automorphism gadgets in homological product codes, part 1] \label{thm:d_eff 3D HGP 1}
    Suppose physical qubit permutations do not spread errors, and we have a (quantum $\times$ classical) left-sector homological product code according to Definition \ref{defn:left-sector 3D homological product code} with automorphism gadgets given by \eqref{eq:3D HGP gadgets}. If the input gadgets of the input quantum CSS code are code automorphisms, i.e. $U \in \mathrm{S}_{n_{\rm Q}}$ in \eqref{eq:3D HGP gadgets}, then the resulting automorphism gadgets will be distance-preserving.
\end{thm}

\begin{proof}
    We will analyze the two automorphism gadgets \eqref{eq:3D HGP gadgets} separately. For $\tilde{U}_{\rm c}$ \eqref{eq:3D HGP U_C}, note that its action on the left sector of physical qubits is a permutation. Since Lemma \ref{lem:3D HGP logical support 2} says that the left-sector logical weight cannot be less than the code distance, the fault distance of $\tilde{U}_{\rm c}$ is equal to the code distance. The situation is analogous for $\tilde{U}_{\rm Q}$ \eqref{eq:3D HGP U_Q} because the $U \otimes \ident$ acting on the left-sector is also a permutation. Again, Lemma \ref{lem:3D HGP logical support 2} tells us that the left-sector logical weight cannot decrease, and so at least $\tilde{d}_X = d_{\rm Q}d_{\rm c}$ faults are required to effect a nontrivial logical $\overline{X}$, and at least $\tilde{d}_Z = d_{\rm Q}$ faults are required to effect a nontrivial logical $\overline{Z}$.
\end{proof}

The trickier case is when the input gadget of the input CSS code is not a strict code automorphism but rather a hybrid gadget ($U \notin \mathrm{S}_{n_{\rm Q}}$) such the automorphism gadgets in Section \ref{sec:aut gadgets}. Recall that our automorphism gadgets generically have physical operations that include both a permutation and a circuit, and so they fall outside of the assumptions of Theorem \ref{thm:d_eff 3D HGP 1}. If our input CSS code was a HGP code or another homological product code with automorphism gadgets, it would be nice if we could lift those gadgets to the combined homological product code. Fortunately, we can leverage the product structure once again to show that this is indeed the case. In fact, our proof technique is closely related to that of Theorem~4.2 in Ref.~\cite{evra2022decodable}. Let $|\cdot|_{\mathrm{\Lambda}_{\rm L}}$ denote the weight on the $\mathrm{\Lambda}_{\rm L} \otimes \ident \subset \tilde{\mathrm{\Lambda}}_{\rm L}$ subsystem, corresponding to the rows belonging to $\mathrm{\Lambda}_{\rm L}$ of the input quantum code.

\begin{prop}[Homological product restricted logical sector weight]
\label{prop:3D HGP restricted logical sector weight}
    Suppose we have a (quantum $\times$ classical) homological product code \eqref{eq:3D HGP H_X,H_Z} where the input CSS code is either a left-sector HGP code by Definition \ref{defn:left-sector HGP code} or a left-sector (quantum $\times$ classical) homological product code by Definition \ref{defn:left-sector 3D homological product code}; denote this input left sector as $\mathrm{\Lambda}_{\rm L}$. Suppose the input CSS code has distances $(d_X,d_Z)$, and the input classical code has distance $d_{\rm c}$. Let $\mathbf{x} \in \ker \tilde{H}_Z \backslash \operatorname{rs} \tilde{H}_X$ denote the support of a nontrivial logical $\overline{X}$ operator, and let $\mathbf{z} \in \ker \tilde{H}_Z \backslash \operatorname{rs} \tilde{H}_X$ denote the support of a nontrivial logical $\overline{Z}$ operator. Then we have
    \begin{align}
        \abs{\mathbf{x}}_{\mathrm{\Lambda}_{\rm L}} \geq d_X d_{\rm c} \quad,\quad \abs{\mathbf{z}}_{\mathrm{\Lambda}_{\rm L}} \geq d_Z \, .
    \end{align}
\end{prop}

\begin{figure}[t]
    \centering
    \includegraphics[width=0.5\textwidth]{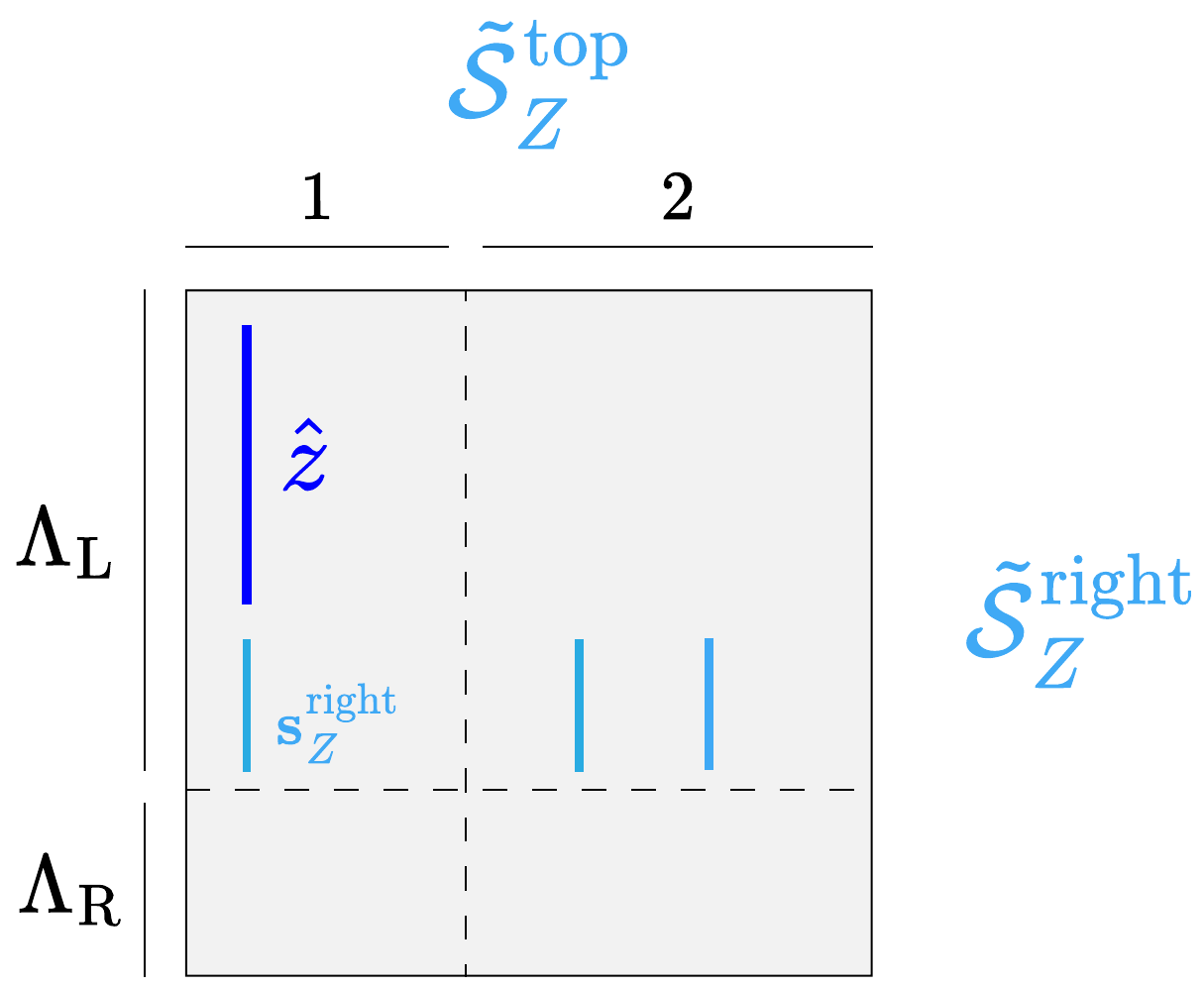}
    \caption{The partitioning of $\tilde{\mathrm{\Lambda}}_{\rm L}$ used for the proof of Proposition \ref{prop:3D HGP restricted logical sector weight}. $\hat{\mathbf{z}}$ is the support of a canonical logical $\overline{Z}$ operator, and $\mathbf{s}^\text{right}_Z$ is the support of an element of $\tilde{\mathcal{S}}^\text{right}_Z$.}
    \label{fig:3D left sector partition}
\end{figure}

\begin{proof}
    Before we proceed, let $\tilde{\mathrm{\Lambda}}_{\rm L}$ denote the left sector of the resulting homological product code, and let $\mathrm{\Lambda}_{\rm L}$ denote the left sector of the input HGP or homological product code. Let $\tilde{\mathbf{r}}(\cdot) \in \mathbb{F}^{n_{\rm Q}}_2$ denote the row support in $\tilde{\mathrm{\Lambda}}_{\rm L}$, and, within $\tilde{\mathrm{\Lambda}}_{\rm L}$, let $\tilde{\mathbf{r}}_{\mathrm{\Lambda}_{\rm L}}(\cdot)$ denote the \emph{restricted} row support on the $\mathrm{\Lambda}_{\rm L}$ rows of $\tilde{\mathrm{\Lambda}}_{\rm L}$, corresponding to the left sector of the input HGP or homological product code (recall that there are $n_{\rm Q} = |\mathrm{\Lambda}_{\rm L}| + |\mathrm{\Lambda}_{\rm R}|$ total rows in $\tilde{\mathrm{\Lambda}}_{\rm L}$).
    
    We first analyze the logical $\overline{Z}$ operators. Suppose we have a nontrivial logical $\overline{Z}$ operator with support $\mathbf{z}$ and canonical support $\hat{\mathbf{z}}$ given in \eqref{eq:3D HGP G_Z}. Note that there are two ``types'' of $Z$-checks that we need to analyze, given by the two row-blocks in \eqref{eq:3D HGP H_Z}: a ``top-type'' $Z$-check and a ``right-type''. Note that any $\mathbf{s}_Z \in \rs \tilde{H}_Z$ can be decomposed into a linear combination of top-type and right-type checks: $\mathbf{s}_Z = \mathbf{s}^{\rm top}_Z + \mathbf{s}^{\rm right}_Z$. Furthermore, we will decompose the top-type $Z$-checks into those that act in the columns indexed by $\{\mathbf{e}_i\} \notin \rs h$ (region 1) and its complement (region 2); see Figure \ref{fig:3D left sector partition} for a sketch of the above partitioning. By linearity we can analyze the total contribution of the stabilizer element by summing the independent contributions from each constituent type; as such we have that $\mathbf{z} = \hat{\mathbf{z}} + \mathbf{s}^{\rm top_1}_Z + \mathbf{s}^{\rm top_2}_Z + \mathbf{s}^{\rm right}_Z$. We begin with the contribution from $\mathbf{s}^{\rm top_1}_Z$. Recall that the top-type $Z$-checks act as $H_Z\otimes\ident$, i.e. independent copies of $H_Z$ among the columns of $\tilde{\mathrm{\Lambda}}_{\rm L}$. As a result, Lemma \ref{lem:3D HGP logical support 2} tells us that $|\tilde{\mathbf{r}}_{\mathrm{\Lambda}_{\rm L}}(\hat{\mathbf{z}}+\mathbf{s}^{\rm top_1}_Z)| \geq |\tilde{\mathbf{r}}_{\mathrm{\Lambda}_{\rm L}}(\hat{\mathbf{z}})| \geq d_Z$. Now, notice that $\hat{\mathbf{z}}+\mathbf{s}^{\rm top_1}_Z$ is still completely supported in the columns of region 1. Since these columns are not in $\rs h$ by definition, whenever we try to reduce the weight in any column in region 1 by appending a right-type $Z$-check, we must spawn additional support in new columns in region 2. As a result, the row weight cannot be decreased from these right-type $Z$-checks; i.e. we have $|\tilde{\mathbf{r}}_{\mathrm{\Lambda}_{\rm L}}(\hat{\mathbf{z}}+\mathbf{s}^{\rm top_1}_Z + \mathbf{s}^{\rm right}_Z)| \geq |\tilde{\mathbf{r}}_{\mathrm{\Lambda}_{\rm L}}(\hat{\mathbf{z}}+\mathbf{s}^{\rm top_1}_Z)| \geq |\tilde{\mathbf{r}}_{\mathrm{\Lambda}_{\rm L}}(\hat{\mathbf{z}})| \geq d_Z$. Finally, we will deal with $\mathbf{s}^{\rm top_2}_Z$. Note that after the previous step of including $\mathbf{s}^{\rm right}_Z$, if the row weight in region 1 is greater or equal to $d_Z$, then we are done because $\mathbf{s}^{\rm top_2}_Z$ has no support in region 1. So we only need to handle the case where the row weight in region 1 is decreased below $d_Z$ due to $\mathbf{s}^{\rm right}_Z$. By the previous argument, any row weight missing from region 1 must be compensated in region 2. Now, suppose $\mathbf{s}^{\rm top}_Z$ was able to remove this new row weight in region 2 such that the total row weight $|\tilde{\mathbf{r}}_{\mathrm{\Lambda}_{\rm L}}(\hat{\mathbf{z}}+\mathbf{s}^{\rm top_1}_Z + \mathbf{s}^{\rm right}_Z + \mathbf{s}^{\rm top_2}_Z)| = |\tilde{\mathbf{r}}_{\mathrm{\Lambda}_{\rm L}}(\mathbf{z})| < d_Z$. However, due to the parallel structure of $H_Z\otimes\ident$ in the top-type checks, it immediately implies the existence of a $\mathbf{s}^{\prime\mathrm{top}_1}_Z$ in region 1 such that $|\tilde{\mathbf{r}}_{\mathrm{\Lambda}_{\rm L}}(\hat{\mathbf{z}}+\mathbf{s}^{\rm top_1}_Z + \mathbf{s}^{\prime\mathrm{top}_1}_Z)| \equiv |\tilde{\mathbf{r}}_{\mathrm{\Lambda}_{\rm L}}(\mathbf{z}')| < d_Z$: simply take the checks in $\mathbf{s}^{\prime\mathrm{top}_1}_Z$ as well as their copies in the columns of region 1 and eliminate the row support in region 1 one column at a time, which is guaranteed if $\mathbf{s}^{\prime\mathrm{top}_2}_Z$ was able to remove row weight in region 2. Note that $\mathbf{z}' \equiv \hat{\mathbf{z}}+\mathbf{s}^{\rm top_1}_Z + \mathbf{s}^{\prime\mathrm{top}_1}_Z$ is related to $\hat{\mathbf{z}}$ by an element $\mathbf{s}^{\rm top_1}_Z + \mathbf{s}^{\prime\mathrm{top}_1}_Z \in \rs \tilde{H}_Z$ and so is a nontrivial logical $\overline{Z}$ operator itself. However, $|\tilde{\mathbf{r}}_{\mathrm{\Lambda}_{\rm L}}(\mathbf{z}')| < d_Z$ is in contradiction with Lemma \ref{lem:3D HGP logical support 2}, and so we must have $|\mathbf{z}|_{\mathrm{\Lambda}_{\rm L}} \geq |\tilde{\mathbf{r}}_{\mathrm{\Lambda}_{\rm L}}(\mathbf{z})| \geq d_Z$.

    For the $X$ sector, observe that all $X$-checks act as $H_X \otimes \ident$ in $\tilde{\mathrm{\Lambda}}_{\rm L}$ and hence as independent copies of $H_X$ among the columns, as illustrated in Figure \ref{fig:3D HGP layout}. Since each column is effectively an independent copy of the HGP code from the perspective of the $X$-checks, we can perform the analysis column by column. From \eqref{eq:3D HGP G_X}, recall that a nontrivial canonical $\overline{X}$ operator has support $\hat{\mathbf{x}}$ among $d_{\rm c}$ columns. Within each of these columns, $\hat{\mathbf{x}}$ is supported on at least $d_X$ rows. Now, Lemma \ref{lem:3D HGP logical support} tells us that $|\tilde{\mathbf{r}}_{\mathrm{\Lambda}_{\rm L}}(\hat{\mathbf{x}} + \mathbf{s}_X)| \geq d_X$ within each column for any $\mathbf{s}_X \in \rs \tilde{H}_X$. Since this inequality holds for each column independently, we conclude that $|\mathbf{x}|_{\mathrm{\Lambda}_{\rm L}} \geq d_X d_{\rm c}$.
\end{proof}

\begin{thm}[Effective distance preservation of automorphism gadgets in homological product codes, part 2] \label{thm:d_eff 3D HGP 2}
    Suppose physical qubit permutations do not spread errors, and we have a (quantum $\times$ classical) left-sector homological product code according to Definition \ref{defn:left-sector 3D homological product code} with automorphism gadgets given by \eqref{eq:3D HGP gadgets}. If the input quantum CSS code is a left-sector HGP code according to Definition \ref{defn:left-sector HGP code}, with distance-preserving automorphism gadgets given by \eqref{eq:HGP gadgets}, then the inherited automorphism gadgets in the homological product code will still be distance-preserving.
\end{thm}

\begin{proof}
    We will restrict our attention to $\tilde{U}_{\rm Q}$ in \eqref{eq:3D HGP gadgets}, since $\tilde{U}_{\rm c}$ falls within the scope of Theorem \ref{thm:d_eff 3D HGP 1}. In \eqref{eq:3D HGP gadgets}, note that the left-sector action of $\tilde{U}_{\rm Q}$ is given by $U \otimes \ident$, which has the interpretation of applying $U$ among the left-sector columns in parallel. If $U$ is an automorphism gadget of the input HGP code, then it further factorizes according to \eqref{eq:HGP gadgets}. Let $\mathbf{z} \in \ker \tilde{H}_X \backslash \rs\tilde{H}_Z$ $\big( \mathbf{x} \in \ker \tilde{H}_Z \backslash \rs\tilde{H}_X \big)$ be the support of any nontrivial logical $\overline{Z}$ ($\overline{X}$) operator of the left-sector homological product code. Then Propositions \ref{prop:left-sector homological product code distance} and \ref{prop:3D HGP restricted logical sector weight} tell us that
    \begin{align}
        \abs{\mathbf{z}}_{\mathrm{\Lambda}_{\rm L}} \geq \tilde{d}_Z \quad,\quad \abs{\mathbf{x}}_{\mathrm{\Lambda}_{\rm L}} \geq \tilde{d}_X \, .
    \end{align}
    Since the input automorphism gadget only acts as a permutation within the $\mathrm{\Lambda}_{\rm L}$ rows of $\tilde{\mathrm{\Lambda}}_{\rm L}$, at least $\tilde{d}_Z$ elementary faults are required to incur a logical $\overline{Z}$ error. Similarly, at least $\tilde{d}_X$ elementary faults are required to incur a logical $\overline{X}$ error.
\end{proof}

\subsection{Fault tolerance for homological product codes: quantum $\times$ quantum}
\label{sec:FT quantum x quantum}

Before we begin our analysis on the fault tolerance of our automorphism gadgets for (quantum $\times$ quantum) homological product codes, we first discuss the code parameters of such homological product codes. These parameters would directly set the baseline for how we can expect our gadgets to perform fault-tolerantly.

We first state a few well-known techniques to modify the parity-check matrices of classical linear codes that we use to prove some facts regarding the parameters for our (quantum $\times$ quantum) homological product codes.

\begin{defn}[Punctured code]\label{defn:punctured-code}
    Let $H \in \mathbb{F}_2^{m \times n}$ be the parity-check matrix of a classical binary linear code $\mathcal{C}$. In addition, let $B \subseteq [n]$ be a set of column indices. Then, we say that the classical linear code $\mathcal{C}$ is punctured with respect to $B$ when its parity-check matrix is now given by the submatrix $H|_{B} \coloneqq H[[m], B] \in \mathbb{F}_2^{m \times |B|}$.    
\end{defn}

In other words, when we puncture a code with respect to a set of column indices, we are dropping all the columns in the matrix other than the columns corresponding to the set of indices.

\begin{defn}[Shortened code]\label{defn:shortened-code}
Let $G \in \mathbb{F}_2^{k \times n}$ be the generator matrix of a classical binary linear code $\mathcal{C}$. In addition, let $B \subseteq [n]$ be a set of column indices. Then, we say that the classical linear code $\mathcal{C}$ is shortened with respect to $B$ when its generator is now given by the submatrix $G|^{B} \coloneqq G[A, B] \in \mathbb{F}_2^{|A| \times |B|}$ for some $A \subseteq [k]$ such that $G[\{a\}, [n]\setminus B] = \mathbf{0}$ for any $a \in A$.    
\end{defn}

Because we mainly work with parity-check matrices instead of generator matrices, we abuse the definition and replace $G$ with $H$ in Definition~\ref{defn:shortened-code} when we say we shorten $H$ with respect to some set of column indices. Recall that if $H$ is the parity-check matrix of the classical code $\mathcal{C}$ that has a generator matrix $G$, then $H$ is the generator matrix of the classical code $\mathcal{C}^\perp$ that is dual to $\mathcal{C}$. Thus, one can also consider shortening $H$ as the equivalent of shortening the dual code $\mathcal{C}^\perp$. 

Lastly, we define how we can augment the set of parity-checks imposed by a parity-check matrix $H$ for a classical linear code $\mathcal{C}$.

\begin{defn}[Augmented checks]\label{defn:augmented-checks}
    Let $H \in \mathbb{F}_2^{m \times n}$ be the parity-check matrix of a classical binary linear code $\mathcal{C}$. Let $S  = \left\{\mathbf{v}^\transpose\;|\;\mathbf{v} \in \mathbb{F}_2^{n}\right\}$ be a set of row vectors that correspond to new additional parity-checks that act on $n$ bits. Then, the augmented parity-check matrix $H^{+S}$ is the new parity-check matrix of a classical linear code $\mathcal{C}'$ with the set of rows $S$ added to the set of rows in $H$. Unless otherwise stated, the set of rows $S$ is appended at the bottom of the matrix $H$ i.e.,
    \[H^{+S} \coloneqq \begin{pmatrix}
        H \\ S
    \end{pmatrix}.\]
\end{defn}

\subsubsection{Code parameters for homological product codes: quantum $\times$ quantum}\label{sec:code-params-quantum-times-quantum}

Recall that we are only interested in the middle logical qubits of the quantum $\times$ quantum homological product code by our construction in Section~\ref{sec:quantumxquantum}.
To be more explicit, we are interested in keeping the logical qubits corresponding to the cohomology $\mathcal{H}^0[{\rm Q}] \otimes \mathcal{H}^0[{\rm Q}']$ for quantum codes ${\rm Q}$ and ${\rm Q}'$ and treating the other cohomologies that correspond to $\mathcal{H}^i[{\rm Q}] \otimes \mathcal{H}^{-i}[{\rm Q}']$ as gauge qubits for $i \in \mathbb{Z} \setminus \{0\}$. By a simple evaluation of the rank of the cohomology $\mathcal{H}^0[{\rm Q}] \otimes \mathcal{H}^0[{\rm Q}']$, we see that the number of logical qubits for the homological product code $\tilde{{\rm Q}}$ is given by
\[\mathrm{rank}\,\left(\frac{\ker H_Z}{\mathrm{rs} H_X}\right) \cdot \mathrm{rank}\,\left(\frac{\ker H'_Z}{\mathrm{rs} H'_X}\right) = k^{}_{\rm Q} \cdot k'_{\rm Q}.\]
Assuming that ${\rm Q}$ and ${\rm Q}'$ are both associated to 3-term chain complexes, we have determined that this ``middle-sector'' homological product code $\tilde{{\rm Q}}$ corresponding to ${\rm Q}$ and ${\rm Q}'$ has the following parameters:
\begin{align}
    \tilde{n} &= n^{}_{\rm Q} \cdot n'_{\rm Q} + m_Z' \cdot m^{}_X + m^{}_Z \cdot m_X' \\
    \tilde{k} &= k^{}_{\rm Q} \cdot k'_{\rm Q}
\end{align}

Firstly, we make an argument in the following lemma that dressing our $\overline{X}$ and $\overline{Z}$ logical operators for our middle-sector homological product code with its gauge operators does not reduce the Hamming weight of our logical operators:

\begin{lem}[Dressed logical operators]\label{lem:dressed-logical-ops}
    Suppose we have a $\llbracket n_{\rm Q}, k_{\rm Q}, (d_X, d_Z)\rrbracket$ quantum code ${\rm Q}$ and a $\llbracket n'_{\rm Q}, k'_{\rm Q}, (d_X', d_Z')\rrbracket$ quantum code ${\rm Q}'$ with parity-check matrices $H_X, H_Z, H'_X, H'_Z$ respectively. Let $\tilde{{\rm Q}}$ be a middle-sector CSS code constructed by taking the homological product code between ${\rm Q}$ and ${\rm Q}'$. Then, dressing the $\overline{X}$ logical operators from $\mathcal{H}^0[{\rm Q}] \otimes \mathcal{H}^0[{\rm Q}']$ and the $\overline{Z}$ logical operators from $\mathcal{H}_0[{\rm Q}] \otimes \mathcal{H}_0[{\rm Q}']$ strictly increases their Hamming weights.
\end{lem}
\begin{proof}
    We begin by formulating a canonical logical operator basis for the $X$ gauge operators. Recall that our set of $X$ gauge operators is the set of spurious cohomologies $\left\{\mathcal{H}^i[{\rm Q}] \otimes \mathcal{H}^{-i}[{\rm Q}']\right\}_{i \in \mathbb{Z} \setminus \{0\}}$. Each of these spurious cohomologies is supported on the physical qubits that lie in the vector space ${\rm Q}^i \otimes {\rm Q}^{\prime\,-i}$ for $i \neq 0$ respectively. Note that ${\rm Q}_i$ is simply the $i$\textsuperscript{th} $\mathbb{F}_2$-vector space in the chain complex ${\rm Q}$. In other words, when written in their canonical forms, these $X$ gauge operators do not share support with the logical $\overline{X}$ operators since the logical $\overline{X}$ operators are completely supported on the qubits in the vector space ${\rm Q}^0 \otimes {\rm Q}^{\prime\,0}$. Therefore, dressing the logical $\overline{X}$ operators with an arbitrary $X$ gauge operator only increases the support of the logical operator on a disjoint set of qubits. The argument for the dressed logical $\overline{Z}$ operators follows in the same way.
\end{proof}

Now, we proceed to state the following theorem that bounds the minimum distances of the (dressed) logical operators of the middle-sector homological product code $\tilde{{\rm Q}}$. In the theorem statement and its proof, we use $\tilde{{\rm Q}}$ to represent the quantum code and its associated chain complex interchangeably. It should be clear from the context which object we are referring to. Our proof of Proposition~\ref{prop:distance-bounds-subsystem-hom-prod-code} uses techniques introduced in Refs.~\cite{Tillich_2014, Zeng_2019, tan2024hgp}. 

\begin{prop}[Distance bounds for middle-sector homological product codes]\label{prop:distance-bounds-subsystem-hom-prod-code}
    Suppose we have a $\llbracket n_{\rm Q}, k_{\rm Q}, (d_X, d_Z)\rrbracket$ quantum code ${\rm Q}$ and a $\llbracket n'_{\rm Q}, k'_{\rm Q}, (d_X', d_Z')\rrbracket$ quantum code ${\rm Q}'$ with parity-check matrices $H_X, H_Z, H'_X, H'_Z$ respectively. Let $\tilde{{\rm Q}}$ be a subsystem CSS code constructed by taking the homological product code between ${\rm Q}$ and ${\rm Q}'$. Then, the 0-cosystolic and 0-systolic distances of the middle-sector homological product code are bounded by:
    \[\max(d^{}_X, d'_X) \leq d^0(\tilde{{\rm Q}}) \leq d^{}_X \cdot d'_X,\]
    \[\max(d^{}_Z, d'_Z) \leq d_0(\tilde{{\rm Q}}) \leq d^{}_Z \cdot d'_Z\]
    when we disregard the 0-cohomologies and 0-homologies introduced by the $\overline{X}$ and $\overline{Z}$ gauge operators.
    In addition, let $\mathbf{x}$ and $\mathbf{z}$ be arbitrary non-trivial dressed $\overline{X}$ and $\overline{Z}$ logical operators for $\tilde{{\rm Q}}$ respectively. Then, we have:
    \[\abs{\mathbf{x}_{\rm M}} \geq \max(d^{}_X, d'_X),\]
    \[\abs{\mathbf{z}_{\rm M}} \geq \max(d^{}_Z, d'_Z).\]
\end{prop}
\begin{proof}
    To prove the proposition's statement, we focus on establishing the distance bounds for the cohomology operators i.e., $\overline{X}$ logical operators. The arguments put forward will also hold for the homology operators i.e., $\overline{Z}$ logical operators.
    Recall that the logical operator basis for the middle-sector homological product code $\tilde{{\rm Q}}$ is given by Eq.~\ref{eq:4D HGP G_X} which we restate below for convenience:
    \[\tilde{G}_{X,{\rm M}} = \left(\, \mathbf{0} \;\big|\; G^{}_X \otimes G'_X \;\big|\; \mathbf{0} \,\right) = \left(\begin{array}{c|c|c}\mathbf{0} & \frac{\ker H_Z}{\mathrm{rs}\, H_X} \otimes \frac{\ker H_Z'}{\mathrm{rs}\,H_X' } & \mathbf{0} \end{array}\right).\]
    Take  will have support on at least $d^{}_X \cdot d_X'$ physical qubits from the expression above. Thus, we trivially obtain the upper bound $d^0(\tilde{{\rm Q}}) \leq d^{}_X \cdot d_X'$. 
    
    Next, we proceed to show that $\abs{\mathbf{x}_{\rm M}} \geq \max(d^{}_X, d_X')$ for an arbitrary nontrivial $\overline{X}$ logical operator $\mathbf{x}$. Note that this directly implies that $d^0(\tilde{{\rm Q}}) \geq \max(d_X, d_X')$ since $\abs{\mathbf{x}_M} \leq d^0(\tilde{{\rm Q}})$ in the worst case. To show that $\abs{\mathbf{x}_M} \geq \max(d_X, d_X')$, we prove that $\abs{\mathbf{x}_M}_r \geq d_X$ and $\abs{\mathbf{x}_M}_c \geq d_X'$ for a non-trivial $\overline{X}$ logical operator. 
    
    Suppose $\abs{\mathbf{x}_M}_r < d_X$. Let $R_M \subseteq [n_{\rm Q}]$ such that for all $r \in R_M$, there exists some $c \in [n_{\rm Q}']$ such that $\mathbf{x}_M[\{r\},\{c\}]$ is non-trivial. Note that $\abs{\mathbf{x}_M}_r = \abs{R_M}$.
    Now, let us puncture $H_Z$ with respect to $R_M$ to obtain $H_Z|_{R_M}$. Recall that this means that we have dropped all columns of $H_Z$ with indices that lie outside of $R_M$. We also shorten $H_X$ with respect to $R_M$ to obtain $H_X|^{R_M}$. Recall that this means that we have dropped all rows of $H_X$ with non-trivial support on columns in $R_M$ before dropping all columns in the resulting matrix with indices that lie outside of $R_M$. Suppose we have dropped $\ell$ rows from $H_X$ in the process of shortening. Then, we add $S$, i.e. $\ell$ rows of all zeros, to $H_X|^{R_M}$ at the indices where they were dropped to obtain a shortened, augmented check matrix $\left(H_X|^{R_M}\right)^{+S} \in \mathbb{F}_2^{m_X \times \abs{R_M}}$. Note that the orthogonality condition is still satisfied i.e., $H_X|^{R_M} H_Z|_{R_M}^{\transpose} = \left(H_X|^{R_M}\right)^{+S} H_Z|_{R_M}^{\transpose} = 0$. 
    
    Let us construct a chain complex (CSS code) ${\rm Q}[R_M]$ with the boundary maps $\left(H_X|^{R_M}\right)^{+S}$ and $H_Z|_{R_M}$. 
    Because $\abs{R_M} = \abs{\mathbf{x}_M}_r < d_X$, we have a trivial $0$-cohomology $\frac{\ker H_Z|_{R_M}}{\mathrm{rs}\,H_X|^{R_M}}$ for ${\rm Q}[R_M]$. Since $\mathrm{rs}\,H_X|^{R_M} = \left(H_X|^{R_M}\right)^{+S}$, we also obtain a trivial 0-cohomology group  $\frac{\ker H_Z|_{R_M}}{ \mathrm{rs}\,\left(H_X|^{R_M}\right)^{+S}}$ for our quantum code ${\rm Q}[R_M]$.
    Now, let us construct a new subsystem homological product code $\tilde{{\rm Q}}[R_M]$ from ${\rm Q}[R_M]$ and ${\rm Q}'$ that has check matrices $\tilde{H}_X[R_M]$ and $\tilde{H}_Z[R_M]$.   
    We point out that our non-trivial $\overline{X}$ logical operator $\mathbf{x}$ is well-defined on the physical qubits of $\tilde{{\rm Q}}[R_M]$ since the space of 0-cells in both $\tilde{{\rm Q}}$ and $\tilde{{\rm Q}}[R_M]$ are completely identical. In other words, the qubit vector space can be decomposed into the same sets of qubit subspaces. In addition, we have $\tilde{H}_Z[R_M] \mathbf{x} = 0$.

    Next, we claim that the subsystem homological product code $\tilde{{\rm Q}}[R_M]$ encodes no logical qubit. Using Proposition~\ref{prop:kunneth-formula} and discarding all spurious cohomologies other than $\mathcal{H}^0[{\rm Q}[R_M]] \otimes \mathcal{H}^0[{\rm Q}']$, we see that \[\mathrm{rank}\, \mathcal{H}^0[\tilde{{\rm Q}}[R_M]] = \mathrm{rank}\,\mathcal{H}^0[{\rm Q}[R_M]] \cdot \mathrm{rank}\,\mathcal{H}^0[{\rm Q}'] = 0 \cdot k' = 0.\]
    In other words, $\mathbf{x}$ can be obtained from a linear combination of checks from $\tilde{H}_X[R_M]$ and $X$ gauge operators from the discarded spurious 0-cohomologies. Then, 
    \[\mathbf{x}^\transpose = \alpha^\transpose \tilde{H}_X[R_M] + \mathbf{x}_{\mathrm{gauge}}^{\transpose}\]
    for some $\alpha \in \mathbb{F}_2^{\abs{R_M} m'_X + m_X n'_{\rm Q}}$ and some $\mathbf{x}_{\mathrm{gauge}}$ that is some linear combination of $\overline{X}$ gauge operators in the canonical basis.   
    Note that there are many choices of $\alpha$ since we have introduced $\ell$ rows of zeros when we augmented the shortened check matrix $H_X|^{R_M}$. We choose $\alpha$ to be any such vector with trivial support in the rows with indices that directly interact with the row indices of the shortened, augmented check matrix for those $\ell$ rows in the larger $\tilde{H}_X[R_M]$ check matrix.
    Using Lemma~\ref{lem:dressed-logical-ops}, we know that our spurious 0-cohomologies can be written in a canonical basis that has trivial support on the qubit vector subspace ${\rm Q}[R_M]^0 \otimes {\rm Q}^{\prime\,0}$. This implies that $\mathbf{x}_M$ is obtained completely from the linear combination of $X$ stabilizer generators in $\tilde{H}_X[R_M]$ i.e., $\mathbf{x}_M^\transpose = \alpha^\transpose \tilde{H}_X[R_M].$
    Let us decompose $\alpha$ into $\begin{pmatrix}
        \alpha[T] \\ \alpha[B]
    \end{pmatrix}$ for $\alpha[T] \in \mathbb{F}_2^{m_X n_{\rm Q}'}$ and $\alpha[B] \in \mathbb{F}_2^{\abs{R_M} m'_X}$. Using Eq.~\ref{eq:4D HP H_X}, we have 
    \[\alpha[B]^\transpose (\ident_{\abs{R_M}} \otimes H'_X) + \alpha[T]^\transpose (\left(H_X|^{R_M}\right)^{+S} \otimes \ident_{n'_{\rm Q}}) = \mathbf{x}_M^\transpose.\]
    Let us denote $\alpha' \in \mathbb{F}_2^{n_{\rm Q} m_X' + m_X n'_{\rm Q}}$ that has the decomposition $\begin{pmatrix}
        \alpha[T] \\ \alpha'[B]
    \end{pmatrix}$ for the same $\alpha[T]$ and $\alpha'[B] \in \mathbb{F}_2^{n_{\rm Q} m'_X}$ is obtained by padding $\alpha[B]$ appropriately with zeros at the row indices re-introduced when we increase the dimension from $|R_M|$ to $n_{\rm Q}$. It is not too hard to see that by our construction of the shortened, augmented check matrix $\left(H_X|^{R_M}\right)^{+S}$ and $\alpha$, we have
    \[\alpha^{\prime\,\transpose} \tilde{H}_X = \alpha'[B]^\transpose (\ident_{n_{\rm Q}} \otimes H'_X) + \alpha[T]^\transpose (H_X \otimes \ident_{n'_{\rm Q}}) = \mathbf{x}_M^\transpose,\]
    \[\alpha^{\prime\,\transpose} \tilde{H}_X + \mathbf{x}_{\mathrm{gauge}}^{\transpose} = \mathbf{x}^{\transpose}.\]
    Since $\mathbf{x}$ can be constructed from a linear combination of $X$ stabilizer generators and gauge operators when $\abs{\mathbf{x}_M}_r < d_X$, we have shown that $\abs{\mathbf{x}_M}_r \geq d_X$ for $\mathbf{x}$ to be a non-trivial $\overline{X}$ logical operator. The same proof strategy would also work for showing that $\abs{\mathbf{x}_M}_c \geq d'_X$. This gives us $\abs{\mathbf{x}_M} \geq \max(d_X, d_X')$ which naturally implies that $d^0(\tilde{{\rm Q}}) \geq \max(d_X, d_X')$ when we consider the relevant 0-cohomologies in our middle-sector homological product code. The arguments above can also be used to show the same bounds for the dressed $\overline{Z}$ logical operators i.e., 0-homologies.
\end{proof}

Proposition~\ref{prop:distance-bounds-subsystem-hom-prod-code} is related to the result described in Ref.~\cite{Zeng_2020}. In the work done by Zeng and Pryadko, they provided upper and lower bounds for the (co)systolic distances of the tensor product of chain complexes. However, they only stated the bounds for the case where one of the chain complexes only has two non-trivial vector spaces and explored the case where some of the qubit vector spaces are projected away. Projecting away these qubit vector subspaces is necessary in their case because the shortening of $H_X$ would have decreased the dimension of the qubit subspaces eg. ${\rm Q}^{1} \otimes {\rm Q}^{\prime\,-1}$ which could have potentially rendered the nontrivial logical $\overline{X}$ operator to be undefined in that qubit subspace if it was supported on the physical qubits that have now disappeared as a result of the shortening. 

In our result, we have kept all the qubit vector spaces and considered chain complexes of arbitrary lengths; however, we have also chosen to discard the spurious (co)homologies. We were able to sidestep the technical issue that Zeng and Pryadko had by considering the augmentation of the $H_X$ parity-check matrix to ensure that the dimensions of all qubit vector subspaces are preserved so that the nontrivial logical operator is still well-defined for the proof strategy to go through. Our result can also be understood as using two separate applications of Theorem~\ref{thm:d_eff 3D HGP 1} from Ref.~\cite{tan2024hgp}. In Figure~\ref{fig:schematic-quantum-quantum}, we have shaded two sets of six squares where one of the set of squares is shaded in blue and the other set of squares is shaded in red.
Each set of six squares imposes a requirement on the qubits in $M$ to have either column weight or row weight of $d_Z$ or $d_Z'$ in order for a Pauli $Z$ operator to be a nontrivial logical operator of the middle-sector homological product code. This intuition guided the proof of Proposition~\ref{prop:distance-bounds-subsystem-hom-prod-code}. While the presence of two separate square blocks of checks that act row-wise and column-wise on the a single square block of qubits was able to constrain any non-trivial logical operator such that its support on that single block of qubits has weight at least $d_Z \cdot d_Z'$ for the quantum $\times$ classical homological product code, this is not the case for the quantum $\times$ quantum case because we also have two separate square blocks of checks that can potentially ``erase'' the support on the qubits in that single block by propagating them to the other left and right qubit blocks in an adversarial way as shown in Figure~\ref{fig:schematic-quantum-quantum}. Therefore, it is rather challenging to remove the gap between the lower and upper bound on the (co)systolic distances for the middle-sector homological product code in the most general case.

\usetikzlibrary{patterns}
\begin{figure}
  \centering
  \begin{tikzpicture}
    \draw (0,0) rectangle (3,3); 
    \node at (1.5,1.5) {$\mathbf{R_X}$};
    \draw (3.4,3.4) rectangle (6.4,6.4); 
      \pattern[pattern=north east lines, pattern color=red] (3.4,3.4)--(3.4,6.4)--(6.4,6.4)--(6.4,3.4)--cycle;
              \pattern[pattern=north west lines, pattern color=blue] (3.4,3.4)--(3.4,6.4)--(6.4,6.4)--(6.4,3.4)--cycle;
    \node at (4.9,4.9) {$\mathbf{M}$ \textbf{qubits}};
    \draw (3.4,0) rectangle (6.4,3); 
    \pattern[pattern=north east lines, pattern color=red] (3.4,0)--(3.4,3)--(6.4,3)--(6.4,0)--cycle;
    \node at (4.9,1.5) {$\mathbf{M_X[B]}$};
		\draw (3.4,6.8) rectangle (6.4,9.8); 
        \pattern[pattern=north east lines, pattern color=red] (3.4,6.8)--(3.4,9.8)--(6.4,9.8)--(6.4,6.8)--cycle;
                \pattern[pattern=north west lines, pattern color=blue] (3.4,6.8)--(3.4,9.8)--(6.4,9.8)--(6.4,6.8)--cycle;
			\node at (4.9,8.3) {$\mathbf{M_Z[T]}$};
    \draw (6.8,3.4) rectangle (9.8,6.4); 
    \pattern[pattern=north east lines, pattern color=red] (6.8,3.4)--(6.8,6.4)--(9.8,6.4)--(9.8,3.4)--cycle;
        \pattern[pattern=north west lines, pattern color=blue] (6.8,3.4)--(6.8,6.4)--(9.8,6.4)--(9.8,3.4)--cycle;
    \node at (8.3,4.9) {$\mathbf{M_Z[B]}$};
    \draw (6.8,0) rectangle (9.8,3); 
    \pattern[pattern=north east lines, pattern color=red] (6.8,0)--(6.8,3)--(9.8,3)--(9.8,0)--cycle;
    \node at (8.3,1.5) {$\mathbf{R}$ \textbf{qubits}};
		\draw (6.8,6.8) rectangle (9.8,9.8); 
        \pattern[pattern=north east lines, pattern color=red] (6.8, 6.8)--(6.8,9.8)--(9.8,9.8)--(9.8,6.8)--cycle;
        \pattern[pattern=north west lines, pattern color=blue] (6.8, 6.8)--(6.8,9.8)--(9.8,9.8)--(9.8,6.8)--cycle;
			\node at (8.3,8.3) {$\mathbf{R_Z}$};
    \draw (0,3.4) rectangle (3,6.4); 
    \pattern[pattern=north west lines, pattern color=blue] (0, 3.4)--(3,3.4)--(3,6.4)--(0,6.4)--cycle;
    \node at (1.5,4.9) {$\mathbf{M_X[T]}$};
		\draw (0,6.8) rectangle (3,9.8); 
        \pattern[pattern=north west lines, pattern color=blue] (0, 6.8)--(3,6.8)--(3,9.8)--(0,9.8)--cycle;
    \node at (1.5,8.3) {$\mathbf{L}$ \textbf{qubits}};
    \draw (-0.5,0) edge[-] (-0.5,3); 
    \node at (-0.8,1.5) {$m_X$};
    \draw (-0.5,3.4) edge[-] (-0.5,6.4); 
    \node at (-0.8,4.9) {$n_{\rm Q}$};
		\draw (-0.5,6.8) edge[-] (-0.5,9.8); 
    \node at (-0.8,8.3) {$m_Z$};
    \draw (0,-0.5) edge[-] (3,-0.5); 
    \node at (1.5,-0.75) {$m_X'$};
    \draw (3.4,-0.5) edge[-] (6.4,-0.5); 
    \node at (4.9,-0.75) {$n_{\rm Q}'$};
    \draw (6.8,-0.5) edge[-] (9.8,-0.5); 
    \node at (8.3,-0.75) {$m_Z'$};
  \end{tikzpicture}
  \caption{Schematic for the quantum $\times$ quantum homological product code constructed from quantum CSS codes with 3-term chain complexes. 
  For each of the square grid in the figure, we can identify the relationship between the square grid and its adjacent neighbors from the linear (co)boundary maps inherited from the constitutent quantum codes.
  For example, the grid labeled with $M$ qubits contains $n_{\rm Q} \cdot n_{\rm Q}'$ physical qubits. For each of the $n_{\rm Q}$ rows, we have $X$ checks from $M_X[T]$ and $Z$ checks from $M_Z[B]$ acting on a row of $n'_{\rm Q}$ physical qubits in $M$.
  There are thus effectively $n_{\rm Q}$ copies of the quantum code ${\rm Q}'$ that can be identified with these $n_{\rm Q}$ rows. 
  $M_X[T]$ (corr. $M_Z[B]$) refers to the $X$ (corr. $Z$) stabilizer generators that form the top (corr. bottom) row of the generator matrix in Eq.~\ref{eq:4D HP H_X} (corr. Eq.~\ref{eq:4D HP H_Z}).
  The blue and red shaded squares in the figure allows us to ``approximate'' some of the interactions between the vector spaces of the product chain complex with two separate instances of a quantum $\times$ classical homological product code which each has six vector spaces.
 }
 \label{fig:schematic-quantum-quantum}
\end{figure}
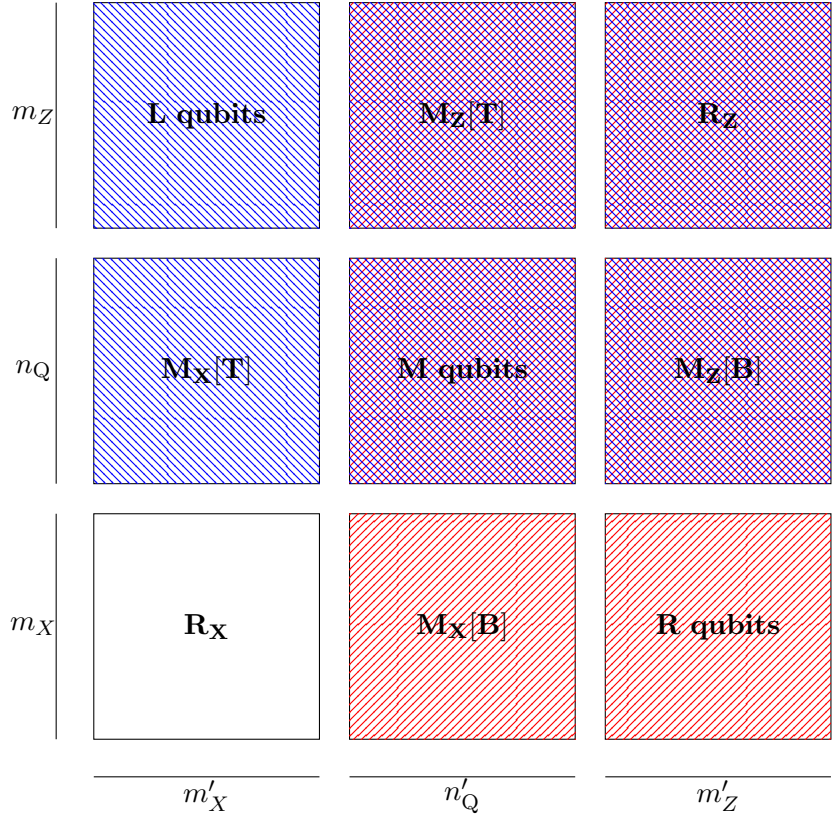

\subsubsection{Discussion on fault tolerance for quantum $\times$ quantum homological product codes}
Based on the discussion in Section~\ref{sec:code-params-quantum-times-quantum}, we know that any non-trivial dressed $\overline{X}$ logical operator of a middle-sector homological product code has to have support of weight at least $\max(d_X, d_X')$ on the $M$ qubits i.e., the vector subspace ${\rm Q}^0 \otimes {\rm Q}^{\prime\,0}$. In the event where the dressed distance of the middle-sector homological product code is indeed $\max(d_X, d_X')$, then we achieve the same kind of fault-tolerance discussed in the previous sections since the small circuits that can introduce faults only act on physical qubits that are not in the vector subspace ${\rm Q}^0 \otimes {\rm Q}^{\prime\,0}$. On the other hand, if the dressed distance of the middle-sector homological product code saturates the upper bound of $d_X \cdot d_X'$, then it is possible for our logical operators to have significant support outside of the vector subspace ${\rm Q}^0 \otimes {\rm Q}^{\prime\,0}$. This means that faults from the small circuits constructed in Section~\ref{sec:auto-gadget-quantum-times-quantum-hom-prod-code} that act on these other qubit subspaces can potentially contribute to logical errors in a way that is not inherently fault-tolerant.

We now use the bounds obtained in Proposition~\ref{prop:distance-bounds-subsystem-hom-prod-code} to explicitly bound the extent of the fault tolerance of our automorphism gadgets of our (quantum $\times$ quantum) middle-sector homological product code. We state it in the following theorem and proof for completeness.

\begin{thm}[Effective distance bounds for automorphism gadgets in (quantum $\times$ quantum) middle-sector homological product codes] \label{thm:d_eff 4D HGP}
    Suppose physical qubit permutations are noise-free, and we have a $\llbracket n,k,d \rrbracket$ middle-sector homological product code with automorphism gadgets given by \eqref{eq:4D HGP gadgets}. If the input gadgets of the input quantum CSS codes are distance-preserving, then the new automorphism gadgets of the middle-sector homological product code preserve the $X$ and $Z$ fault distances up to $\max(d_X, d_X')$ and $\max(d_Z, d_Z')$ respectively.
\end{thm}

\begin{proof}
    We will analyze one of the two automorphism gadgets \eqref{eq:4D HGP gadgets} for the $X$ fault distance because the same argument would hold for the other automorphism gadget as well as the $Z$ fault distance. Consider the gadget $\tilde{U}_{\rm Q}$ stated in \eqref{eq:4D HGP U_Q1}. The $U_1 \otimes \ident$ acting on the middle physical qubits may not strictly be a permutation. Again, Proposition~\ref{prop:distance-bounds-subsystem-hom-prod-code} tells us that the middle-sector logical weight is bounded from below by $\max(d_X, d_X')$ , and so we just need to also check that the $U \otimes \ident$ acting on the ``columns'' of the quantum CSS code does not spread errors in a dangerous way. By assumption, $U_1$ does not decrease the fault distance of the input quantum CSS code below its code distance. In other words, at least $\max(d_X, d_X')$ elementary faults are still required in the automorphism gadget to generate any dressed $\overline{X}$ logical operators \eqref{eq:4D HP G_X,G_Z} living in columns of the middle sector. Thus, at least $\max(d_X, d_X')$ faults are required to effect a nontrivial logical $\overline{X}$, and at least $\max(d_Z, d_Z')$ faults are required to effect a nontrivial logical $\overline{Z}$.
\end{proof}

In Ref.~\cite{Zeng_2020}, the following conjecture was proposed:

\begin{conj}[Systolic distances of tensor product of bounded chain complexes~{\cite[Restatement of Conjecture 18]{Zeng_2020}}] \label{conj:systolic-distances}
    The $i$-systolic distance $d_i(\mathcal{A} \otimes \mathcal{B})$ in a tensor product of any pair of bounded chain complexes $\mathcal{A}$ and $\mathcal{B}$ of vector spaces over a finite field is given by 
    \[d_i(\mathcal{A} \otimes \mathcal{B}) = \min_{j\in \mathbb{Z}}d_j(\mathcal{A})d_{i-j} (\mathcal{B}).\]
\end{conj}
When we consider the middle-sector version of our homological product code, we observe that the conjectured $0$-(co)systolic distances in Ref.~\cite{Zeng_2020} matched the upper bound for the distance of our middle-sector homological product code as stated in Proposition~\ref{prop:distance-bounds-subsystem-hom-prod-code}. In addition, Zeng and Pryadko noted in their paper that they performed extensive numerical simulations and were unable to find a single instance where Conjecture~\ref{conj:systolic-distances} was not true. Thus, it could very much just be a proof technique issue that is preventing us from showing $d^0(\tilde{{\rm Q}}) = d_X \cdot d_X'$ and $d_0(\tilde{{\rm Q}}) = d_Z \cdot d_Z'$ in Proposition~\ref{prop:distance-bounds-subsystem-hom-prod-code}. Because the proof that lower bounds the weight of any dressed $\overline{X}$ and $\overline{Z}$ logical operator in the ${\rm Q}^0 \otimes {\rm Q}^{\prime\,0}$ physical qubit vector space is closely related to the lower bound on the effective distance as a result of the automorphism gadget for our quantum$\times$ quantum homological product code, a proof for $d^0(\tilde{{\rm Q}}) = d_X \cdot d_X'$ and $d_0(\tilde{{\rm Q}}) = d_Z \cdot d_Z'$ could potentially imply that any dressed logical operator would also have weight at least $d_X \cdot d_X'$ or $d_Z\cdot d_Z'$ in the ${\rm Q}^0 \otimes {\rm Q}^{\prime\,0}$ physical qubit vector space. Thus, we conjecture the following:

\begin{conj}[Effective distance preservation of automorphism gadgets in (quantum $\times$ quantum) middle-sector homological product codes] \label{conj:fault-tolerance}
    Suppose physical qubit permutations do not spread errors, and we have a $\llbracket n,k,d \rrbracket$ middle-sector homological product code with automorphism gadgets given by \eqref{eq:4D HGP gadgets}. If the gadgets of the input quantum CSS codes preserve the effective distance, then so will the new automorphism gadgets of the homological product code.
\end{conj}

\subsection{Alleviating correlated errors}
\label{sec:alleviating correlated errors}

Theorems~\ref{thm:d_eff 2D HGP}, \ref{thm:d_eff 3D HGP 1}, \ref{thm:d_eff 3D HGP 2} and \ref{thm:d_eff 4D HGP} are statements regarding the fault tolerance of our automorphism gadgets.
Although, in the most general case, we have to implement a quantum circuit on a subset of the physical qubits, for HGP and left-sector homological product codes this subsystem circuit does not spread errors in a dangerous way that decreases the effective circuit-level distance of the code, under the assumption of free permutations. Of course, in an actual quantum processor, physical qubit permutations may introduce noise due to crosstalk and heating from movement operations. In addition, idling errors that were negligible for short times may accumulate and become non-negligible when accounting for the time to implement the subsystem circuit portion of our automorphism gadgets. Because of these reasons, we emphasize that the effective distance preservation stated in Theorems \ref{thm:d_eff 2D HGP} and \ref{thm:d_eff 3D HGP 1} is \emph{only relevant} when our simplified noise model can well-approximate the actual noise in the machines. In addition to idling errors, a potentially deep circuit on the right sector may cause small errors to propagate to large ones and confuse our classical decoding algorithm. Although our gadgets are distance-preserving, and so these large errors are \emph{detectable}, it could be the case that there exists a logical operator with a large fraction of its support on the right sector. If the right sector has a large density of correlated errors, a simple minimum-weight decoder may potentially infer a wrong correction and incur a logical error.

The first, and perhaps most pressing, question to ask is when do we need to actually implement the subsystem circuit $W$? Suppose, for instance, that our fault-tolerant architecture is comprised of multiple code blocks that store all of our logical qubits for computation, and we are performing logical gates using a combination of automorphism gadgets, measurement (e.g. lattice surgery) gadgets, state-injection gadgets and transversal gates. For simplicity, we will examine the following cases where an automorphism gadget on a code block is immediately proceeded by:
\begin{enumerate}
    \item A transversal measurement of the entire code block
    \item A logical Pauli measurement
    \item An interblock transversal gate
\end{enumerate}

In case 1, a transversal measurement involves measuring every data qubit in the block in the $X$ or $Z$ basis. Such a measurement can be used to fault-tolerantly measure all logical qubits in the block in the $\overline{X}$ or $\overline{Z}$ bases and works for any CSS code. The statistics of the transversal measurement is both invariant to any permutation as well as any CNOT circuit. The key observation is that any permutation or CNOT circuit does not change the Pauli type of an observable, and so the outcome of any observable consisting of Pauli $X$s or Pauli $Z$s before the automorphism gadget can be reconstructed by tracking how the observable would change through the gadget and combining the appropriate single-qubit measurement outcomes. For left-sector logical Pauli operators in our homological product codes, this tracking would involve a simple relabeling. So up to an in-software relabeling of the data qubits, we do not need to physically implement the automorphism gadget at all prior to the transversal measurement; in a sense, the transversal measurement ``absorbs'' the automorphism gadget. 

For case 2, we would like to fault-tolerantly measure a logical Pauli operator, such as a single $\overline{X}_i$ or a joint $\overline{X}_i\overline{X}_j$, without disturbing the uninvolved logical qubits. Such an addressable logical Pauli measurement can be used to design addressable logical Clifford gates or perform targeted state injection. For intrablock logical operators, we can use either the general-purpose ``adapters'' of \cite{swaroop2025} or ``extractors'' of \cite{he2025extr}; for interblock logical operators, we can also use the related ``bridge'' systems presented in both references. For these schemes, the specific ancillary system that is attached to the code block depends on the weight(s) of the logical operator(s) being measured as well as the structure of the code stabilizers within its neighborhood. Similar to the first case, when we have an automorphism gadget preceding our logical measurement, we can simply track the support of our logical operator through the permutation (recall that the gadget only acts as a permutation on the left sector) and attach the ancillary system to its new support. Equivalently, we can perform the physical permutation on the left sector and ignore the $W$ on the right sector, which would perform the desired logical transformation (in the canonical basis) but change the underlying stabilizer code. The weight of the logical operator does not change through the permutation, and so the only adjustment to the ancillary system would come from the modified stabilizer structure. The only reason to implement the subsystem circuit $W$ would be if one desires to use the \emph{same} ancillary system to measure the logical operators both prior to and following the automorphism gadget.

For case 3, we envision that we have multiple blocks of the same code, and we would like to perform a transversal gate between the blocks. There are two situations to consider: when the other blocks are in known states versus unknown states. When the other blocks are in known states, as is the case with state injection, we can ``offload'' the automorphism gadget to the ancillary blocks as follows. Similar to the previous case, we can perform the relabeling part of the automorphism gadget on the left sector and ignore the circuit part on the right sector; however, when we do this relabeling, we will deform the underlying stabilizer code. In order for the transversal gate to work properly, the ancillary blocks will also need to be initialized (or distilled) according to the deformed code. Since the deformed code's logical operators have the same weight as the original code, we do not anticipate the distillation overhead on the deformed code to be substantially greater than that on the original code. For the second situation where the other blocks are in unknown states, as is the case when they are data/computational blocks themselves, the above trick no longer works. We will then generally need to perform the full automorphism gadget in order to properly realize the action of the transversal gate.

Now suppose that implementing the subsystem circuit $W$ becomes unavoidable for reasons such as those mentioned above. We briefly mention a few ways to mitigate potential error propagation when $W$ is not shallow. The proposed methods all come with their own space and time overheads, and any particular choice will boil down to the specific choice of code and automorphism gadget. The first is to introduce additional ``flag'' qubits that detect the spread of correlated errors and adjust our decoder accordingly, as has been done for fault-tolerant syndrome extraction of non-LDPC codes \cite{Chao_2018}. With this method, one will also need to carefully schedule the flag qubits so that their inclusion does not decrease the effective distance. The second method involves finding a low-depth collection of $\{w\} \subseteq \mathcal{A}^\perp$ which generates the entire group $\mathcal{A}^\perp$; recall that $\mathcal{A}^\perp \simeq \mathcal{A}$ when $d,d^\perp>2$ from Corollary \ref{cor:dual aut equal}. Any element of $\mathcal{A}^\perp$ can then be compiled into a sequence of low-depth components interleaved with intermediate error correction (see Section \ref{sec:Hamming codes} on Hamming codes for an explicit example). Since each component has low depth, the ``light cone'' of correlated errors is bounded, and so we can catch and correct them during the intermediate error correction cycles between the components. The cost of this second method is the additional time that comes with this compilation scheme. A third method also involves intermediate error correction, but we relax the requirement that the intermediate steps return back to the original codespace. Instead, we will keep track of how the stabilizer checks (and hence code) change through $W$ and perform error correction using the ``deformed'' codes, a procedure known as pieceable fault tolerance \cite{Hill_2013, Yoder_2016}; note that since $W$ is Clifford, the deformed codes are still CSS. We can then perform intermediate error correction with respect to the new stabilizer checks using correlated decoding \cite{Cain_2024}. The drawback with this approach is that we have no control over the weight of these new stabilizer checks, which could become rather large in the bulk of the $W$ circuit.

Lastly, we briefly mention a potential decoding strategy that can complement the methods mentioned above. The strategy involves measuring an enlarged set of checks, or stabilizer generators, that lets us distinguish between elementary errors on the left sector and correlated errors on the right sector. We will analyze the case for HGP codes, but we note that an analogous transformation can be made for the homological product codes. Given a parity-check matrix $H \in \mathbb{F}^{m \times n}_2$ and a matrix $B \in \mathbb{F}^{\ell\times m}_2$, we can construct a new parity-check matrix $H' \equiv BH \in \mathbb{F}^{\ell\times n}_2$ with $\rs H' \subseteq \rs H$. The new parity checks are related to the original checks through $B$. For the HGP $H_X$ matrix \eqref{eq:HGP H_X}, we will choose $B = \ident \otimes g_2$. Then we have
\begin{align}
    H'_X = BH^{}_X = \big(\, h_1 \otimes g_2 \,|\, \ident \otimes g^{}_2h^\transpose_2 \,\big) = \big(\, h_1 \otimes g_2 \,|\, \mathbf{0} \,\big) \, ,
\end{align}
where we have used the fact that $g^{}_2 h^\transpose_2 = 0$ by definition of $g_2$. A similar procedure can be used to construct new $Z$-checks using $B=g_1\otimes\ident$. Notice that these new checks are strictly supported in the left sector of the HGP code, and their extra syndrome information could potentially be used to distinguish between left-sector and right-sector errors in addition to the original checks in $H_X$. In fact, these new checks are precisely the stabilizer generators for the subsystem HGP \cite{bacon2006} and subsystem homological product codes \cite{Zeng_2020}, which can be viewed as a projection of the original HGP or homological product code onto the left sector of data qubits. However, the weight of these new checks are proportional to the classical code distances due to presence of $g_1$ ($g_2$) in $H'_X$ ($H'_Z$), and as such additional resources will be required to extract their syndromes in a fault-tolerant manner.


\section{Classical codes with automorphisms}\label{sec:classical families}

In the previous section, we saw that classical and quantum code automorphisms induce fault-tolerant logical operations in the corresponding homological product codes. Due to the flexibility of the homological product, \emph{any} choice of classical or quantum input codes can be tensored together, and so constructing a suitably symmetric homological product code effectively boils down to looking at symmetric input codes. In this section, we will explore various families of classical codes with rich automorphism structures that could potentially be used in our framework. Note that computing the automorphism group of a generic linear code is related to the \textsc{Permutation Code Equivalence} problem, which admits a polynomial-time reduction to \textsc{Graph Isomorphism} \cite{Petrank_1997}, for which the current fastest algorithm runs in quasipolynomial time \cite{Babai_2016, helfgott2017}. Nonetheless, when a code admits special structure, such as a Reed-Muller code, its automorphism group can be efficiently computed.

\subsection{Cycle codes of graphs}\label{sec:cycle codes}

Given a simple, connected graph $\mathsf{G} = (V,E)$, the cycle code $\mathcal{C}(\mathsf{G})$ is defined as the linear code whose parity-check matrix $H \in \mathbb{F}^{\abs{V}\times\abs{E}}_2$ is the edge-vertex incidence matrix of $\mathsf{G}$ \cite{Kasami_1961, Hakimi_1968}. In particular, we place bits on the edges of $\mathsf{G}$ and parity checks on the vertices. Using the fact that each edge must be attached to two vertices, the connectedness of $\mathsf{G}$ implies that the set of edges $\big\{e^{(i)}_j\big\}_j$ that lie in the neighborhood of any vertex $v_i \in V$ can be generated by ``adding (modulo 2)'' all the edges $\big\{e^{(k)}_\ell\big\}_\ell$ that lie in the neighborhood of all other vertices $v_k \in V \setminus \{v_i\}$. In other words, the cycle code $\mathcal{C}(\mathsf{G})$ has $\abs{V} - 1$ linearly independent checks since any check can be obtained from a linear combination of the other checks. 

Because a cycle code is defined according to a graph, many of its parameters can be efficiently computed from properties of this graph. If $\Delta$ is the maximum degree of $\mathsf{G}$, then $H$ is $(2,\Delta)$-LDPC, and we have $\abs{E} \leq \Delta\abs{V}/2$ with equality if $\mathsf{G}$ is regular. The transpose code with parity checks on the edges and bits on the vertices is simply the repetition code, which is another way to see the global linear dependency among the rows of $H$. It follows that the code dimension is $k = \abs{E} - \abs{V} + 1$, which is also the rank of the fundamental group $\pi_1(\mathsf{G})$ of the graph. Logical codewords have support on closed loops or cycles in $\mathsf{G}$. It follows that the code distance is given by the girth, or minimum cycle length of $\mathsf{G}$, which can be efficiently computed in $O\big(\abs{V}\abs{E}\big)$ operations via breadth-first search. Furthermore, since the Tanner graph of $\mathcal{C}$ is essentially given by $\mathsf{G}$, automorphisms of $\mathsf{G}$ are automatically Tanner graph automorphisms of $\mathcal{C}$; in other words we have $\mathrm{Aut}(\mathsf{G}) \simeq \mathcal{T}(\mathcal{C}) \subseteq \mathrm{Aut}(\mathcal{C})$. Unfortunately, finding the automorphism group of a generic graph can only be done in quasipolynomial time \cite{Babai_2016, helfgott2017} and so may become computationally intractable at large $n$. For some specific graphs, however, there can be simple arguments to deduce their automorphism groups.

\begin{figure}[t]
    \centering
    \includegraphics[width=0.85\textwidth]{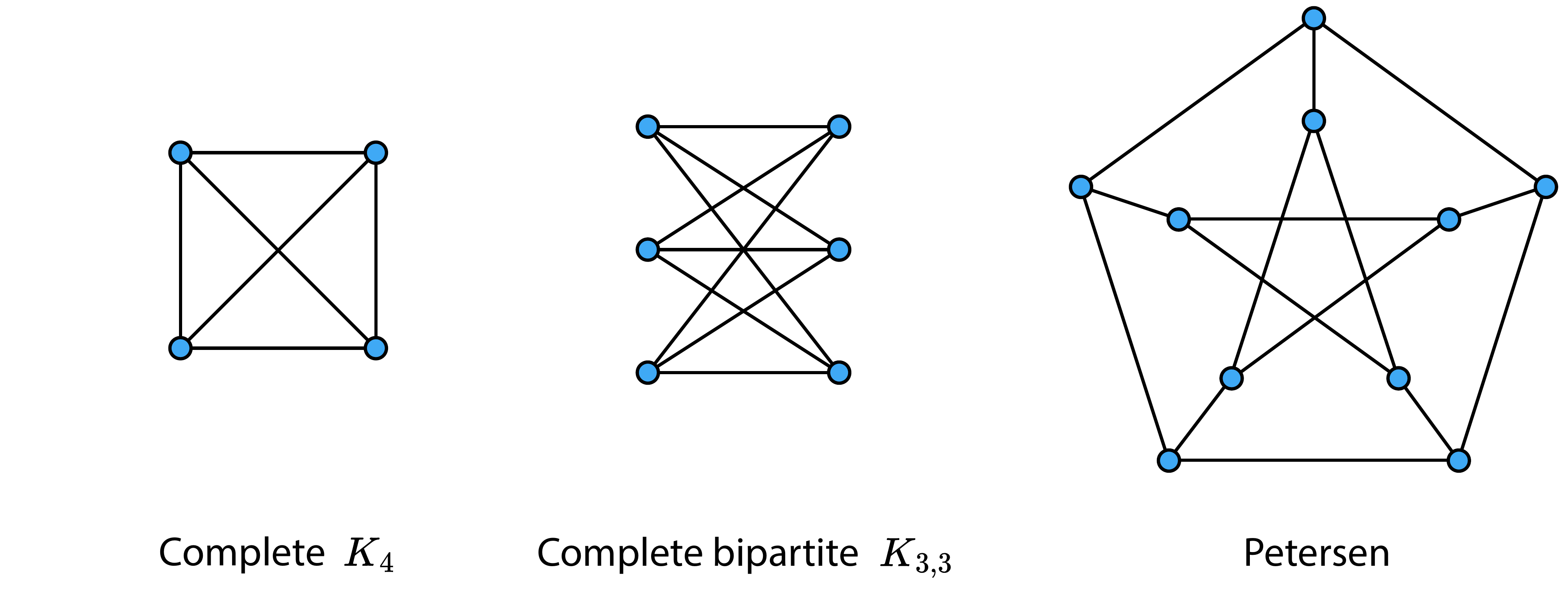}
    \caption{The graphs of the cycle codes in Table \ref{tab:cycle codes} are drawn.}
    \label{fig:example graphs}
\end{figure}

As an example, a parity-check matrix for the complete graph on 4 vertices $K_4$ (see the leftmost graph in Figure \ref{fig:example graphs} for an illustration) is given by
\begin{align}\label{eq:H_K4}
    H_{K_4} =
    \underbrace{
    \begin{pmatrix}
        1 & 1 & 0 & 0 & 1 & 0 \\
        0 & 1 & 1 & 0 & 0 & 1 \\
        0 & 0 & 1 & 1 & 1 & 0 \\
        1 & 0 & 0 & 1 & 0 & 1
    \end{pmatrix}
    }_\text{edges} \scalebox{1.5}{$\Bigg\rbrace$} \text{vertices}
    \, .
\end{align}
Note that each column contains two 1s since each edge must be attached to two vertices, which also means that the sum of all rows is 0 (mod 2), i.e. each row can be obtained by summing all other rows. Thus, there are three linearly independent checks.
Since $K_4$ is a complete graph, every vertex is connected to every other vertex, and so any permutation of the vertices leaves the graph unchanged. That is to say, the graph automorphism condition \eqref{eq:graph automorphism} is satisfied for any $4\times 4$ permutation matrix acting to the left of \eqref{eq:H_K4}.

Since we are interested in highly symmetric graphs with large automorphism groups, we will typically restrict to $\Delta$-regular graphs; note that two vertices need to have the same degree in order to map to each other through an automorphism. For a $\Delta$-regular graph with $n$ edges, there are $2n/\Delta$ vertices, and so the corresponding cycle code has rate $k/n = 1-2/\Delta + 1/n$, which is nonvanishing for $n\rightarrow\infty$ as long as $\Delta > 2$. Note that $\Delta=2$ corresponds to the ring graph which gives the $[n,1,n]$ repetition code. Unfortunately, as soon as $\Delta>2$, the Moore bound on girth tells us that
\begin{align}\label{eq:Moore bound}
    g = d \leq 2\log_{\Delta-1}n + O(1) \, .
\end{align}
As a consequence, if we try to construct a constant-rate LDPC cycle code, our code distance will be at most logarithmic in $n$. We know that \eqref{eq:Moore bound} is close to being tight since the family of Ramanujan spectral expanders \cite{Lubotzky_1988} obey $g \geq \frac{4}{3}\log_{\Delta-1}n + O(1)$. Although the distance of constant-rate cycle codes are not quite satisfactory from an asymptotic standpoint, they can be rather competitive at small code sizes. Table \ref{tab:cycle codes} lists three examples with $n\leq 15$, taken from the Forster census of symmetric graphs with uniform degree 3 (cubic) up to 512 vertices \cite{foster1988}. Each of the listed codes also has $d^\perp > 2$, and so Corollary \ref{cor:dual aut equal} tells us that each graph automorphism corresponds to a distinct logical operation.

\begin{table}[t]\renewcommand{\arraystretch}{1.4}
\centering
\begin{tabular}{c|c|c|c|c|c}
Cubic graphs              & $[n,k,d]$  & $d^\perp$ & $\mathrm{Aut}(\mathsf{G}) \simeq \mathcal{T}$                                      & $\abs{\mathrm{Aut}(\mathsf{G})}$ & $\mathrm{Aut}(\mathcal{C})$ \\ \hline
Complete $K_4$               & $[6,3,3]$  & 3         & $\mathrm{S}_4$                                         & 24 & $\mathrm{S}_4$                      \\
Complete bipartite $K_{3,3}$ & $[9,4,4]$  & 3         & $\mathrm{S}_3 \times \mathrm{S}_3 \rtimes \mathbb{Z}_2$ & 72  &  $\mathrm{S}_3 \times \mathrm{S}_3 \rtimes \mathbb{Z}_2$                    \\
Petersen                     & $[15,6,5]$ & 3         & $\mathrm{S}_5$                                         & 120   &  $\mathrm{S}_5$                
\end{tabular}
\caption{Cycle code parameters for a selection of small cubic graphs with the largest known automorphism groups (for their size) are displayed. For these examples, the graph automorphism groups are the full code automorphism groups, i.e. $\mathcal{T} \simeq \mathrm{Aut}(\mathcal{C}) \simeq \mathcal{A}$, verified using the \textsf{GUAVA} package in GAP \cite{GUAVA_GAP}.}
\label{tab:cycle codes}
\end{table}

For concreteness, we briefly discuss how we computed the graph automorphisms for the complete bipartite $K_{3,3}$ graph and the Petersen graph. For the complete bipartite $K_{3,3}$ graph, we can permute the three vertices on each side of the bipartition independently, giving us $\mathrm{S}_3 \times \mathrm{S}_3$.
In addition, we can reflect the graph along the bipartition which effectively swaps the vertices of the two halves.
Because reflections do not commute with the permutations that may take place on each side of the bipartition, we end up with a semi-direct product between $\mathrm{S}_3 \times \mathrm{S}_3$ and $\mathbb{Z}_2$.
In the case of the Petersen graph, consider a bijective mapping that maps each of the 10 vertices in the graph to some unique ordered pair of indices in $\{1, 2, 3, 4, 5\}$ such that a vertex only shares an edge with another vertex if and only if their pairs of indices do not share any index. For example, the vertex that is associated with $(2, 3)$ does not share an edge with the vertex that is mapped to $(3, 5)$ but shares an edge with the vertices that are associated with $(1, 4), (1, 5)$ and $(4, 5)$. It is not too hard to see that any permutation in $\mathrm{S}_5$ preserves this relationship since we utilize all possible pairs of indices in $\{1, 2, 3, 4, 5\}$ with our 10 vertices. With a bit more work, one can show that there are no additional automorphisms beyond $\mathrm{S}_5$ \cite{foster1988}.

\subsection{Group-algebra codes}\label{sec:group-algebra codes}

Perhaps the most straightforward way to construct classical LDPC codes with automorphism groups is with the help of a group algebra. Specifically, given a finite group $\mathcal{G}$ of order $\abs{\mathcal{G}}=n$, an element of the (binary) group algebra $a \in \mathbb{F}_2[\mathcal{G}]$ takes the form
\begin{align}\label{eq:group algebra element}
    a = \sum_{i=1}^{n} c_i g_i \;,\quad c_i \in \mathbb{F}_2 \;,\; g_i \in \mathcal{G} \, .
\end{align}
An intuitive way to think of the group algebra is that the coefficients $c_i$ interact through addition while the group elements $g_i$ control multiplication. To get a parity-check matrix out of \eqref{eq:group algebra element}, we simply replace each $g_i$ with a binary matrix representation $\mathbb{B}[g_i]$, which is typically taken to be the (left) regular representation; the classical parity-check matrix is then given by $H \equiv \mathbb{B}[a]$. In the regular representation, each $\mathbb{B}[g_i]$ is an $n \times n$ permutation matrix whose rows and columns index the group elements and whose entries tabulate the group's Cayley multiplication table. Because the row and column weight of a permutation matrix is always 1, the row and column weights of $H$ will be upper-bounded by the number of nonzero terms in the sum \eqref{eq:group algebra element}.

Suppose $\mathcal{G}$ is abelian. Then it is easy to see that for any classical group-algebra code $\mathcal{C}$ built from $\mathbb{F}_2[\mathcal{G}]$, we have $\mathcal{G} \subset \mathrm{Aut}(\mathcal{C})$ since for any $g_j \in \mathcal{G}$,
\begin{align}\label{eq:ag=ga}
    a g_j = \bigg( \sum_{i=1}^n c_i g_i \bigg) g_j = g_j \sum_{i=1}^n c_i g_i = g_j a
\end{align}
using the distributive property and the fact that $g_ig_j = g_jg_i$ for an abelian group. Plugging in the regular representation $H = \mathbb{B}[a]$ into \eqref{eq:ag=ga}, we immediately see that
\begin{align}\label{eq:HB=BH}
    H \mathbb{B}[g_j] = \mathbb{B}[g_j] H  \, ,
\end{align}
and upon comparing with the automorphism condition \eqref{eq:aut conditions}, we conclude that $\sigma = W = \mathbb{B}[g_j]$. Since $W=\sigma$ is also a permutation matrix, these transformations are graph automorphisms (Def. \ref{defn:graph automorphism}) and hence lift to exact automorphisms of the HGP code.

A commonly used family of abelian group-algebra codes are the cyclic codes, where we take $\mathcal{G} = \mathbb{Z}_n$. Since the elements of $\mathbb{Z}_n$ can all be generated by the unit translation $x$, it is customary to express \eqref{eq:group algebra element} in terms of a degree $\leq n$ polynomial in $x$. For example, a parity-check matrix for the repetition code is given by the polynomial $a = 1+x$. Taking $\sigma = W = \mathbb{B}[x]$, we see that the HGP operators \eqref{eq:HGP gadgets} act as unit translations of all qubits in the $\hat{x}$ and $\hat{y}$ directions (with periodic boundaries). This particular automorphism ``translation'' gadget has been recently used to lower the spacetime overhead for fault-tolerant computation using HGP codes \cite{xu2024fast}. In our framework, we see that this automorphism gadget arises from Tanner graph automorphisms associated with group-algebra codes over $\mathbb{F}_2[\mathbb{Z}_n]$.

\begin{table}[t]\renewcommand{\arraystretch}{1.4}
\centering
\begin{tabular}{c|c|c|c|c|c}
    Base group $\mathcal{G}$   &  $H$   & $[n=\abs{\mathcal{G}},k,d]$  & $d^\perp$   & $\mathrm{Aut}(\mathcal{C})$ & $\abs{\mathrm{Aut}(\mathcal{C})}$  \\ \hline
    $\mathbb{Z}_7$ & $1+x+x^3$ & $[7,3,4]$ & 3 & $\mathrm{GL}_3(\mathbb{F}_2)$ & 168  \\
    $\mathrm{D}_6 \simeq \mathbb{Z}_6 \rtimes \mathbb{Z}_2$ & $1+r+sr^{-1}$ & $[12,4,6]$ & 3 & $(\mathrm{A}^2_4 \rtimes \mathbb{Z}_2) \rtimes \mathbb{Z}_2$ & 576  \\
    $\mathrm{D}_8 \simeq \mathbb{Z}_8 \rtimes \mathbb{Z}_2$ & $1+r^2+r^3+sr^{-1}$ & $[16,6,6]$ & 4 & $\mathbb{Z}^4_2 \rtimes \mathrm{S}_6$ & 11520
\end{tabular}
\caption{Parameters for a few small abelian and non-abelian group-algebra codes are displayed. The first code is the dual Hamming or simplex code of length 7. For the dihedral group $\mathrm{D}_\ell$ of order $2\ell$, we use the left-regular representation according to the presentation $\mathrm{D}_\ell = \langle r,s \;|\; r^\ell=s^2=(rs)^2=1 \rangle$. Automorphism groups are computed using the \textsf{GUAVA} package in GAP \cite{GUAVA_GAP}.}
\label{tab:group-algebra codes}
\end{table}

For non-abelian $\mathcal{G}$, it is not immediately clear how we can get an analogue of \eqref{eq:HB=BH} since we will not be able to commute $\mathbb{B}[g_j]$ through $H$ in general. However, we can employ the following trick which has been previously used to construct both small-scale \cite{2BGA} as well as large-scale \cite{Panteleev_2022_good, Dinur_2022_LTC} LDPC codes. Since $\mathcal{G}$ is non-abelian, we need to distinguish between the left-regular and right-regular representations, corresponding to left and right group-multiplication respectively. Denote $L[\cdot]$ and $R[\cdot]$ the left-regular and right-regular representations accordingly. Then we can still use \eqref{eq:group algebra element}, but we define $h = L[a]$ to be the left-regular representation of $a$. Then $R[g]$ for any $g \in \mathcal{G}$ commutes with $h$ due to the associativity of group multiplication, i.e. for $g_1g_2g_3$ the order of multiplying $g_1$ ($g_3$) on the left (right) does not matter.
Table \ref{tab:group-algebra codes} lists group-algebra codes arising from the abelian group $\mathbb{Z}_7$ as well as the non-abelian groups $\mathrm{D}_6$ and $\mathrm{D}_8$.

Although their construction is relatively straightforward, group-algebra codes have their number of physical bits equal to the size of the group of interest, $n = \abs{\mathcal{G}}$, as a consequence of the regular representation. For example, if one wishes to use $\mathcal{G} = \mathrm{S}_\ell$, then the number of physical bits is $n=\ell!$. As we will see in the later examples, there exist other classical code constructions where the number of physical bits can be much less than the size of the group.

\subsection{Lifted group-algebra codes}

We can generalize group-algebra codes by using a matrix of polynomial instead of a single (scalar) polynomial, yielding lifted group-algebra codes.
Lifting is a common technique used in the literature on classical coding theory to generate large LDPC codes from smaller known LDPC codes. Given an arbitrary Tanner graph that corresponds to some classical LDPC code $\mathcal{C}$, we can lift it to an LDPC code with block length that is $\ell$ times larger.
The resulting Tanner graph is called an $\ell$-\emph{lift} or $\ell$-fold \emph{cover graph} for the base Tanner graph. An easy way to interpret the result of the lift is to consider replacing each vertex $v$ of the base graph with $\ell$ copies $v_1, \ldots, v_\ell$ of the same vertex and each edge $e$ that connects $v$ and $v'$ with $\ell$ copies of the same edge $e_1, \ldots, e_\ell$. In the $\ell$-\emph{lift}, the edge $e_i$ connects the vertices $v_i$ and $v'_{\pi(i)}$ for some permutation $\pi \in S_\ell$. We emphasize that the permutation $\pi$ does not need to be the same for the different edges in the base graph. 

Consider \emph{shift $\ell$-lifts}, which are lifts based on a cyclic subgroup $\Gamma_\ell$ generated by the permutation $(1, 2, \ldots, \ell) \in S_\ell$.
Before we elaborate on the algebraic structure of the parity check matrices that are obtained as a result of a shift $\ell$-lift, we first define what an $\ell\times\ell$ circulant matrix is. We reuse the notation provided in Ref.~\cite{panteleev2021quantum}.

\begin{defn}[$\ell\times\ell$ Circulant Matrix]\label{defn:circulant_matrix}
    An $\ell\times \ell$ circulant matrix $M$ over the field $\mathbb{F}_q$ is the following:
    \[M = \begin{pmatrix}
        a_0 & a_{\ell - 1} & \ldots & a_1 \\
        a_1 & a_0 & \ldots & a_2 \\
        \vdots & \vdots & \ddots & \vdots \\
        a_{\ell - 1} & a_{\ell - 2} & \ldots & a_0
    \end{pmatrix},\]
    where $a_i \in \mathbb{F}_q$ for all $i \in [\ell - 1] \cup \{0\}$.
    The matrix $M$ can also be decomposed into the following form:
    \[M = a_0\ident_{\ell} + a_1 P + \ldots + a_{\ell - 1} P^{\ell - 1},\]
    where $\ident_\ell$ is the $\ell\times\ell$ identity matrix and $P$ is the permutation matrix representing the right cyclic shift $(1, 2, \ldots, \ell) \in S_\ell$ i.e., 
    \[P = \begin{pmatrix}
        0 & 0 & \ldots & 1 \\
        1 & 0 & \ldots & 0 \\
        0 & 1 & \ldots & 0 \\
        \vdots & \vdots & \ddots & \vdots \\
        0 & 0 & \ldots & 0
    \end{pmatrix}.\]
\end{defn}

It is not too hard to see that the ring of $\ell \times \ell$ circulant matrices over $\mathbb{F}_q$ is isomorphic to the ring of polynomials over $\mathbb{F}_q$ modulo the polynomial $x^{\ell} - 1$ i.e., $\mathbb{F}_q[x] / (x^\ell - 1)$, given $P^\ell = \ident_\ell$.
Thus, we can succinctly express the matrix $M$ in the form of an $(\ell-1)$-degree $\mathbb{F}_q$-polynomial: 
\[m = a_0 + a_1 x + \ldots + a_{\ell - q}x^{\ell - 1}.\]
We denote $M \coloneqq \mathbb{B}(m)$ where $\mathbb{B}$ simply denotes the block matrix form of the polynomial. We also abuse notation by allowing $\mathbb{B}$ to act on matrices such that it changes the polynomial in each matrix entry into its block matrix form.

Now, suppose we are given a parity check matrix $H_0 \in \mathbb{F}_2^{c \times n}$.
We can obtain a new parity check matrix $H \in \mathbb{F}_2^{c\ell \times n\ell}$ that corresponds to the shift $\ell$-lift of the Tanner graph for $H_0$. The matrix $H$ is a quasi-cyclic matrix i.e., a block matrix where each block is an $\ell \times \ell$ circulant matrix as defined in Def.~\ref{defn:circulant_matrix}. It is thus convenient to express the block matrices of $H$ as the following:
To construct $H$, we can first construct the matrix $H_{poly}$ over the ring of polynomials $\mathbb{F}_2[x] / (x^\ell - 1)$ from $H_0$.
To do that, we simply multiply each entry $\left(H_0\right)_{ij}$ of $H_0$ with some polynomial $m_{ij} \in \mathbb{F}_2[x] / (x^\ell - 1)$ to obtain $H_{poly} \in \left(\mathbb{F}_2[x] /(x^{\ell} - 1) \right)^{c \times n}$.
We then use the isomorphism between $\mathbb{F}_2[x] / (x^{\ell} - 1)$ and the $\ell \times \ell$ circulant matrices described in Def.~\ref{defn:circulant_matrix} to construct the individual block matrices $\mathbb{B}\left(\left(H_0\right)_{ij} m_{ij}\right)$ that replaces each polynomial $\left(H_0\right)_{ij} m_{ij}$.
This allows us to obtain $\mathbb{B}\left(H_{poly}\right) = H$.

We include an example below where we set $\ell = 4$:
\begin{align}
    H_0 &= \begin{pmatrix}
        1 & 1 & 0 \\
        1 & 0 & 1
    \end{pmatrix}, \\
    m &= \begin{pmatrix}
        x & 1 + x^2 & x^3 \\
        0 & x & x^2
    \end{pmatrix}, \\
    H_{\poly} &= \begin{pmatrix}
        x & 1 + x^2 & 0 \\
        0 & 0 & x^2
    \end{pmatrix}, \\
    H &= \mathbb{B}\left(H_{\poly}\right) = \left(\begin{array}{cccc|cccc|cccc}
        0 & 0 & 0 & 1 & 1 & 0 & 1 & 0 & 0 & 0 & 0 & 0 \\
        1 & 0 & 0 & 0 & 0 & 1 & 0 & 1 & 0 & 0 & 0 & 0\\
        0 & 1 & 0 & 0 & 1 & 0 & 1 & 0 & 0 & 0 & 0 & 0\\
        0 & 0 & 1 & 0 & 0 & 1 & 0 & 1 & 0 & 0 & 0 & 0\\
        \hline
        0 & 0 & 0 & 0 & 0 & 0 & 0 & 0 & 0 & 0 & 1 & 0\\
        0 & 0 & 0 & 0 & 0 & 0 & 0 & 0 & 0 & 0 & 0 & 1\\
        0 & 0 & 0 & 0 & 0 & 0 & 0 & 0 & 1 & 0 & 0 & 0\\
        0 & 0 & 0 & 0 & 0 & 0 & 0 & 0 & 0 & 1 & 0 & 0 
    \end{array}\right). 
\end{align}

Without assuming any underlying automorphisms for $H_0$, the shift $\ell$-lift for $H_0$, that is, $H$, is equipped with a set of automorphisms that arise from the lift. 
Consider the permutation of the following form:
\[\pi = \bigoplus_{i = 1}^{n} \sigma,\; \sigma \in \Gamma_\ell.\]
In other words, the permutation $\pi$ is a direct sum of permutations of the subgroup $\Gamma_\ell \subseteq S_\ell$ generated from $(1, 2, 3, \ldots, \ell)$.
The matrix representation for $\pi$ is the following:
\[\pi = \begin{pmatrix} \sigma & 0 & \ldots & 0 \\ 0 & \sigma & \ldots & 0 \\ \vdots & \vdots & \ddots & \vdots \\ 0 & 0 & \ldots & \sigma\end{pmatrix} = \begin{pmatrix} x^{s} & 0 & \ldots & 0 \\ 0 & x^{s} & \ldots & 0 \\ \vdots & \vdots & \ddots & \vdots \\ 0 & 0 & \ldots & x^{s} \end{pmatrix},\]
for some $s \in [\ell - 1] \cup \{0\}$.
It is not difficult to see that
\[H\pi = \left(\bigoplus_{i = 1}^{c} \sigma\right)H \eqqcolon WH\] is
due to the fact that $\Gamma_\ell$ is an abelian group.
And similar to the case for the group-algebra codes, because $W$ is also a permutation matrix, these transformations are graph automorphisms (Def. \ref{defn:graph automorphism}) and hence lift to exact automorphisms of the HGP code.

It is useful to consider the case where the group $\mathcal{G}$ is non-abelian.
In fact, non-abelian lifts were essential for the construction of asymptotically good quantum LDPC codes~\cite{Panteleev_2022_good}.
Similar to the case for group-algebra codes, we face the same obstacle that prevents us from commuting $\pi$ through $H$ when $\mathcal{G}$ is non-abelian.
However, it is easy to see that the trick used to circumvent the obstacle for the group-algebra codes can also be used in the case for the lifted codes.
We simply utilize both the left and right actions of $\mathcal{G}$ which trivially commute by the associativity of group multiplication.

\subsection{$[2^r-1, 2^r-r-1, 3]$ Hamming codes}
\label{sec:Hamming codes}

Hamming codes are one of the first linear codes studied after the simple repetition code. They have minimum distance $d=3$ and dimension $k = n - \log_2(n+1)$. They are known as perfect codes because they saturate the Hamming (sphere-packing) bound and thus partition the physical space $\mathbb{F}^n_2$ into sets (balls) of correctable errors. The parity-check matrix of a Hamming code of level $r$ consists of all $2^r-1$ nonzero bitstrings of length $r$ in $\mathbb{F}^r_2$. As a result, each row has weight $(n+1)/2$ and the maximum column weight is $r$. As an example, consider the $[7,4,3]$ $r=3$ Hamming code with parity-check matrix
\begin{align}
    H_{[7,4,3]} = \begin{pmatrix}
        1 & 0 & 0 & 1 & 1 & 0 & 1 \\
        0 & 1 & 0 & 1 & 0 & 1 & 1 \\
        0 & 0 & 1 & 0 & 1 & 1 & 1
    \end{pmatrix} \, .
\end{align}
One can see in the above example that the 7 columns of $H$ comprise all 7 nonzero bitstrings on 3 bits. Using the structure of $H$ for the Hamming codes, one can easily show that $\mathcal{A}^\perp = \mathrm{GL}_r(\mathbb{F}_2)$ as follows. For any $W \in \mathcal{A}^\perp$, observe that the columns of $WH$ transform according to $W$. Since $W$ is invertible and $H$ contains all nonzero bitstrings, the action of $W$ simply permutes the columns of $H$, which implicitly defines a corresponding permutation matrix $\sigma \in \mathrm{S}_n$ such that we obtain the automorphism condition $WH=H\sigma$. So we have $\mathrm{GL}_r(\mathbb{F}_2) \subset \mathcal{A}^\perp$; at the same time, since $H$ has rank $r$, we must also have $\mathcal{A}^\perp \subset \mathrm{GL}_r(\mathbb{F}_2)$. Thus, $\mathcal{A}^\perp = \mathrm{GL}_r(\mathbb{F}_2)$. In addition, since both the distance and dual distance are greater than 2, Corollary \ref{cor:dual aut equal} tells us that the corresponding logical gate group $\mathcal{A} \simeq \mathcal{A}^\perp = \mathrm{GL}_r(\mathbb{F}_2)$. Furthermore, since $\mathrm{S}_r$ is a subgroup of $\mathrm{GL}_r(\mathbb{F}_2)$, we also have $\mathcal{T}(H) = \mathrm{S}_r$, which is maximal.

Note that the physical implementation of a generic matrix $W \in \mathcal{A}^\perp = \mathrm{GL}_r(\mathbb{F}_2)$ may involve multiqubit gates and a superconstant circuit depth. However, we can always decompose any $W$ into elementary column operations that physically involve simple SWAPs and CNOTs, the number of which is at most $r(r-1)$. Furthermore, since each elementary operation also belongs to $\mathcal{A} \simeq \mathrm{GL}_r(\mathbb{F}_2)$, each step in the decomposition is a code automorphism, and so we can interleave error correction along the way. To summarize, for the Hamming code family with the standard form of the parity-check matrix, we have $\mathcal{A} \simeq \mathrm{GL}_r(\mathbb{F}_2)$ with Tanner graph automorphism subgroup $\mathcal{T} \simeq \mathrm{S}_r$. An arbitrary $W \in \mathcal{A}^\perp$ can be decomposed into at most $r(r-1) = O(\log^2 n)$ steps $W = W_1W_2,W_3,\dots$ comprised of permutations $\mathcal{T}$ and constant-depth circuits.


\subsection{$[2^r-1,r,2^{r-1}]$ simplex codes}

The dual of the $[2^r-1, 2^r-r-1, 3]$ Hamming codes are known as the $[2^r-1,r,2^{r-1}]$ simplex codes. They have minimum distance $d = (n+1)/2$ and dimension $k = \log_2(n+1)$, which saturates the Plotkin bound. In determining the automorphism group of Hamming codes, we first calculated $\mathcal{A}^\perp$ and then applied Corollary~\ref{cor:dual aut equal}. This $\mathcal{A}^\perp$ is $\mathcal{A}$ for the corresponding simplex code. Hence we again have that the corresponding logical gate group $\mathcal{A} \simeq \mathcal{A}^\perp \simeq \mathrm{GL}_r(\mathbb{F}_2)$. Since we have that $r = k$, we see that simplex codes saturate the bound of Corollary~\ref{cor:upper bound gates}; they have the largest possible automorphism group for a given number of logical bits.

For the Hamming codes, we had $W \in \mathrm{GL}_r(\mathbb{F}_2)$, and so we were able to easily characterize the graph automorphism subgroup as well as decompose any element into elementary steps that could be implemented in constant depth interleaved with error correction. For the simplex codes, the $W$s generate the same group but embedded in a larger matrix group $\mathrm{GL}_m(\mathbb{F}_2)$ where $m \geq 2^r-1-r > r$ for $r\geq 3$. For practical considerations, we would like to choose a parity-check matrix $H$ which minimizes the (gate) overhead for implementing each $W$ on the physical level. One choice of $H$ is to leverage the fact that all binary Hamming codes can be written as cyclic codes. For a length $n=2^r-1$ Hamming code, its codewords can be generated by any primitive polynomial of degree $r$ in $\mathbb{F}_2[x]/(x^n-1)$ \cite{MacWilliamsSloane}. For the simplex codes, this means we can build $H$ using the group-algebra construction with $\mathcal{G} = \mathbb{Z}_n$ and a primitive polynomial which is typically a trinomial\footnote{The constraint on the polynomial can be relaxed to only having a primitive greatest common divisor with $x^n+1$ \cite{malcolm2025}.}. As examples, the $[7,3,4]$ simplex code has parity-check matrix $H = \mathbb{B}[1+x+x^3]$, and the $[15,4,8]$ simplex code has parity-check matrix $H = \mathbb{B}[1+x+x^4]$. By writing the parity-check matrices of the simplex codes in cyclic form, we can ensure that we have $\mathbb{Z}_n$ as a subgroup of their Tanner graph automorphism groups.


\subsection{Other Reed-Muller codes}

Some of the previously mentioned codes can be considered instantiations of Reed Muller codes up to puncturing. In general, Reed Muller $\mathrm{RM}(r,m)$ codes~\cite{muller1954application, abbe2020} are defined by evaluating multivariate polynomials over $\mathbb{F}_2$. Consider the polynomial ring with $m$ variables $\mathbb{F}_2[x_1, x_2, ..., x_m]$. For a polynomial $f \in \mathbb{F}_2[x_1, x_2, ..., x_m]$ and a vector $y = (y_1, y_2, ...,y_m) \in \mathbb{F}_2^m$, we can compute the evaluation of $y$, $f(y_1, y_2, ..., y_m)$. Now define $\mathrm{Eval}(f)$ to be the bitstring obtained from evaluating $f(y)$ each of the $2^m$ vectors in $\mathbb{F}_2^m$. The code $\mathrm{RM}(m,r)$ is then defined as:
\begin{align}
    \label{eq:rm}
    \mathrm{RM}(r,m) = \big\{ \mathrm{Eval}(f) ~|~ f \in \mathbb{F}_2[x_1, ..., x_m], ~\deg(f) \le r \big\}.    
\end{align}
Note that over $\mathbb{F}_2$, $x^n = x$, and so a degree $r$ polynomial in $\mathbb{F}_2[x_1,...,x_m]$ has monomials that look like $x_1x_2...x_r$. From \eqref{eq:rm} it can be seen that $n = 2^m$. The number of logical bits $k$ is simply the number of polynomials with degree less than or equal to $r$:
\begin{align}
k = \sum_{i=0}^r {m \choose i},
\end{align}
and the distance is $d = 2^{m-r}$.
When $r=0$, we obtain the $[2^m, 1, 2^m]$ repetition code, and when $r = m$, we obtain the $[2^m, 2^m, 1]$ trivial code. Additionally, other codes can be derived by interpolating between these values. A small example is the $[8,7,2]$ code $\mathrm{RM}(2,3)$, which has the following generator matrix:

\begin{align}
G =
\begin{pmatrix}
\mathrm{Eval}(1) \\
\mathrm{Eval}(x_1) \\
\mathrm{Eval}(x_2) \\
\mathrm{Eval}(x_3) \\
\mathrm{Eval}(x_1x_2) \\
\mathrm{Eval}(x_1x_3) \\
\mathrm{Eval}(x_2x_3)
\end{pmatrix}
=
\begin{blockarray}{cccccccc}
 \rotatebox{60}{\tiny 000} & \rotatebox{60}{\tiny 001} & \rotatebox{60}{\tiny 010} & \rotatebox{60}{\tiny 011} & \rotatebox{60}{\tiny 100} & \rotatebox{60}{\tiny 101} & \rotatebox{60}{\tiny 110} & \rotatebox{60}{\tiny 111} & \\
\begin{block}{(cccccccc)}
1 & 1 & 1 & 1 & 1 & 1 & 1 & 1 \bigstrut[t]\\
0 & 0 & 0 & 0 & 1 & 1 & 1 & 1 \\
0 & 0 & 1 & 1 & 0 & 0 & 1 & 1 \\
0 & 1 & 0 & 1 & 0 & 1 & 0 & 1 \\
0 & 0 & 0 & 0 & 0 & 0 & 1 & 1 \\
0 & 0 & 0 & 0 & 0 & 1 & 0 & 1  \\
0 & 0 & 0 & 1 & 0 & 0 & 0 & 1 \bigstrut[b] \\
\end{block}
\end{blockarray}
\end{align}
Here, rows correspond to the input basis vectors, while columns correspond to the evaluations of the variables $Z_1,Z_2\dots Z_m$.   

In general, the size of the automorphism group of $\mathrm{RM}(r,m)$ is the general affine group $\mathrm{GA}_m(\mathbb{F}_2) \simeq \mathbb{F}_2^m \rtimes \mathrm{GL}_m(\mathbb{F}_2)$. This can be seen easily from the code construction: each evaluation point of the variables $x_i$ corresponds to a bit position. In total there are $2^m$ bit positions. An affine transformation on the variables corresponds to a permutation of the bit positions, thus any affine transformation is a valid automorphism.


Note that there are three groups to compare: $\mathrm{GL}_k(\mathbb{F}_2)$, $\mathrm{GL}_m(\mathbb{F}_2)$ and $\mathrm{S}_n$:
\begin{itemize}
    \item When $r=1$, then $k=m+1$. Here $\mathrm{GL}_m(\mathbb{F}_2)$ is roughly the same as $\mathrm{GL}_k(\mathbb{F}_2)$, while the set of permutations $\mathrm{S}_n$ is much bigger.
    \item When $r=m-1$, then $k=2^m-1$. For this case, $\mathrm{GL}_k(\mathbb{F}_2) = \mathrm{GL}_{n-1}(\mathbb{F}_2)$, and while the automorphism group is $\mathrm{S}_n$, this is still much smaller than $\mathrm{GL}_k(\mathbb{F}_2)$. This is the repetition code like case. 
    \item For $r=m/2$, then $k=n/2$. For this case, the size of the automorphism group is $\abs{\mathrm{GL}_m(\mathbb{F}_2)} \approx (0.29)2^{m^2} = \mathrm{\Theta}\big( n^{\log n} \big)$.
\end{itemize}
Alternatively, we can consider punctured RM codes:
\subsubsection{Punctured Reed Muller codes as cyclic codes}

Reed Muller codes are not cyclic in general, however, they can be punctured to obtain cyclic codes. 
\begin{defn}
    A \textbf{punctured RM} code is denoted $\mathrm{RM}^{*}(r,m)$ and is defined by taking an $\mathrm{RM}(r,m)$ code and deleting the coordinate corresponding to $Z_1,\ldots Z_m=0$ from all codewords.
\end{defn}
As given in \cite{macwilliams_1964}, we have the following theorem:
\begin{thm}
       Punctured RM code $RM(r,m)^{*}$ corresponds to a cyclic code, which has zeros $\alpha^s$ that satisfy:
       \begin{align}
           1 \leq w_2(s)\leq m-r-1 \qquad 1 \leq s\leq 2^m-2
       \end{align}
   where $s$ labels the cyclotomic cosets of the field $GF(2^m)$ that partition the integers as:
   \begin{align}
       \{0,1,\ldots 2^m\} = \bigcup_{s}C_s
   \end{align}
   where $s$ runs through the coset representatives$\mod 2^m$.
   
   The minimal polynomial $M_s(x)$ of $\alpha^s$ is:
   \begin{align}
       M^{(s)}(x) = \prod_{i\in C_s} (x-\alpha^i)
   \end{align}
   where $\alpha$ denotes the primitive element of the field.
   
   The generator and check polynomials can be defined in terms of the minimal polynomials $M^{(s)}(x)$ of $\alpha^s$.
   \begin{align}
       g(x) = \prod_{1\leq w_2(s)\leq m-r-1,1\leq s \leq 2^m-2} M^{s}(x)
   \end{align}
   \begin{align}
       h(x) = (x+1)\prod_{m-r\leq w_2(s)\leq m-1, 1 \leq s \leq 2^m-2 } M^{s}(x)
   \end{align}
\end{thm}
It is worth noting that the automorphism group for punctured RM code $\mathrm{RM}^{*}(r,m)$ is $\mathrm{GL}_m(\mathbb{F}_2)$. We highlight two specific examples which have interesting properties. The first example contains the punctured $\mathrm{RM}^*(1,m)$. These codes correspond to dual Hamming or simplex codes. In particular, take $\mathrm{RM}(1,m)$ and drop the first bit and the first row corresponding to all ones. The resulting $[2^m-1, m, 2^{m-1}]$ is the dual of the Hamming code with parameters [$2^{m}-1$,$2^{m}-m-1,3$]. 
The generator polynomial for $RM^{*}(1,3)$ can be found to be:
\begin{align}
    g(x)= 1 + x^2 + x^3 +x^4,
\end{align}
which corresponds to the check polynomial $h(x)$ of the $[7,4,3]$ Hamming code.

The second example contains some punctured families of RM codes that can also be thought of as Bose–Chaudhuri–Hocquenghem (BCH) codes \cite{Bose_1960_2, Bose_1960_1, Hocquenghem_1959}, a special subclass of cyclic codes. As an example, let us look at $\mathrm{RM}(2,5)$. By puncturing the code (say at the zero position) one obtains the cyclic BCH code $[31,16,\geq 7]$ code. The generator polynomial for the above BCH code is given as:
\begin{align}
    g(x) = x^{15} + x^{14} + x^{13} + x^{12} + x^{10} + x^9 + x^8 + x^6 + x^5 + x^3 + 1.
\end{align}


\section{Examples of quantum codes with automorphism gadgets}
\label{sec:examples}

In this section, we explicitly construct several hypergraph product codes, analyze their parameters, and provide some brief applications of their automorphism gadgets using the formalism from Section \ref{sec:aut gadgets}.

\subsection{$\llbracket 52,10,3 \rrbracket$ HGP code with $\mathcal{A} \simeq \mathrm{S}_4 \times \mathrm{S}_4$}

\begin{figure}[t]
    \centering
    \includegraphics[width=0.7\textwidth]{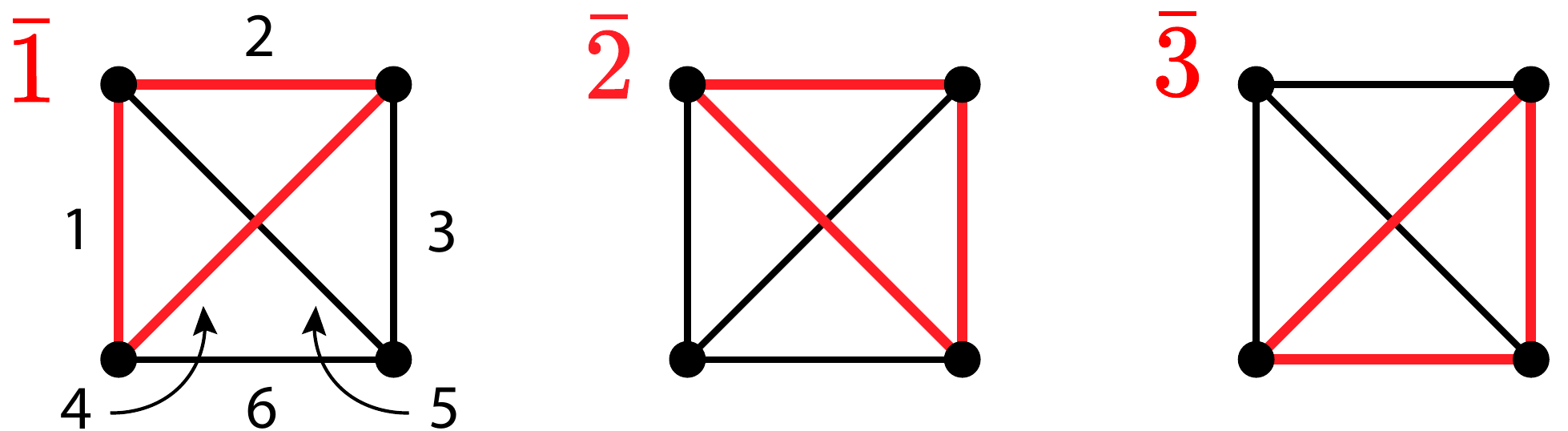}
    \caption{Codeword generators of the $K_4$ cycle code corresponding to the three rows of \eqref{eq:K_4 generator matrix}. Black numbers label the physical bits on edges, and red numbers label the logical bits.}
    \label{fig:K4 codewords}
\end{figure}

For the first example, we will take the $[6,3,3]$ cycle code on the complete graph on 4 vertices ($K_4$). The parity-check matrix was defined in \eqref{eq:H_K4}. Consider the following generator matrix:
\begin{align}\label{eq:K_4 generator matrix}
    g = \begin{pmatrix}
        1 & 1 & 0 & 1 & 0 & 0 \\
        0 & 1 & 1 & 0 & 1 & 0 \\
        0 & 0 & 1 & 1 & 0 & 1
    \end{pmatrix} \, ,
\end{align}
which corresponds to three out of the four ``triangular'' cycles in $K_4$; see Figure \ref{fig:K4 codewords} for an illustration. Using cycle notation for permutations, it is easy to verify that logical transposition $(\bar{1}\bar{2})$ is achieved via the physical permutation $(15)(34)$, and the logical transposition $(\bar{2}\bar{3})$ is achieved via $(24)(56)$. These two logical transpositions generate the entire logical permutation group $\mathrm{S}_3$ on all three logical qubits. In addition to all logical permutations, we also have the logical transformation
\begin{align}\label{eq:V_K4 CNOT}
    v = \begin{pmatrix}
        1 & 1 & 1 \\
        0 & 1 & 0 \\
        0 & 0 & 1
    \end{pmatrix}
\end{align}
corresponding to the physical permutation $(25)(46)$. The rows of \eqref{eq:V_K4 CNOT} describe the multiqubit logical gate $\mathrm{C}_{\bar{1}}\mathrm{NOT}_{\bar{2}\bar{3}} = \mathrm{C}_{\bar{1}}\mathrm{NOT}_{\bar{2}} \cdot \mathrm{C}_{\bar{1}}\mathrm{NOT}_{\bar{3}}$ where logical qubit $\bar{1}$ is controlled and logical qubits $\bar{2}$ and $\bar{3}$ are targeted. The logical permutations in combination with \eqref{eq:V_K4 CNOT} generate the (Tanner graph) automorphism group $\mathrm{S}_4$ for the $K_4$ cycle code. For each physical permutation $\sigma$ described above, we also have a corresponding row-operation $w$ given by $wh=h\sigma$. Since we are using a cycle code, every $w$ will also be a permutation; for example, the $w$ corresponding to \eqref{eq:V_K4 CNOT} is the permutation $(23)$. We take the HGP \eqref{eq:HGP H_X, H_Z} of \eqref{eq:H_K4} with itself to obtain a $\llbracket 52,10,3 \rrbracket$ HGP code, with 36 data qubits and 9 logical qubits in the left sector, and 16 data qubits and 1 logical qubit in the right sector. The CSS parity checks have weight 5, and all qubits participate in either 2 or 3 $X$-type checks or $Z$-type checks. It will be convenient to represent the logical state of the 9 left logical qubits as a $3\times 3$ block
\begin{align}\label{eq:K4 HGP 3x3 logicals}
    \ket{\overline{\Psi}_{\rm L}} = \left| \begin{array}{ccc}
        \overline{\psi}_1 & \overline{\psi}_2 & \overline{\psi}_3 \\
        \overline{\psi}_4 & \overline{\psi}_5 & \overline{\psi}_6 \\
        \overline{\psi}_7 & \overline{\psi}_8 & \overline{\psi}_9
    \end{array} \right\rangle
\end{align}
akin to the geometric picture of Figure \ref{fig:HGP layout}. The logical permutations of the classical code now lift to permutations of the ``rows'' and ``columns'' in \eqref{eq:K4 HGP 3x3 logicals} which corresponds to the group $\mathrm{S}_3 \times \mathrm{S}_3$. On the physical level, we will also be permuting rows and columns of both the left sector and right sector of data qubits, using $\sigma$ for the left sector and $w$ for the right sector. Note that the right logical qubit is invariant to all these automorphisms since its logical operators are repetition codewords (from the transpose of a cycle code). Adding in any ``diagonal'' logical SWAP such as $(\bar{1}\bar{5})$ from other means\footnote{e.g. an extractor \cite{he2025extr}} enlarges this permutation subgroup ($\mathrm{S}_3 \times \mathrm{S}_3$) to the full permutation group on all 9 left logical qubits. The inclusion of any two-qubit logical CNOT gate then enables us to compile the entire class of affine gates\footnote{Recall that the affine class of gates on $k$ (qu)bits is generated by all CNOTs and forms the group $\mathrm{GL}_k(\mathbb{F}_2)$.} on these 9 logical qubits.

\subsection{$\llbracket 48,6,3 \rrbracket$ HGP code with $\mathcal{A} \simeq \mathrm{S}_4 \times \mathrm{S}_4$ and transversal CZ gates}
\label{sec:K4 HGP copy-cup}

For the second example, we will also take the $K_4$ cycle code with parity-check matrix \eqref{eq:H_K4} and generator matrix \eqref{eq:K_4 generator matrix}. But this time, we take the HGP of the $K_4$ cycle code with its transpose, the $[4,1,4]$ repetition code with bits on the 4 vertices of $K_4$ and parity checks on the 6 edges. The resulting HGP code has parameters $\llbracket 48,6,3 \rrbracket$ with 3 left logical qubits and 3 right logical qubits; note that there is a symmetry between the left and right sectors now. The $X$-type and $Z$-type checks have weights 6 and 4 respectively, and each qubit participates in 2 $X$-type and 3 $Z$-type checks. For this example, we will care about the right logical qubits, and so for clarity we write down the full canonical logical bases for both left and right sectors:
\begin{subequations}
\begin{align}
    G_{Z,{\rm L}} &= (\, g \otimes \mathbf{e}_1 \;|\; \mathbf{0} \,)  \\
    G_{X,{\rm L}} &= (\, E_1 \otimes \mathbf{1} \;|\; \mathbf{0} \,)  \\
    G_{Z,{\rm R}} &= (\, \mathbf{0} \;|\; \mathbf{e}_1 \otimes g \,)  \\
    G_{X,{\rm R}} &= (\, \mathbf{0} \;|\; \mathbf{1} \otimes E_1 \,) \, ,
\end{align}
\end{subequations}
where $g$ is given by \eqref{eq:K_4 generator matrix}, $\mathbf{e}_1 = (1,0,0,0) \notin \ker h^\transpose$, $\mathbf{1} = (1,1,1,1)$ is the repetition codeword, and $E_1 \in \mathbb{F}^{3\times 6}_2$ is a matrix whose rows are not in $\ker h$. From the above logical basis, we notice that the first left automorphism gadget $U_{1,{\rm L}} = (\sigma \otimes \ident) \oplus (\sigma' \otimes \ident)$ (here $\sigma'\equiv w$) acts nontrivially on the 3 left logical qubits and trivially on the 3 right logical qubits. Likewise, the second right automorphism gadget $U_{2,{\rm R}} = (\ident \otimes \sigma') \oplus (\ident \otimes \sigma)$ acts nontrivially on the right logical qubits and trivially on the left logical qubits. Note that for this example, $U_{2,{\rm L}}$ and $U_{2,{\rm R}}$ are redundant, similarly for $U_{1,{\rm L}}$ and $U_{1,{\rm R}}$. Thus, the group of logical automorphisms generated from $U_{1,{\rm L}}$ and $U_{2,{\rm R}}$ is isomorphic to $\mathrm{S}_4 \times \mathrm{S}_4$. The actions on the left and right logical qubits are identical to those in the classical $K_4$ cycle codes: all logical permutations as well as the multitargeted CNOT \eqref{eq:V_K4 CNOT}. Importantly, the `$\times$' here represents independent operations between the left and right sectors, whereas for the first example everything was performed solely on the left sector.

\begin{figure}[t]
    \centering
    \includegraphics[width=0.99\textwidth]{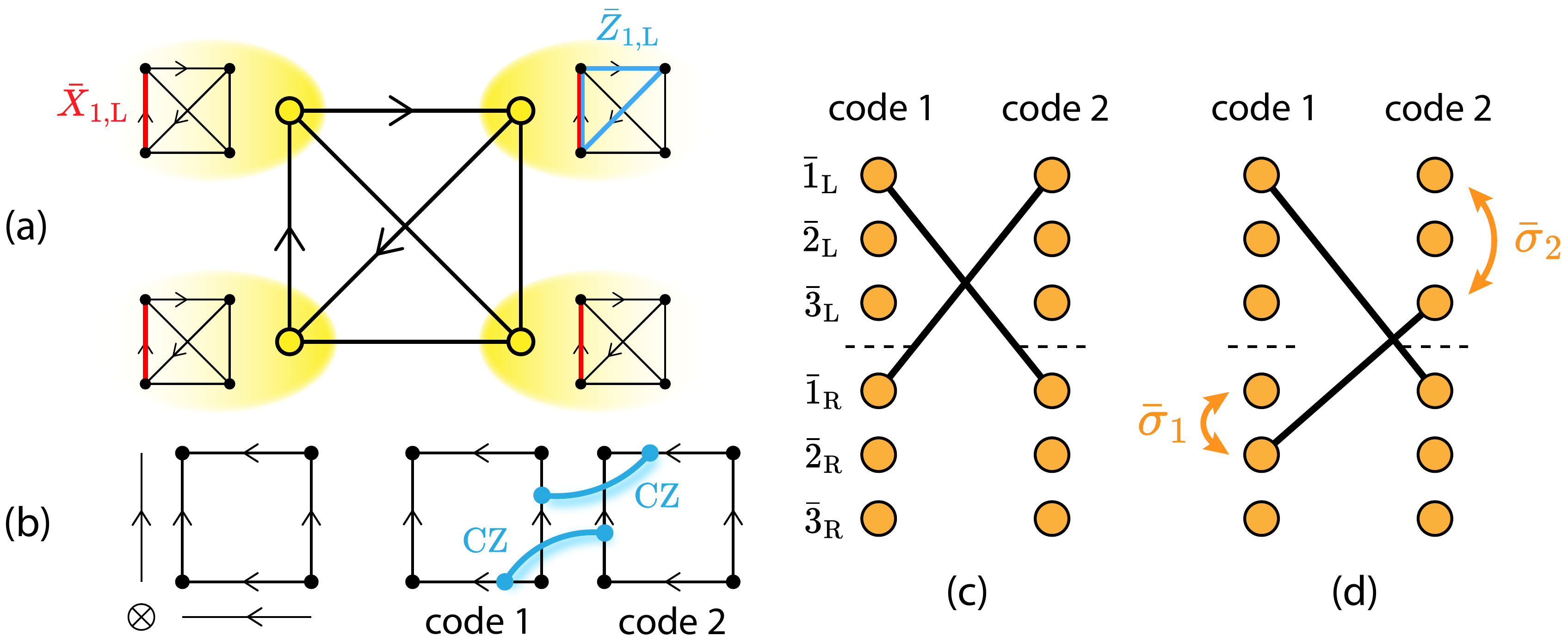}
    \caption{\textbf{(a)} One can view the $K^2_4$ graph as putting ``inner'' $K_4$ graphs, corresponding to one input $K_4$, on the vertices of an ``outer'' $K_4$, the other input, and connecting different inner graphs with transversal edges according to the edges of the outer graph. Logical operators for the first left logical qubit are drawn. An orientation satisfying the Leibniz condition along the support of (classical) codeword 1 is drawn on both inner and outer graphs. \textbf{(c)} Sketch of the copy-cup CZ gate on the data qubits (edges) resulting from a nonzero cup product. \textbf{(c)} The logical action of the copy-cup CZ gate is depicted. Orange circles depict logical qubits, and thick lines between logical qubits denote logical CZ gates. \textbf{(d)} With access to automorphism gadgets, the logical CZ gates from (c) can now address different logical qubits.}
    \label{fig:K4xK4 logicals}
\end{figure}

There are several reasons to choose the transpose code as the second input code. Recall that the $K_4$ cycle code defines bits on edges and parity checks on the vertices of $K_4$. When we take the HGP of the $K_4$ cycle code with its transpose, the resulting HGP code can also be defined according to a graph: the Cartesian graph product $K_4 \times K_4 = K^2_4$. Qubits are located on the edges, $X$-checks on the vertices, and $Z$-checks on the square ``faces'' of the product graph $K^2_4$ spanned by an edge in one copy of $K_4$ and an edge in the second copy. As a consequence, the $X$-syndromes for $Z$ errors can be decoded with minimum-weight matching, similar to that in the surface code; $Z$-syndromes however cannot be decoded in this manner and require other means. Topologically, a graph can be interpreted as a simplicial complex with vertices as 0-simplices and edges as 1-simplices. As such, the HGP product complex is imbued with transversal CZ gates stemming from a cohomology operation known as a cup product \cite{golowich2024non, breuckmann2024cups}. We defer the interested reader to the paper by Breuckmann \emph{et al.} \cite{breuckmann2024cups}, and we will simply recite their relevant results for our purposes. The first step to defining the ``copy-cup'' CZ gate is to assign directionality to the edges, i.e. an orientation, of the input $K_4$ graphs; we allow for edges to be ``free'' and have no direction. For a vertex $v \in V$, let $\delta(v) = \delta_{\rm in}(v) \cup \delta_{\rm out}(v) \cup \delta_{\rm free}(v) \in E$ denote the edges incident to $v$. The orientation needs to satisfy the ``Leibniz condition'': $\abs{\delta_{\rm in}(v)} + \abs{\delta_{\rm out}(v)} = 0$ (mod 2) for all $v \in V$, in other words the parity of incoming and outgoing edges must be the same at every vertex. Upon choosing appropriate orientations for both our input $K_4$ graphs, we also induce an orientation on their Cartesian product $K^2_4$: whenever we copy an edge, we simply also copy its corresponding orientation; see Figure \ref{fig:K4xK4 logicals}a for a sketch. The cup product is nonzero between any two directed edges on $K^2_4$ that satisfy
\begin{enumerate}
    \item The two directed edges do not originate from the same input graph. Roughly, one edge stems from one input graph and the other edge from the other input graph.
    \item The endpoint vertex of the first edge is the starting point of the second edge. In other words, the edges are consecutively oriented.
\end{enumerate}
The copy-cup CZ gate between two code blocks is then constructed as follows. Whenever two edges satisfy the above conditions for a nonzero cup product, we apply a CZ gate between the first edge (according to the consecutive orientation) on the first code block and the second edge on the second code block; see Figure \ref{fig:K4xK4 logicals}b for a sketch.

When we choose the orientations on the input $K_4$ graphs to follow our classical codewords (Figures \ref{fig:K4 codewords} and \ref{fig:K4xK4 logicals}a), the copy-cup CZ gate results in logical $\overline{\rm CZ}$ gates between logical qubits of different sectors; see Figure \ref{fig:K4xK4 logicals}c. For example in Figure \ref{fig:K4xK4 logicals}a, if we choose both input orientations to follow the support of logical bit 1 in the classical $K_4$ cycle code, then the copy-cup CZ gate performs logical $\overline{\rm CZ}$ between logical qubits $\bar{1}_{\rm L}$ and $\bar{1}_{\rm R}$ of two code blocks. Choosing the input orientations to follow different classical codewords results in $\overline{\rm CZ}$s that involve other logical qubits. Recall that we can perform arbitrary intrasector logical permutations using the automorphism gadgets lifted from the $K_4$ cycle code. Using these logical permutations, we can have the copy-cup CZ gates address different logical qubits, like in Figure \ref{fig:K4xK4 logicals}d. This degree of addressability is greater than what we would simply get from the freedom of choosing the initial orientations.

\subsection{$\llbracket 288,9,(16,3) \rrbracket$ homological product code with $\mathcal{A} \simeq \mathrm{S}_4\times\mathrm{S}_4\times\mathrm{S}_4$ and transversal CCZ gates}

We only briefly discuss our last example, since it is a simple extension of the previous one in Section \ref{sec:K4 HGP copy-cup}. We take a (quantum $\times$ classical) homological product \eqref{eq:3D HGP H_X,H_Z} of the $\llbracket 48,6,3 \rrbracket$ HGP code from the previous section with the transpose $K_4$ cycle code to obtain a $\llbracket 288,9,(d_X=16,d_Z=3) \rrbracket$ homological product code. Similar to the previous example, the structure of this code can be defined with respect to a graph, this time the triple Cartesian product of three $K_4$ graphs. $X$-checks are defined on the vertices, data qubits on the edges, and $Z$-checks on the faces (again spanned by edges originating from distinct $K_4$ copies). $X$-checks have weight 9, and each qubit participates in 2 of them. $Z$-checks have weight 4, and each qubit participates in 6 of them. The 9 logical qubits are partitioned into three sectors: left, middle and right, with 3 logical qubits belonging to each sector. We also have three types of automorphism gadgets \eqref{eq:3D HGP gadgets}: one from each input $K_4$ cycle code. Similar to the example in the previous section, the automorphism gadgets include arbitrary logical permutations within each sector.

The cup product formalism \cite{breuckmann2024cups} also applies to this code, and with an appropriate choice of orientations on the three input $K_4$ graphs, we can obtain transversal CCZ gates on our homological product code. Similar to Figure \ref{fig:K4xK4 logicals}b, the three input orientations will induce an orientation on the product graph. For the copy-cup CCZ gate between three code blocks, we apply a CCZ gate between every three consecutive edges that originate from different input $K_4$ graphs: the first edge on the first code block, the second edge on the second code block, and the third edge on the third code block. On the logical level, the copy-cup CCZ gate will perform logical $\overline{\rm CCZ}$ gates on logical qubits residing in different sectors and in different code blocks: i.e. logical qubit $\bar{1}_{\rm L}$ on the first code block, $\bar{1}_{\rm M}$ on the second code block, and $\bar{1}_{\rm R}$ on the third code block, as well as their interblock permuted triples. Leveraging logical permutations from the automorphism gadgets, we can increase the total number of logical qubit triples that the copy-cup $\overline{\rm CCZ}$s act on and hence increase the addressability of the copy-cup gates.


\section*{Acknowledgments}
We thank Victor Albert, Ali Fahimniya, Andrew Lucas and Shayan Majidy for helpful discussions.
This material is based upon work supported in part by the
Defense Advanced Research Projects Agency (DARPA)
under Agreement HR00112490357, the NSF QLCI award OMA2120757 and the DoE ASCR Quantum Testbed
Pathfinder program (awards No.~DE-SC0019040 and No.~DE-SC0024220).
SJST acknowledges funding and support from Joint Center for Quantum Information and Computer Science (QuICS) Lanczos Graduate Fellowship, MathQuantum Graduate Fellowship, and the National University of Singapore (NUS) Development Grant. 
\bibliographystyle{alpha}
\renewcommand*{\bibfont}{\small}
\bibliography{thebib}
\end{document}